\documentclass{llncs}

 \usepackage{times}
 \usepackage{srcltx}
 \usepackage{braket}
 \usepackage{tikz}
 \usepackage{verbatim}
 \usepackage{longtable}
 \usepackage{color}
 \usepackage{stmaryrd} 
 \usepackage{xspace}
 \usepackage{multirow}

\usepackage{amsmath}
\usepackage{latexsym}
\usepackage{textcomp}

\usepackage{bm}

\usepackage{tikz}
\usetikzlibrary{calc,trees,positioning,arrows,chains,shapes.geometric,%
    decorations.pathreplacing,decorations.pathmorphing,shapes,%
    matrix,shapes.symbols}
   
\usepackage{comment}
\usepackage{todonotes} 
\usepackage{verbatim}

\usepackage{latexsym,amsmath,amssymb}

\usepackage[british]{babel}
\usepackage[latin1]{inputenc}
\usepackage{xcolor}

\newcommand{\myparagraph}[1]{\vspace{5pt}\noindent\textbf{#1}}
\newcommand{\myequation}[1]{\begin{equation}
\setlength{\abovedisplayskip}{0.3\abovedisplayskip}
\setlength{\belowdisplayskip}{0.3\belowdisplayskip}
#1
\end{equation}}

\usepackage{mdwlist}

\usepackage{url}
\urldef{\mailsa}\path|{cdpu, riis, nielson}@imm.dtu.dk|

\newcommand{\is}{::=}
\newcommand{\Request}{{Request}\xspace}
\newcommand{\PS}{{PolicySet}\xspace}
\newcommand{\Policy}{{Policy}\xspace}
\newcommand{\Rule}{{Rule}\xspace}

\newcommand{\Target}{{Target}\xspace}
\newcommand{\Condition}{{Condition}\xspace}
\newcommand{\Match}{{Match}\xspace}
\newcommand{\AnyOf}{{AnyOf}\xspace}
\newcommand{\AllOf}{{AllOf}\xspace}
\newcommand{\Attribute}{{Attribute}\xspace}
\newcommand{\cat}{\ensuremath{\mathit{cat}}}

\newcommand{\XACML}{\hbox{\tiny{XACML}}}
\newcommand{\PXACML}{\ensuremath{\Prog_{\!\XACML}}\xspace}
\newcommand{\Pgenerator}{\ensuremath{\Prog_{\!\mathit{generator}}}\xspace}

\newcommand{\mcPS}{\ensuremath{\mc{P\!S}}}
\newcommand{\mcPSid}[1][\!\ID]{\ensuremath{\mcPS_{#1}}}
\newcommand{\mcP}{\ensuremath{\mc{P}}}
\newcommand{\mcPid}[1][\!\ID]{\ensuremath{\mcP_{#1}}}
\newcommand{\mcR}{\ensuremath{\mc{R}}}
\newcommand{\mcRid}[1][\!\ID]{\ensuremath{\mcR_{#1}}}
\newcommand{\mcC}{\ensuremath{\mc{C}}\xspace}
\newcommand{\mcT}{\ensuremath{\mc{T}}\xspace}
\newcommand{\mcE}{\ensuremath{\mc{E}}\xspace}
\newcommand{\mcA}{\ensuremath{\mc{A}}\xspace}
\newcommand{\mcM}{\ensuremath{\mc{M}}\xspace}
\newcommand{\mcQ}{\ensuremath{\mc{Q}}\xspace}
\newcommand{\mcAt}{\ensuremath{\mc{A}ttr}}

\newcommand{\eval}{\ensuremath{\mathsf{eval}}\xspace}
\newcommand{\seq}[1]{\ensuremath{\langle #1 \rangle}}

\newcommand{\la}{\leftarrow}
\newcommand{\ra}{\rightarrow}

\newcommand{\ground}[1]{\ensuremath{\mathit{ground}(#1)}}
\newcommand{\Body}{Body}
\newcommand{\mc}[1]{\ensuremath{\mathcal{#1}}}
\newcommand{\mb}[1]{\ensuremath{\mathbf{#1}}}

\newcommand{\prog}[1]{\[ {\footnotesize \begin{array}{ll} #1 \end{array} } \] }
\newcommand{\Prog}{\ensuremath{\Pi}\xspace}

\newcommand{\val}{\ensuremath{\mathsf{val}}}

\newcommand{\m}{\ensuremath{\mathsf{m}}}
\newcommand{\nm}{\ensuremath{\mathsf{nm}}}
\newcommand{\idt}{\ensuremath{\mathsf{idt}}}
\newcommand{\error}{\ensuremath{\mathsf{error}}}
\renewcommand{\t}{\ensuremath{\mathsf{t}}}
\newcommand{\f}{\ensuremath{\mathsf{f}}}
\newcommand{\p}{\ensuremath{\mathsf{p}}\xspace}
\renewcommand{\d}{\ensuremath{\mathsf{d}}\xspace}
\newcommand{\na}{\ensuremath{\mathsf{na}}\xspace}
\newcommand{\dec}{\ensuremath{\mathsf{decision\_of}}}
\newcommand{\comb}{\ensuremath{\mathsf{CombID}}}
\newcommand{\id}{\ensuremath{\mathsf{i_d}}\xspace}
\newcommand{\ip}{\ensuremath{\mathsf{i_p}}\xspace}
\newcommand{\idp}{\ensuremath{\mathsf{i_{dp}}}\xspace}
\newcommand{\algo}{\ensuremath{\mathsf{algo}}}
\newcommand{\po}{\ensuremath{\mathsf{po}}}
\newcommand{\denyo}{\ensuremath{\mathsf{do}}}
\newcommand{\ooa}{\ensuremath{\mathsf{ooa}}}
\newcommand{\fa}{\ensuremath{\mathsf{fa}}}
\newcommand{\mynot}{\ensuremath{\mathbf{not}~}}

\newcommand{\cond}{\ensuremath{\mathit{cond}}}

\newcommand{\gap}{\ensuremath{\mathit{gap}}}

\newcommand{\semantics}[1]{\ensuremath{\llbracket #1\rrbracket}\xspace}
\newcommand{\semanticsM}[1]{\ensuremath{\semantics{#1}\xspace}}
\newcommand{\semanticsC}[1]{\ensuremath{\semantics{#1}\xspace}}
\newcommand{\semanticsR}[1]{\ensuremath{\semantics{#1}\xspace}}
\newcommand{\semanticsP}[1]{\ensuremath{\semantics{#1}\xspace}}
\newcommand{\mynote}[1]{{#1}}

\DeclareMathAlphabet\mathbfcal{OMS}{cmsy}{b}{n}

\begin{document}


\title{XACML 3.0 in Answer Set Programming -- \\
Extended Version
}
%
\titlerunning{XACML 3.0 in Answer Set Programming}

\author{Carroline Dewi Puspa Kencana Ramli, Hanne Riis Nielson, Flemming Nielson}
\authorrunning{Carroline Dewi Puspa Kencana Ramli, Hanne Riis Nielson, Flemming Nielson}

\institute{Department of Informatics and Mathematical Modelling \\
Danmarks Tekniske Universitet \\
Lyngby, Denmark\\
\email{\mailsa}}

\maketitle
\vspace{-15pt}
\begin{abstract}
We present a systematic technique for transforming XACML 3.0 policies in Answer Set Programming (ASP). We show that the resulting logic program has a unique answer set that directly corresponds to our formalisation of the standard semantics of XACML 3.0 from \cite{Ramli2011}. We demonstrate how our results make it possible to use off-the-shelf ASP solvers to formally verify properties of access control policies represented in XACML, such as checking the completeness of a set of access control policies and verifying policy properties. 

\keywords{XACML, access control, policy language, Answer Set Programming}
\end{abstract}
\vspace{-30pt}
\section{Background}
XACML (eXtensible Access Control Markup Language) is a prominent access control language that is widely adopted both in industry and academia. XACML  is an international standard in the field of information security and in \mynote{February 2005,  XACML version 3.0 was ratified by OASIS.}\footnote{The Organization for the Advancement of Structured Information Standards (OASIS) is a global consortium that drives the development, convergence, and adoption of e-business and web service standards.} XACML  represents a shift from a more static security approach as exemplified by ACLs (Access Control Lists) towards a dynamic approach, based on Attribute Based Access Control (ABAC) systems. These dynamic security concepts are more difficult to understand, audit and interpret in  real-world implications. The use of XACML requires not only the right tools but also well-founded concepts for policy creation and management.

The problem with XACML is that its specification is described in natural language (c.f.\ \cite{XACML3.0}) and manual analysis of the overall effect and consequences of a large XACML policy set is a very daunting and time-consuming task. How can a policy developer be certain that the represented policies capture all possible requests? Can they lead to conflicting decisions for some request? Do the policies satisfy all required properties? These complex problems cannot be solved easily without some automatised support. 

To address this problem we propose a logic-based XACML analysis framework using Answer Set Programming (ASP). With ASP we model an XACML Policy Decision Point (PDP) that loads XACML policies and evaluates XACML requests against these policies. The expressivity of ASP and the existence of efficient implementations of the answer set semantics, such as \texttt{clasp}\footnote{\url{http://www.cs.uni-potsdam.de/clasp/}} and  \texttt{DLV}\footnote{\url{http://www.dlvsystem.com/}}, provide the means for declarative specification and verification of properties of XACML policies.

Our work is depicted in Figure~\ref{f:our work}. \mynote{There are two main modules, viz. the PDP simulation module and the access control (AC) security property verification module.} In the first module, we transform an XACML query  and XACML policies from the original format in XML syntax into abstract syntax which is more compact than the original. Subsequently we generate a query program $\Prog_{\mc{Q}}$ and XACML policies program $\PXACML$ that correspond to the XACML query and the XACML policies, respectively. \mynote{We show that the corresponding answer set \mynote{(AS)} of $\Prog_{\mcQ} \cup \PXACML$ is unique and it coincides with the semantics of original XACML policy evaluation.} In the second module, we demonstrate how our results make it possible to use off-the-shelf ASP solvers to formally verify properties of AC policies represented in XACML.  \mynote{First we encode the AC security property and a generator for each possible domain of XACML policies into logic programs $\Prog_{\!\mathit{AC\_property}}$ and $\Pgenerator$, respectively. The encoding of AC property is in the negated formula in order to show at a later stage that each answer set corresponds to a counter example that  violates the AC property. Together with the combination of $\PXACML \cup \Prog_{\!\mathit{AC\_property}} \cup \Pgenerator$ we show that the XACML policies satisfy the AC property when there is no available answer set.}
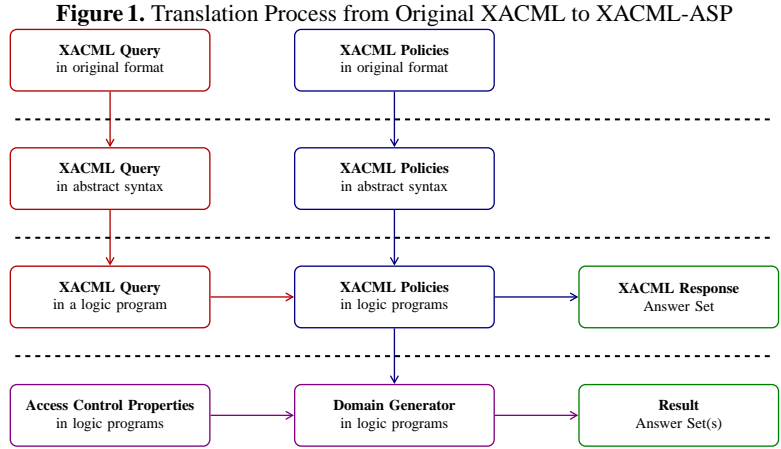
\vspace{-20pt}
\begin{figure}
\centering 
\label{f:our work}
\caption{Translation Process from Original XACML to XACML-ASP}
\scalebox{0.63}{
\begin{tikzpicture}
\tikzstyle{textnode}=[minimum height=1.3cm, text width=4cm, text centered, draw, rectangle, rounded corners]
\tikzstyle{query}=[draw=red!70!black, ->, >=angle 60, thick]
\tikzstyle{policies}=[draw=blue!50!black, ->, >=angle 60, thick]
\tikzstyle{answer}=[draw=green!50!black, ->, >=angle 60, thick]
\tikzstyle{property}=[draw=blue!50!red, ->, >=angle 60, thick]

\node[textnode, query] (qorig) at (0, 0) { \textbf{XACML Query} \\ in original format};
\node[textnode, query] (qabstr) at (0, -2.5) { \textbf{XACML Query} \\ in abstract syntax};
\node[textnode, query] (qlp) at (0, -5) { \textbf{XACML Query} \\ in a logic program};
\draw[query] (qorig) -- (qabstr);
\draw[query] (qabstr) -- (qlp);

\node[textnode, policies] (porig) at (6, 0) { \textbf{XACML Policies} \\ in original format};
\node[textnode, policies] (pabstr) at (6, -2.5) { \textbf{XACML Policies} \\ in abstract syntax};
\node[textnode, policies] (plp) at (6, -5) { \textbf{XACML Policies} \\ in  logic programs};
\draw[policies] (porig) -- (pabstr);
\draw[policies] (pabstr) -- (plp);

\node[textnode, answer] (response) at (12, -5) { \textbf{XACML Response} \\ Answer Set};
\draw[query] (qlp) -- (plp);
\draw[policies] (plp) -- (response);

\node[textnode, property] (acprop) at (0, -7.5) {\textbf{Access Control Properties} \\ in logic programs};
\node[textnode, property] (domgen) at (6, -7.5) {\textbf{Domain Generator} \\ in logic programs};
\node[textnode, answer] (result) at ( 12, -7.5) {\textbf{Result} \\ Answer Set(s)};
\draw[policies] (plp) -- (domgen);
\draw[property] (acprop) -- (domgen);
\draw[property] (domgen) -- (result);

\draw[dashed, very thick, black] (-2, -1.25) -- (14.0, -1.25);
\draw[dashed, very thick, black] (-2, -3.75) -- (14.0, -3.75);
\draw[dashed, very thick, black] (-2, -6.25) -- (14.0, -6.25);
\end{tikzpicture} 
}
\end{figure}

\noindent \textit{Outline.} \mynote{We consider the current version, XACML 3.0, 
Committee Specification 01, 10 August 2010. }
in Section~\ref{s:xacml} we explain the abstract syntax and semantics of XACML 3.0. 
Then we describe the transformation of XACML 3.0 components into logic programs in Section~\ref{s:transformation}. We show the relation between XACML 3.0 semantics and the answer sets in Section~\ref{s:xacml-asp}. Next, in Section~\ref{s:analysis}, we show how to verify AC properties, such as checking the completeness of a set of policies. In Section~\ref{s:related work} we discuss the related work. We end the paper with conclusions and future work. 

\section{XACML 3.0}
\label{s:xacml}
In order to avoid superfluous syntax of XACML 3.0, first we present the abstract syntax of XACML 3.0 which only shows the important components of XACML 3.0. We continue the explanation by presenting the semantics of XACML 3.0 components' evaluation based on Committee Specification \cite{XACML3.0}.  We take the work of Ramli \textit{et. al} work \cite{Ramli2011} as our reference. 

\subsection{Abstract Syntax of XACML 3.0}
Table \ref{t:syntax} shows the abstract syntax of XACML 3.0. We use bold font for non-terminal \textbf{symbols}, typewriter font for terminal \texttt{symbols} and \textit{identifiers} and \textit{values} are written in italic font. 
A symbol followed by the star symbol ($^*$) indicates that there are zero or more occurrences of that symbol. Similarly, a symbol followed by the plus symbol ($^+$) indicates that there are one or more occurrences of that symbol. We consider that each policy has a unique identifier (ID).  We use initial capital letter for XACML components such as \PS, \Policy, \Rule, etc., and  small letters  for English terminology. 

\vspace{-10pt}
\begin{table*}[ht!]
\caption{Abstraction of XACML 3.0 Components}
\label{t:syntax}
\centering
\begin{tabular}{|l|lcl|}
\hline  \vspace{-7pt} & & & \\ 
&\multicolumn{3}{|c|}{\textbf{\underline{XACML Policy Components}}} \\
\vspace{-7pt}& & &\\
\PS           & $\bm{\mcPS}$ & \is & $\mcPSid = {[} \bm{\mathcal{T}}, \seq{(\mcPSid \mid \mcPid)^*}, \textbf{CombID} {]}$\\
\Policy      & $\bm{\mcP}$  & \is &  $\mcPid = {[} \bm{\mathcal{T}}, \seq{\mcRid~^{\!\!+}}, \textbf{CombID} {]}$\\
\Rule         & $\bm{\mc{R}}$  & \is &  $\mcRid = {[} \textbf{Effect}, \bm{\mc{T}}, \bm{\mc{C} } {]}$ \\
\Condition  & $\bm{\mc{C}}$  & \is & \texttt{true} \textbar\  $f^{\mathsf{bool}}(a_1, \dotsc, a_n)$\\
\Target      & $\bm{\mc{T}}$   & \is & $\texttt{null} \mid \bigwedge \bm{\mc{E}}^+$ \\
\AnyOf      & $\bm{\mc{E}}$  & \is &  $\bigvee \bm{\mc{A}}^+$ \\
\AllOf        & $\bm{\mc{A}}$ & \is &  $\bigwedge \bm{\mc{M}}^+$ \\
\Match      & $\bm{\mc{M}}$ & \is &  $\bm{\mcAt}$\\
& \textbf{CombID}  & \is & \texttt{po} \textbar\ \texttt{do} \textbar\ \texttt{fa} \textbar\ \texttt{ooa} \\
& \textbf{Effect} & \is & \texttt{p} \textbar\ \texttt{d} \\ 
\Attribute& $\bm{\mcAt}$ & \is & $\mathit{category(attribute\_value)}$ \\
\hline 
\vspace{-7pt} & & & \\ 
& \multicolumn{3}{|c|}{\textbf{\underline{XACML Request Component}}}\\ 
\vspace{-7pt}& & &\\
\Request & $\bm{\mc{Q}}$ & \is & $( \bm{\mcAt} \mid \error(\bm{\mcAt}) )^+ $\\
\hline
   \end{tabular}
\end{table*}

There are three levels of policies in XACML, namely \PS, \Policy and \Rule. \PS or \Policy can act as the root of a set of access control policies, while \Rule is a single entity that describes one particular access control policy. Throughout this paper we consider that \PS is the root of the set of access control policies.

Both \PS and \Policy function as containers for a sequence of  \PS, \Policy or \Rule. A \PS  contains either a sequence of  \PS  elements or a sequence of \Policy elements, while a \Policy can only contain a sequence of \Rule elements. Every sequence of \PS, \Policy or \Rule elements has an associated \emph{combining algorithm}. There are four common combining algorithms defined in XACML 3.0, namely \emph{permit-overrides} ($\po$), \emph{deny-overrides} ($\denyo$), \emph{first-applicable} ($\fa$) and \emph{only-one-applicable} ($\ooa$). 

A  \Rule describes an individual access control policy. It regulates whether an access should be \textit{permitted} (\p) or \textit{denied} (\d). All \PS, \Policy and \Rule are applicable whenever their \Target matches with the \Request. When the \Rule's \Target matches the \Request, then the applicability of the \Rule is refined by its \Condition. 

A \Target element identifies the set of decision requests that the parent element is intended to evaluate. The \Target element must appear as a child of a \PS and \Policy element and may appear as a child of a \Rule element. The empty \Target for \Rule element is indicated by \texttt{null} attribute. The \Target element contains a conjunctive sequence of \AnyOf elements. The \AnyOf element contains a disjunctive sequence of \AllOf elements, while the \AllOf element contains a conjunctive sequence of \Match elements. \mynote{Each \Match element specifies an attribute  that a \Request should match.}

A \Condition is a Boolean function over attributes or functions of attributes. In this abstraction, the user is free to define the \Condition as long as its expression returns  a Boolean value, i.e., either true or false. Empty \Condition is always associated to true. 

A \Request contains a set of attribute values for a particular access request and the error messages that occurred during the evaluation of attribute values. 

\subsection{XACML 3.0 Formal Semantics}
\label{ss:formal semantics}

The evaluation of XACML policies starts from the evaluation of \Match elements and continues bottom-up until the evaluation of the root of the XACML element, i.e., the evaluation of \PS. 
\mynote{For each XACML element $X$ we denote by $\semantics{X}$ a semantic function associated to $X$. To each \Request element, this function assigns a value from a set of values that depends on the particular type of the XACML element $X$. 
For example, the semantic function $\semantics{X}$, where $X$ is a \Match element, ranges over the set $\set{\m, \nm, \idt}$, while its range is the set $\set{\t, \f, \idt}$ when $X$ is a \Condition element. A further explanation will be given below. }
An XACML component returns an indeterminate value whenever the decision cannot be made. This happens when there is an error during the evaluation process. See \cite{Ramli2011} for further explanation of the semantics of XACML~3.0. 

\myparagraph{Evaluation of \Match, \AllOf, \AnyOf and \Target Components.}
Let $X$ be either a \Match, an \AllOf, an \AnyOf or a \Target component and let $\mathbf{Q}$ be a set of all possible Requests.  A \textit{\Match semantic function} is a mapping
$ \semanticsM{X} : \mathbf{Q} \rightarrow \set{\m, \nm, \idt}$, where $\m, \nm$ and $\idt$ denote \textit{match}, \textit{no-match} and \textit{indeterminate}, respectively.  

Our evaluation of \Match element is based on equality function.\footnote{Our \Match evaluation is a simplification compared with \cite{XACML3.0}.} We check whether there are any attribute values in \Request element that match the \Match attribute value.


Let $\mc{Q}$ be a \Request element and let $\mc{M}$ be a \Match element. The evaluation of \Match $\mc{M}$ is as follows
\myequation{
\setlength{\abovedisplayskip}{0.3\abovedisplayskip}
\setlength{\belowdisplayskip}{0.3\belowdisplayskip}
\label{eq:match}
   \semanticsM{\mcM}(\mc{Q}) = 
   \begin{cases}
  \m   & \textrm{if }  \mcM \in \mc{Q} \textrm{ and } \error(\mcM) \notin \mc{Q} \\
  \nm & \textrm{if } \mcM \notin \mc{Q} \textrm{ and } \error(\mcM) \notin \mc{Q}\\
  \idt   & \textrm{if } \error(\mcM) \in \mc{Q}\\
    \end{cases}
}

The evaluation of \AllOf is a conjunction of a sequence of \Match elements. The value of \m, \nm\ and \idt\ corresponds to  true, false and undefined in 3-valued logic, respectively. 
   
Given a \Request $\mc{Q}$, the evaluation of \AllOf, $\mc{A} = \bigwedge_{i = 1}^n \mc{M}_i$, is as follows
\myequation{
\label{eq:allof}
   \semanticsM{\mc{A}}(Q) = 
   \begin{cases}
  \m   & \textrm{if } \forall i:   \semanticsM{\mc{M}_i}(Q) = \m\\
  \nm & \textrm{if } \exists i: \semanticsM{\mc{M}_i}(Q) = \nm \\
  \idt   & \textrm{otherwise }\\
    \end{cases}
}
where each $\mc{M}_i$ is a \Match element. 

The evaluation of \AnyOf element is a disjunction of a sequence of \AllOf elements. Given a \Request $\mc{Q}$, the evaluation of \AnyOf, $\mc{E} = \bigvee_{i = 1}^n \mc{A}_i$, is as follows
\myequation{
\label{eq:anyof}
   \semanticsM{\mc{E}}(\mc{Q}) = 
   \begin{cases}
  \m   & \textrm{if } \exists i: \semanticsM{\mc{A}_i}(\mc{Q}) = \m\\
  \nm & \textrm{if } \forall   i: \semanticsM{\mc{A}_i}(\mc{Q}) = \nm \\
  \idt   & \textrm{otherwise }\\
    \end{cases}
}
where each $\mc{A}_i$ is an \AllOf element. 

The evaluation of \Target element is a conjunction of a sequence of \AnyOf elements. An empty \Target, indicated by \texttt{null} attribute, is always evaluated to \m. Given a \Request $\mc{Q}$, the evaluation of \Target, $\mc{T} = \bigwedge_{i = 1}^n \mc{E}_i$, is as follows
\myequation{
\label{eq:target}
   \semanticsM{\mc{T}}(\mc{Q}) = 
   \begin{cases}
  \m   & \textrm{if } \forall i:   \semanticsM{\mc{E}_i}(Q) = \m \mbox{ or } \mc{T} = \texttt{null}\\
  \nm & \textrm{if } \exists i: \semanticsM{\mc{E}_i}(Q) = \nm \\
  \idt   & \textrm{otherwise }\\
    \end{cases}
}
where each $\mc{E}_i$ is an \AnyOf element. 

\myparagraph{Evaluation of \Condition.}
Let $X$ be a \Condition component and let $\mathbf{Q}$ be a set of all possible Requests.  A \textit{\Condition semantic function} is a mapping
$ \semanticsC{X} : \mathbf{Q} \rightarrow \set{\t, \f, \idt}$, where $\t, \f$ and $\idt$ denote \textit{true}, \textit{false} and \textit{indeterminate}, respectively. 

\mynote{The evaluation of \Condition element is based on the evaluation of  its Boolean function as described in its element. To keep it abstract, we do not specify specific functions; however, we use an unspecified function, \eval,  that returns $\Set{\t, \f, \idt}$.} 



Given a \Request $\mc{Q}$, the evaluation of \Condition $\mc{C}$ is as follows
\myequation{
\label{eq:condition}
   \semanticsC{\mc{C}}(\mc{Q}) =  \eval(\mc{C}, \mc{Q})
}

\myparagraph{Evaluation of \Rule.}
Let $X$ be a \Rule component and let $\mathbf{Q}$ be a set of possible Requests. A \textit{\Rule semantic function}  is a mapping 
$ \semanticsR{X} : \mcQ \rightarrow \Set{\p, \d, \ip, \id, \na}$,  where $\p, \d, \ip, \id$ and $\na$ correspond to \textit{permit}, \textit{deny}, \textit{indeterminate permit}, \textit{indeterminate deny} and $not-applicable$, respectively. 

Given a \Request $\mc{Q}$, the evaluation of \Rule $\mcRid = [E, \mc{T}, \mc{C}]$ is as follows
\myequation{
\label{eq:rule}
   \semanticsR{\mcRid}(\mc{Q}) = 
   \begin{cases}
  E & \textrm{if } \semanticsM{\mc{T}}(\mc{Q}) = \m \mbox{ and } \semanticsC{\mc{C}}(\mc{Q}) = \t\\
  \na& \textrm{if } (\semanticsM{\mc{T}}(\mc{Q}) = \m \mbox{ and } \semanticsC{\mc{C}}(\mc{Q}) = \f)  \mbox{ or } \semanticsM{\mc{T}}(\mc{Q}) = \nm \\
  \mathsf{i}_E& \textrm{otherwise}
    \end{cases}
}
where  $E$ is an effect, $E \in \Set{\p, \d}$,  $\mc{T}$ is a \Target element and $\mc{C}$ is a \Condition element.

\myparagraph{Evaluation of \Policy  and \PS.}
Let $X$ be either a \Policy or a \PS component and let $\mathbf{Q}$ be a set of all possible Requests. A \emph{\Policy semantic function} is a mapping
$ \semanticsP{X} : \mcQ \rightarrow \Set{\p, \d, \ip, \id, \idp, \na}$, where $\p, \d, \ip, \id, \idp$ and $\na$ correspond to \textit{permit}, \textit{deny}, \textit{indeterminate permit}, \textit{indeterminate deny}, \textit{indeterminate deny permit} and $not-applicable$, respectively. 

Given a \Request $\mc{Q}$, the evaluation of \Policy $\mcPid = [T, \seq{\mc{R}_1, \ldots, \mc{R}_n}, \comb]$ is as follows
\myequation{
\label{eq:policy}
   \semanticsP{\mcPid}(\mc{Q}) = 
   \begin{cases}
  \id   & \textrm{if } \semanticsM{\mc{T}}(\mc{Q}) = \idt \mbox{ and } \bigoplus_{\comb}(\mb{R})  = \d \\
  \ip   & \textrm{if } \semanticsM{\mc{T}}(\mc{Q}) = \idt \mbox{ and } \bigoplus_{\comb}(\mb{R}) = \p \\
  \na& \textrm{if } \semanticsM{T}(Q) = \nm \mbox{ or } \forall i: \semanticsR{R_i}(Q) = \na\\
  \bigoplus_{\comb}(\mb{R})& \mbox{otherwise}
    \end{cases}
}
where $\mc{T}$ is a \Target element, and each $\mc{R}_i$ is a \Rule element. We use $\mb{R}$ to denote $\seq{\semantics{\mc{R}_1}(\mc{Q}), \ldots, \semantics{\mc{R}_n}(\mc{Q})}$.

 \noindent \textit{Note:} The combining algorithm denoted by $\bigoplus_{\comb}$ will be explained in Sect.~\ref{ss:combining algorithms}.

The evaluation of \PS is exactly like the evaluation of \Policy except that it differs in terms of input parameter. While in \Policy we use a sequence of \Rule elements as an input, in the evaluation of \PS we use a sequence of \Policy or \PS elements. 

\subsection{XACML Combining Algorithms}
\label{ss:combining algorithms}
There are four common combining algorithms defined in XACML 3.0, namely permit-overrides (\po), deny-overrides (\denyo), first-applicable (\fa) and only-one-applicable (\ooa).  \mynote{In this paper, we do not consider the deny-overrides combining algorithm
since it is the mirror of the permit-overrides combining algorithm.}

\myparagraph{Permit-Overrides (\po) Combining Algorithm.}
The permit-overrides combining algorithm is intended for use if a permit decision should have priority over a deny decision. This algorithm has the following behaviour \cite{XACML3.0}. 
\begin{enumerate*}

\item If any decision is ``permit'', the result is ``permit''. 
\item Otherwise, if any decision is ``indeterminate deny permit'', the result is ``indeterminate deny permit''. 
\item Otherwise, if any decision is ``indeterminate permit'' and another decision is ``indeterminate deny'' or ``deny'', the result is ``indeterminate deny permit''. 
\item Otherwise, if any decision is ``indeterminate permit'', the result is ``indeterminate permit''.
\item Otherwise, if decision is ``deny'', the result is ``deny''. 
\item Otherwise, if any decision is ``indeterminate deny'', the result is ``indeterminate deny''. 
\item Otherwise, the result is ``not applicable''. 
\end{enumerate*}

Let $\seq{s_1, \ldots, s_n}$ be a sequence of element of $\Set{\p, \d, \ip, \id, \idp, \na}$. The \textit{permit-overrides combining operator} is defined as follows
\myequation{
\label{eq:po}
   \bigoplus_{po}(\seq{s_1, \ldots, s_n}) = 
   \begin{cases}
   \p & \textrm{if } \exists i: s_i = \p \\
   \idp & \textrm{if } \forall i: s_i \neq \p \textrm{ and } \\ 
   & \phantom{if }(\exists j: s_j = \idp  \\
   & \phantom{if (} \textrm{ or } ( \exists j, j': s_j= \ip \textrm{ and } (s_{j'} = \id \textrm{ or } s_{j'}= \d) )\\
   \ip & \textrm{if } \exists i: s_i= \ip \textrm{ and } \forall j: s_j \neq \ip \Rightarrow s_j = \na \\
   \d & \textrm{if } \exists i: s_i = \d \textrm{ and } \forall j: s_j \neq \d \Rightarrow ( s_j = \id \textrm{ or } s_j = \na )\\
  \id & \textrm{if } \exists i: s_i = \id \textrm{ and } \forall j: s_j \neq \id \Rightarrow s_j = \na \\ 
  \na & \textrm{otherwise}
   \end{cases}
}

\myparagraph{First-Applicable (\fa) Combining Algorithm.}
Each \Rule must be evaluated in the order in which it is listed in the \Policy. If a particular \Rule is applicable, then the result of first-applicable combining algorithm must  be the result of evaluating the \Rule. If the \Rule is ``not applicable'' then the next \Rule in the order must be evaluated. If no further \Rule in the order exists, then the first-applicable combining algorithm must return ``not applicable''. 

Let $\seq{s_1, \ldots,s_n}$ be a sequence of element of $\Set{\p, \d, \ip, \id, \idp, \na}$. The \textit{first-applicable combining operator} is defined as follows: 
\myequation{
\label{eq:fa}
   \bigoplus_{\fa}(\seq{s_1, \ldots, s_n}) = 
\begin{cases}
   s_i & \textrm{if }\exists i : s_i \neq \na \textrm{ and } \forall j : (j < i) \Rightarrow (s_j = \na)\\
   \na & \textrm{otherwise}
\end{cases}
}

\myparagraph{Only-One-Applicable (\ooa) Combining Algorithm.}
If only one \Policy is considered applicable by evaluation of its \Target, then the result of the only-one-applicable combining algorithm must the result of evaluating the \Policy.  If in the entire sequence of \Policy elements in the \PS, there is no \Policy that is applicable, then the result of the only-one-applicable combining algorithm must be ``not applicable''. If more than one \Policy is considered applicable, then the result of the only-one-applicable combining algorithm must be ``indeterminate''.  

Let $\seq{s_1, \ldots,s_n}$ be a sequence of element of $\Set{\p, \d, \ip, \id, \idp, \na}$. The \textit{only-one-applicable combining operator} is defined as follows:
\myequation{
\label{eq:ooa}
   \bigoplus_{\ooa}(\seq{s_1, \ldots, s_n}) = 
   \begin{cases}
   \idp & \textrm{if } (\exists i : s_i = \idp) \textrm{ or } \\
   & \phantom{if } (\exists i, j  : i \neq j \textrm{ and } s_i = ( \d \textrm{ or } \id ) \land s_j =  ( \p \textrm{ or } \ip ) )\\ 
   \id & \textrm{if } (\forall i: s_i \neq (\p \textrm{ or } \ip \textrm{ or } \idp)) \textrm{ and} \\
   & \phantom{if } ((\exists j: s_j = \id) \textrm { or } (\exists j, k: j \neq k \textrm{ and } s_j = s_k = \d) )\\
   \ip & \textrm{if }(\forall i: s_i \neq (\d \textrm{ or } \id \textrm{ or } \idp)) \textrm{ and} \\
   & \phantom{if } ((\exists j: s_j = \ip) \textrm { or } (\exists j, k: j \neq k \textrm{ and } s_j = s_k = \p) )\\
   s_i& \textrm{if } \exists i:  s_i \neq \na \textrm{ and } \forall j: j \neq i \Rightarrow s_j = \na \\
   \na & \textrm{otherwise}
   \end{cases}
}
\section{Transforming XACML Components into Logic Programs}
\label{s:transformation}
In this section we show, step by step, how to transform XACML~3.0 components into logic programs. We begin by introducing the syntax of logic programs (LPs). Then we show the transformation of XACML component into LPs starting from \Request element  to \PS element. We also present transformations for combining algorithms. The transformation of each XACML element is based on its formal semantics explained in Sect.~\ref{ss:formal semantics} and  Sect.~\ref{ss:combining algorithms}.

\subsection{Preliminaries}
We recall basic notation and terminology that we use in the remainder of this paper. 

\myparagraph{First-Order Language.}
 We consider an \emph{alphabet} consisting of (finite or countably infinite) disjoint sets of variables, constants, function symbols, predicate symbols, connectives $\Set{\mathbf{not}, \wedge, \la}$,  punctuation symbols $\Set{\hbox{``('', ``,'', ``)'', ``.''}}$ and special symbols $\Set{\top, \bot}$. We use upper case letters to denote variables and lower case letters to denote constants, function and predicate symbols. Terms, atoms, literals and formulae are defined as usual. The \emph{language} given by an alphabet  consists of the set of all formulae constructed from the symbols occurring in the alphabet.

\myparagraph{Logic Programs.}
 A \emph{rule}  is an expression of the form 
\myequation{\label{rule} A \la B_1 \wedge \dotsb \wedge B_m \wedge \mynot B_{m+1} \wedge \dotsb \wedge \mynot B_n. } 
where $A$ is either an atom or $\bot$ and each $B_i$, $1 \leq i \leq n$, is an atom or $\top$. $\top$ is a valid formula. We usually write $B_1 \wedge \dotsb \wedge B_m \wedge \mynot B_{m+1} \wedge \dotsb \wedge \mynot B_n$ simply as $B_1, \dotsc, B_m, \mynot B_{m+1}, \dotsc, \mynot B_n$.  We call  the rule as a \emph{constraint} when $A = \bot$. 
One should observe that the body of a rule must not be empty. A \emph{fact} is a rule of the form $A \la \top$. 

 
A \emph{logic program} is a finite set of rules. We denote \ground{\Prog} for the set of all ground instances of rules in the program \Prog.


\subsection{XACML  Components Transformation into Logic Programs}
The transformation of XACML components is based on the semantics of each component explained in Sect.~\ref{ss:formal semantics}.  

\subsubsection{\Request Transformation.} \textit{XACML Syntax}: Let $\mc{Q} = \Set{\mathit{cat}_1(a_1), \ldots, \mathit{cat}_n(a_n)}$ be a \Request component. We transform all members of \Request element into facts. The transformation of \Request, $\mc{Q}$, into LP $\Prog_{\mc{Q}}$ is as follows
\prog{\mathit{cat}_i(a_i) & \la \top. \ \ 1 \leq i \leq n}

\subsubsection{XACML Policy Components Transformation.}
\mynote{We use a two-place function $\val$ to indicate the semantics of XACML components where the first argument is the name of XACML component and the second argument is its value.} Please note that the calligraphic font in each transformation indicates the XACML component's name, that is, it does not represent a variable in LP. 

\myparagraph{Transformation of \Match, \AnyOf, \AllOf and \Target Components.}
Given a semantic equation of the form 
$\semantics{X}_{V}(\mcQ) = v \textrm{ if } \cond_1 \textrm{ and } \dotsc \textrm{ and } \cond_n$, 
we produce a rule of the form 
$\val(X, v)  \la \cond_1, \dotsc, \cond_n.$
Given a semantic equation of the form 
$\semantics{X}_{V}(\mcQ) = v \textrm{ if } \cond_1 \textrm{ or } \dotsc \textrm{ or } \cond_n$, we produce a rule of the form 
$\val(X, v)  \la \cond_i.\  1 \leq i \leq n$.
For example, the \Match evaluation $\semanticsM{\mcM}(\mcQ) = \m \textrm{ if }  \cat(a) \in \mcQ \textrm{ and } \error(\cat(a)) \notin \mcQ$ is transformed into a rule in the form $\val(\mcM, \m) \la \mcM,$ $\mynot \error(\mcM).$ \mynote{The truth value of $\mcM$ depends on whether $\mcM \la \top$ is in $\Prog_{\mcQ}$ and the same is the case also for the truth value of $\error(\mcM)$. }

Let $\mcM$ be a \Match component.  The transformation of \Match $\mcM$  into LP  $\Prog_{\mcM}$ is as follows (see \eqref{eq:match} for \Match evaluation)
\prog{
\val(\mcM, \m)  & \la \mcM, \mynot \error(\mcM). \\
\val(\mcM, \nm)& \la \mynot \cat(a), \mynot \error(\mcM). \\
\val(\mcM, \idt) & \la \error(\mcM).  
}

Let $\mcA =\bigwedge_{i=1}^n \mcM_i$ be an \AllOf component where each $\mcM_i$ is a \Match component. The transformation of \AllOf $\mc{A}$  into LP $\Prog_{\mcA}$ is as follows (see \eqref{eq:allof} for \AllOf evaluation)
\prog{
\val(\mcA, \m)   & \la \val(\mcM_1, \m), \dotsc, \val(\mcM_n, \m).  \\
\val(\mcA, \nm) & \la \val(\mcM_i, \nm).\  (1 \leq i \leq n) \\
\val(\mcA, \idt)  & \la \mynot \val(\mcA,\m), \mynot \val(\mcA, \nm). 
}

Let $\mcE = \bigvee_{i = 1}^{n} \mcA_i$ be an \AnyOf component where each $\mcA_i$ is an \AllOf component. The transformation of \AnyOf $\mcE$  into LP $\Prog_{\mc{E}}$ is as follows (see \eqref{eq:anyof} for \AnyOf evaluation)
\prog{
\val(\mcE, \m) & \la \val(\mcA_i, \m). \ (1 \leq i \leq n) \\
\val(\mcE, \nm)   & \la \val(\mcA_1, \nm), \dotsc, \val(\mcA_n, \nm).  \\
\val(\mcE, \idt)  & \la \mynot \val(\mcA,\m), \mynot \val(\mcE, \nm). 
}

Let $\mcT = \bigwedge_{i=1}^{n} \mcT_i$ be a \Target component where each $\mcE_i$ is an \AnyOf component. The transformation of \Target $\mcT$  into LP $\Prog_{\mcT}$ is as follows (see \eqref{eq:target} for \Target evaluation)
\prog{
\val(\mathsf{null}, \m) & \la \top. \\
\val(\mcT, \m)   & \la \val(\mcE_1, \m), \dotsc, \val(\mcE_n, \m).  \\
\val(\mcT, \nm) & \la \val(\mcE_i, \nm).\  (1 \leq i \leq n) \\
\val(\mcT, \idt)  & \la \mynot \val(\mcT,\m), \mynot \val(\mcT, \nm). 
}


\myparagraph{Transformation of \Condition Component.}
The transformation of \Condition $\mcC$ into LP $\Prog_{\mcC}$ is as follows
\prog{
\val(\mcC, V) & \la \eval(\mcC, V). 
}
Moreover, the transformation of \Condition also depends on the transformation of $\eval$ function into LP. Since we do not  describe specific $\eval$ functions, we leave this transformation to the user. 
\begin{example}
A possible  \eval function for "rule r1: patient only can see his or her patient record" is 
\prog{\Prog_{\cond(r1)}:\\
\val(\cond(r1), V) & \la  \eval(\cond(r1), V). \\
\eval(\cond(r1), \t) & \la patient\_id(X), patient\_record\_id(X), \\
& \phantom{\la{}} \mynot \error(patient\_id(X)), \mynot \error(patient\_record\_id(X)). \\
\eval(\cond(r1), \f) & \la patient\_id(X), patient\_record\_id(Y), X \neq Y, \\
& \phantom{\la{}} \mynot \error(patient\_id(X)), \mynot \error(patient\_record\_id(Y)). \\
\eval(\cond(r1), \idt) & \la \mynot \eval(\cond(r1), \t), \mynot \eval(\cond(r1), \f).
}
The $\error(patient\_id(X))$ and $\error(patient\_record\_id(X))$ indicate possible errors that might occur, e.g., the system could not connect to the database so that the system does not know the ID of the patient. \hfill $\Box$
\end{example}

\myparagraph{Transformation of \Rule Component.}
The general step of the transformation of \Rule component is similar to the transformation of \Match component. 

Let $\mcR = [e, \mcT, \mcC]$ be a \Rule component where $e \in \Set{\p, \d}$, $\mcT$ is a \Target and $\mcC$ is a \Condition.  The transformation of \Rule $\mcR$ into LP $\Prog_{\mcR}$ is as follows (see \eqref{eq:rule} for \Rule evaluation)
\prog{
\val(\mcR, e) & \la \val(\mcT, \m), \val(\mcC, \t). \\
\val(\mcR, \na) & \la \val(\mcT, \m), \val(\mcC, \f). \\
\val(\mcR, \na) & \la \val(\mcT, \nm). \\
\val(\mcR, \mathsf{i}_{e}) & \la \mynot \val(\mcR, e), \mynot \val(\mcR, \na).
}

\myparagraph{Transformation of \Policy and \PS Components.}
Given a \Policy component $\mcPid = [\mcT, \seq{\mcR_1, \ldots, \mcR_n}, \comb]$ where $\mcT$ is a \Target, $\seq{\mcR_1, \ldots, \mcR_n}$ is a sequence of \Rule elements and $\comb$ is a combining algorithm identifier. 
In order to indicate that the \Policy contains \Rule $\mcR_i$, for every \Rule $\mcR_i \in \seq{\mcR_1, \ldots, \mcR_n}$,  $\Prog_{\mcPid}$ contains:
\prog{
\dec(\mcPid, \mcR_i, V) & \la \val(\mcR_i, V). \ (1 \leq i \leq n)
}

The transformation for \Policy $\Prog$  into LP $\Prog_{\mcPid}$ is as follows (see \eqref{eq:policy} for \Policy evaluation)

\prog{
\val(\mcPid, \id) & \la \val(\mc{T}, \idt), \algo(\comb, \mcPid, \d). \\
\val(\mcPid, \ip) & \la \val(\mc{T}, \idt), \algo(\comb, \mcPid, \p). \\
\val(\mcPid, \na) & \la \val(\mc{T}, \nm). \\
\val(\mcPid, \na) & \la \val(\mc{R}_1, \na), \ldots, \val(\mc{R}_n, \na). \\
\val(\mcPid, V') & \la \val(\mc{T}, \m), \dec(\mcPid, \mc{R}, V), V \neq \na, \algo(\comb, \mcPid, V'). \\
\val(\mcPid, V') & \la \val(\mc{T}, \idt), \dec(\mcPid, \mc{R}, V), V \neq \na, \algo(\comb, \mcPid, V'),  V' \neq \p.\\
\val(\mcPid, V') & \la \val(\mc{T}, \idt), \dec(\mcPid, \mc{R}, V), V \neq \na, \algo(\comb, \mcPid, V'), V' \neq \d. \\
}
We write a formula $\dec(\mcPid, \mc{R}, V), V \neq \na$ to make sure that there is a \Rule in the \Policy that is not evaluated to \na. We do this to avoid a return value from a combining algorithm that is not \na, even tough all of the \Rule elements are evaluated to \na. The transformation of \PS is similar to the transformation of \Policy component.

\subsection{Combining Algorithm Transformation}
 
\mynote{We define  generic LPs for permit-overrides combining algorithm and only-one-applicable combining algorithm. Therefore, we use a variable  $P$ to indicate a variable over \Policy identifier and $R$, $R_1$ and $R_2$ to indicate  variables over \Rule identifiers.
In case the evaluation of \PS, the input $P$ is for \PS identifier, $R, R_1$ and $R_2$ are for \Policy (or \PS) identifiers. }

\myparagraph{Permit-Overrides Transformation.}
Let $\Prog_{\po}$ be a LP obtained by permit-overrides combining algorithm transformation (see \eqref{eq:po} for the permit-overrides combining algorithm semantics). $\Prog_{\po}$ contains: 
\prog{
\algo(\po, P, \p) & \la \dec(P, R, \p). \\
\algo(\po, P, \idp) & \la \mynot \algo(\po, P, \p), \dec(P, R, \idp). \\
\algo(\po, P, \idp) & \la \mynot \algo(\po, P, \p), \dec(P, R_1, \ip), \dec(P, R_2, \d). \\
\algo(\po, P, \idp) & \la \mynot \algo(\po, P, \p), \dec(P, R_1, \ip), \dec(P, R_2, \id). \\
\algo(\po, P, \ip) & \la \mynot \algo(\po, P, \p), \mynot \algo(\po, P, \idp), \dec(P, R, \ip).\\
\algo(\po, P, \d) & \la \mynot \algo(\po, P, \p), \mynot \algo(\po, P, \idp), \mynot \algo(\po, P, \ip), \\
& \phantom{\la{}} \dec(P, R, \d).  \\
\algo(\po, P, \id) & \la  \mynot \algo(\po, P, \p), \mynot \algo(\po, P, \idp),  \mynot \algo(\po, P, \ip),\\
& \phantom{\la{}} \mynot \algo(\po, P,\d),  \dec(P, R, \id).\\
\algo(\po, P, \na) & \la  \mynot \algo(\po, P, \p), \mynot \algo(\po, P, \idp), \mynot \algo(\po, P, \ip), \\
& \phantom{\la{}} \mynot \algo(\po, P,\d),  \mynot \algo(\po, P, \id).
}

\myparagraph{First-Applicable Transformation.}
Let $\Prog_{\fa}$ be a logic program obtained by first-applicable combining algorithm transformation (see \eqref{eq:fa} for the first-applicable  combining algorithm semantics). For each \Policy (or \PS) which uses this combining algorithm, $\mcPid = [\mcT, \seq{\mcR_1, \ldots, \mcR_n}, \fa]$,  $\Prog_{\mcPid}$ contains:
\prog{
\algo(\fa, \mcPid, E) & \la \dec(\mcPid, \mcR_1, V), V \neq \na. \\
\algo(\fa, \mcPid, E) & \la \dec(\mcPid, \mcR_1, \na), \dec(\mcPid, \mcR_2, E), E \neq \na.\\
& \vdots\\
\algo(\fa, \mcPid, E) & \la \dec(\mcPid, \mcR_1, \na), \ldots, \dec(\mcPid, \mcR_{n-1}, \na), \dec(\mcPid, R_n, E).\\
}

\myparagraph{Only-One-Applicable Transformation.}
Let $\Prog_{\ooa}$ be a logic program obtained by only-one-applicable combining algorithm transformation (see \eqref{eq:ooa} for the only-one-applicable combining algorithm semantics). $\Prog_{\ooa}$ contains:
{\footnotesize\prog{
\algo(\ooa, P, \idp) & \la \dec(P, R, \idp). \\
\algo(\ooa, P, \idp) & \la \dec(P, R_1, \id), \dec(P, R_2, \ip), R_1 \neq R_2.\\
\algo(\ooa, P, \idp) & \la \dec(P, R_1, \id), \dec(P, R_2, \p), R_1 \neq R_2.\\
\algo(\ooa, P, \idp) & \la \dec(P, R_1, \d), \dec(P, R_2, \ip), R_1 \neq R_2.\\
\algo(\ooa, P, \idp) & \la \dec(P, R_1, \d), \dec(P, R_2, \p), R_1 \neq R_2.\\
\algo(\ooa, P, \ip) & \la \mynot \algo(\ooa, P, \idp),  \dec(P, R, \ip). \\
\algo(\ooa, P, \ip) & \la \mynot \algo(\ooa, P, \idp),  \dec(P, R_1, \p), \dec(P, R_2, \p), R_1 \neq R_2.\\
\algo(\ooa, P, \id) & \la \mynot \algo(\ooa, P, \idp),  \dec(P, R, \id). \\
\algo(\ooa, P, \id) & \la \mynot \algo(\ooa, P, \idp),  \dec(P, R_1, \d), \dec(P, R_2, \d), R_1 \neq R_2. \\
\algo(\ooa, P, \p) & \la \mynot \algo(\ooa, P, \idp), \mynot(\ooa, P, \id), \mynot(\ooa, P, \ip),  \dec(P, R, \p).\\
\algo(\ooa, P, \d) & \la \mynot \algo(\ooa, P, \idp), \mynot(\ooa, P, \id), \mynot(\ooa, P, \ip),  \dec(P, R, \d). \\ 
\algo(\ooa, P, \na) & \la \mynot \algo(\ooa, P, \idp), \mynot(\ooa, P, \id), \mynot(\ooa, P, \ip), \\
& \phantom{\la{}}  \mynot \dec(P, R, \d), \mynot \dec(P, R, \p). \\
}
}

\section{Relation between XACML-ASP and XACML 3.0 Semantics}
\label{s:xacml-asp}
In this section we discuss the relationship between the ASP semantics and XACML 3.0 semantics. First, we recall the semantics of logic programs based on their answer sets. Then, \mynote{we show that the program obtained from transforming XACML components into LPs (\PXACML) merges with the query program ($\Prog_{\mcQ}$) and has a unique answer set that the answer set corresponds to the semantics of XACML 3.0.}

\subsection{ASP Semantics}
The declarative semantics of a logic program is given by a model-theoretic semantics  of formulae in the underlying language.  The formal definition of answer set semantics can be found in much literature such as \cite{Baral2003,Gelfond2007}.

The answer set semantics of logic program \Prog assigns to \Prog a collection of \emph{answer sets} --  interpretations  of \ground{\Prog}. An interpretation $I$ of \ground{\Prog} is an answer set for  \Prog if $I$ is minimal (w.r.t. set inclusion) among the  interpretations satisfying the rules of 
\[\begin{array}{ll}
\Prog^I = \{ A \la B_1, \ldots, B_m | & A \la B_1, \ldots, B_m, \mynot B_{m+1}, \ldots, \mynot B_{n} \in \Prog \mbox{ and} \\
 & I(\mynot B_{m+1}, \ldots, \mynot B_{n}) = \mathit{true} \} 
\end{array}
\]
A logic program can have a single unique answer set, many or no answer set(s). Therefore, we show that programs with a particular characteristic are guaranteed to have a unique answer set. 

\myparagraph{Acyclic Programs.} We say that a program is \emph{acyclic} when there is no cycle in the program.The acyclicity in the program  is guaranteed by the existence of a certain fixed assignment of natural numbers to atoms that is called a \emph{level mapping}.

A \emph{level mapping} for a program \Prog is a function 
\[l: \mc{B}_\Prog \ra \mathbf{N}\]
where $\mathbf{N}$ is the set of natural numbers and $\mc{B}_\Prog$ is the Herbrand base for \Prog.  We extend the definition of level mapping to a mapping from ground literals to natural numbers by setting $l(\mynot A) =  l(A)$. 

Let \Prog be a logic program and $l$ be a level mapping for \Prog. \Prog is \emph{acyclic with respect to l} if for every clause $A \la B_1, \ldots, B_m, \mynot B_{m+1}, \ldots, \mynot B_n$ in \ground{\Prog} we find  
\[
l(A) > l(B_i) \ \ \textrm{for all $i$ with $1 \leq i \leq n$}
\] \Prog is \emph{acyclic} if it is acyclic with respect to some degree of level mapping.  Acyclic programs are guaranteed to have a unique answer set \cite{Baral2003}.

\subsection{XACML Semantics Based On ASP Semantics}
We can see from Sect.~\ref{s:transformation} that all of the XACML 3.0 transformation programs are acyclic. Thus, it is guaranteed that  $\Prog_{\XACML}$ has a unique answer set. 

\begin{proposition}
\label{prop:xacml-asp}
Let $\Prog_{\XACML}$ be a program obtained from  XACML 3.0 element transformations and let $\Prog_{\mc{Q}}$ be a program transformation of \Request \mc{Q}.  Let $I$ be the answer set of $\Prog_{\XACML} \cup \Prog_{\mc{Q}}$. Then the following equation holds
\[ \semantics{X}(\mc{Q}) = V \mbox{ iff } \val(X,V) \in I\]
where $X$ is an XACML component. 
\end{proposition}
\textit{Note:} We can see that there is no cycle in all of the program transformations. Thus, there is a guarantee that the answer set of $\PXACML \cup \Prog_{\mc{Q}}$ is unique. 
The transformation of each component into a logic program is based on exactly the definition of its XACML evaluation.  The proof of this proposition can be seen in the extended version in \cite{Ramli2013}.

\section{Analysis XACML Policies Using Answer Set Programming}
\label{s:analysis}

In this section we show how to use ASP for analysing access control security  properties through \PXACML. In most cases, ASP solver can solve combinatorial problems efficiently. There are several combinatorial problems in analysis access control policies, e.g., gap-free property and conflict-free property \cite{Samarati2001,Bruns2008}. In this section we look at gap-free analysis since in XACML 3.0 conflicts never occur.\footnote{A conflict  decision never occurs when we strictly use the standard combining algorithm defined in XACML 3.0, since every combining algorithm always return one value.}  We also present a mechanism for the verification of security properties against a set of access control policies. 

\subsection{Query Generator}
In order to analyse access control property, sometimes we need to analyse all possible queries that might occur. We use \emph{cardinality constraint} (see \cite{Simons2002,Syrjanen}) to generate all possible values restored in the database for each attribute. For example, we have the following generator:
\prog{\mc{P}_{generator}:\\
(1)\ \  1 \{ subject(X) : subject\_db(X)\} 1 & \la \top.   \\
(2)\ \  1 \{ action(X) : action\_db(X) \} 1 & \la \top. \\
(3)\ \  1 \{ resource(X) : resource\_db(X) \} 1 &\la \top.  \\
(4)\ \  1 \{ environment(X) : environment\_db(X) \} 1 & \la \top. \\
}
The first line of the encoding means that we only consider one and only one $subject$ attribute value obtained from the subject database. The rest of the encoding means the same as the $subject$ attribute.  

\subsection{Gap-Free Analysis}
A set of policies is \emph{gap-free} if there is no access request for which there is an absence of decision. XACML  defines that there is  one \PS as the root of a set of policies. Hence, we say that there is a gap whenever we can find a request that makes the semantics of the $\mcPS_{\!\mathit{root}}$ is assigned to \na.  We force ASP solver to find the gap by the following encoding.
\prog{\Prog_{gap}:\\
\gap & \la \val(\mcPS_{\!\mathit{root}}, \na). \\
\bot & \la \mynot \gap. 
}
In order to make sure that a set of policies is gap-free we should generate all possible requests and test whether  at least one request  is not captured by the set of policies. Thus, the answer sets  of program $\mc{P} = \PXACML \cup \Pgenerator \cup \Prog_{\!\mathit{gap}}$ are witnesses that the set of policies encoded in $\PXACML $ is incomplete. When there is no model that satisfies the program then we are sure that the set of policies captures all of possible cases. 
 
\subsection{Property Analysis}
The problem of verifying a security property $\Phi$ on XACML policies is not only to show that the property $\Phi$ holds on $\PXACML$ but also that we want to see the witnesses whenever the property $\Phi$ does not hold in order to help the policy developer refine the policies. Thus, we can see this problem as finding models for $\PXACML\cup \Pgenerator \cup \Prog_{\neg \Phi}$. The founded model is the witness that the XACML policies cannot satisfy the property $\Phi$. 

\begin{example}
Suppose we have a security property:
\begin{center} $\Phi$: An anonymous person \textbf{cannot} read any patient records. \end{center}
Thus, the negation of property $\Phi$ is as follows
\begin{center}$\neg \Phi$: An anonymous person \textbf{can} read any patient records. \end{center}
We define that anonymous persons are those who are neither patients, nor guardians, nor doctors, nor nurses. We encode $\mc{P}_{\neg \Phi}$ as follows
\prog{
(1)\ anonymous & \la \mynot subject(patient), \mynot subject(guardian), \\
& \phantom{\la{}} \mynot subject(doctor), \mynot subject(nurse). \\
(2)\ \bot & \la \mynot anonymous. \\
(3)\ action(read) & \la \top. \\
(4)\ resource(patient\_record) & \la \top. \\
(5)\ \bot & \la \mynot \val(PS_{root}, \p).\\
}
We list all of the requirements (lines 1 -- 4). We force the program to find an anonymous person (line 2). Later we  force that the returned decision should be to permit (line 5).  When the program $\PXACML \cup \Pgenerator \cup \Prog_{\neg \Phi}$ returns models, we conclude that the property $\Phi$ does not hold and the returned models are  the flaws in the policies. On the other hand, we conclude that the property $\Phi$ is satisfied if no model is found.  
\end{example}
\vspace{-10pt}
\section{Related Work}
\label{s:related work}
\vspace{-5pt}
There are some approaches to defining AC policies in LPs, such as Barker \textit{et al. } in \cite{Barker2003} use constraint logic program to define role-based access control, Jajodia \textit{et al. } in \cite{Jajodia1997} using FAM / CAM program -- a logical language that uses  a fixed set of predicates. However, their approaches are based on their own access control policy language whereas our approach is to define a well-known access control policy language, XACML. 

Our approach is inspired by the work of Ahn \textit{et al.} \cite{Ahn2010,Ahn2010a}.
There are three main differences between our approach and the work of Ahn \textit{et al.} 

First, while they consider XACML version 2.0 \cite{xacml2.0}, we address the newer version, XACML 3.0.  The main difference between XACML 3.0 and XACML 2.0 is the treatment of indeterminate values.  As  a consequence, the combining algorithms in XACML 3.0 are more complex than the ones in XACML 2.0.  XACML 2.0 only has a single indeterminate value while XACML 3.0 distinguishes between the following three types of indeterminate values:
\begin{enumerate*}
	\renewcommand{\labelenumi}{\roman{enumi}.}
	\item \emph{Indeterminate permit} (\ip) -- an indeterminate value arising from a policy which could have been evaluated to permit but not deny;
	\item \emph{Indeterminate deny} (\id) -- an indeterminate value arising from a policy which could have been evaluated to deny but not permit;
	\item \emph{Indeterminate deny permit} (\idp) -- an indeterminate value arising from a policy which could have been evaluated as both deny and permit. 
\end{enumerate*}

Second, Ahn \emph{et al.} produce a monolithic logic program that can be used for the analysis of XACML policies while we take a more modular approach by first modelling an XACML PDP as a logic program and then using this encoding within a larger program for property analysis. While Ahn, \textit{et al.} only emphasize the indeterminate value in the combining algorithms, \mynote{our concern is ``indeterminate'' value in all aspect of XACML components, i.e., in \Match, \AnyOf, \AllOf, \Target, \Condition, \Rule, \Policy and \PS components. Hence, we show that our main concern is to simulate the PDP as in XACML model.} 

Finally, Ahn \emph{et al.} translate the XACML specification directly into logic programming, so the ambiguities in the natural language specification of XACML are also reflected in their encodings. To avoid this, we base our encodings on our formalisation of XACML from \cite{Ramli2011}.

\vspace{-10pt}
\section{Conclusion and Future Work}
\vspace{-5pt}
We have modelled the XACML Policy Decision Point in a declarative way using the ASP technique by transforming XACML 3.0 elements into logic programs.   Our transformation of XACML 3.0 elements is directly  based on XACML 3.0 semantics \cite{XACML3.0} and we have shown that the answer set of each program transformation is unique and that it agrees with the semantics of XACML 3.0. Moreover, we can help policy developers analyse their access control policies such as checking policies' completeness and verifying policy properties by inspecting the answer set of $\PXACML \cup \Pgenerator \cup \Prog_{\!\mathit{configuration}}$ -- the program obtained by transforming XACML 3.0 elements into logic programs joined with a query generator program and a configuration program. 


For future work, we can extend our work to handle role-based access control in XACML 3.0 \cite{xacml3.0rbac} and to handle delegation in XACML 3.0 \cite{xacml3.0delegation}. Also, we can extend our work for checking reachability of policies. A policy is reachable if we can find a request such that this policy is applicable. Thus, by removing unreachable policies we will not change the behaviour of the whole set of policies. 

\bibliographystyle{plain}
\bibliography{bibliography}

\newcommand{\niceiff}[3]{
\begin{align*}
& \semantics{#1}(\mc{Q}) = #2
&& \text{if and only if}
&& \val(#1, #2) \in M
\enspace#3
\end{align*}
}

\newcommand{\niceiffcomb}[2]{
\begin{align*}
& \bigoplus_{#1}(\mb{R}) = #2
&& \text{if and only if}
&& \algo(#1, \ppol, #2) \in M
\end{align*}
}
\appendix
\section{ASP Semantics}
\label{ap:asp semantics}
\subsection{Interpretations and Models}
The \emph{Herbrand Universe} \mc{U_\mc{L}} for a language \mc{L} is the set of all ground terms that can be formed from the constants and function symbols appearing in \mc{L}. The \emph{Herbrand base} $\mc{B}_{\mc{L}}$ for a language \mc{L} is the set of all ground atoms that can be formed by using predicate symbols from \mc{L} and ground terms from \mc{U_{\mc{L}}} as arguments. By $\mc{B}_{\Prog}$ we denote the Herbrand base for language underlying the program \Prog. When the context is clear, we are safe to omit  \Prog. 

An \emph{interpretation} $I$ of a program \Prog is a mapping from the Herbrand base $\mc{B}_{\Prog}$ to the set of truth values: true and false ($\Set{\top, \bot}$). All atoms belong to interpretation $I$ are mapped to $\top$. All atoms which does not occur in $I$ are mapped to $\bot$. 

The truth value of arbitrary formulae under some interpretation can be determined from a truth table as usual (see Table \ref{t:truth table}). 
\begin{table}
\centering
\caption{Truth Values for Formulae}
\label{t:truth table}
\[
\begin{array}{c|c|c|c|c}
\phi & \psi & \mynot \phi & \phi \wedge \psi& \phi \la \psi  \\
\hline
\top & \top & \bot & \top & \top \\
\top & \bot & \bot & \bot & \top \\
\bot & \top & \top & \bot & \bot \\
\bot & \bot & \top & \bot & \top
\end{array}
\]
\end{table}

The logical value of ground formulae can be derived from Table \ref{t:truth table} in the usual way. A formula $\phi$ is then \emph{true under interpretation $I$}, denoted by $I(\phi) = \top$, if all its ground instances are true in $I$; it is \emph{false under interpretation $I$}, denoted by $I(\phi) = \bot$, if there is a ground instance of $\phi$ that is false in $I$. 

Let $I$ be an interpretation. $I$  \emph{satisfies} formula $\phi$ if $I(\phi) = \top$. For a program \Prog, we say $I$ \emph{satisfies} of \Prog if $I$ satisfies for every rule in \Prog. An interpretation $I$ is a \emph{model} of formula $\phi$ if $I$ satisfies $\phi$. 

Let \mc{I} be a collection of interpretations. Then an interpretation $I$ is \mc{I} is called \emph{minimal} in \mc{I} if and only if there is no interpretation $J$ in \mc{I} such that $J \subsetneq I$. An interpretation $I$ is called \emph{least} in \mc{I} if and only if $I \subseteq J$ for any interpretation $J$ in \mc{I}. A model $M$ of a program \Prog is called minimal (respectively least) if it is minimal (respectively least) among all models of \Prog.

\subsection{Answer Set}
An interpretation $I$ of \ground{\Prog} is an answer set for  \Prog if $I$ is minimal (w.r.t. set inclusion) among the  interpretations satisfying the rules of 
\[\begin{array}{ll}
\Prog^I = \{ A \la B_1, \ldots, B_m | & A \la B_1, \ldots, B_m, \mynot B_{m+1}, \ldots, \mynot B_{n} \in \Prog \mbox{ and} \\
 & I(\mynot B_{m+1}, \ldots, \mynot B_{n}) = \top\} 
\end{array}
\]

\section{Proofs}

\begin{lemma}
\label{l:1}
Let $M$ be an answer set of program \Prog and let $H \la \Body$ be  a rule in \Prog.
Then, $H \in M$ if $M(\Body) = \top$.
\end{lemma}

\begin{proof}
Let $\Body = B_1, \ldots, B_m, \mynot B_{m+1}, \ldots, \mynot B_n$. To show the lemma holds, suppose $M(\Body) = \top$. Then we find that $\Set{B_1, \dotsc, B_m} \subseteq M$ and $M \cap \Set{B_{m+1}, \ldots, B_n} = \emptyset$. Since $M$ is a minimal model of $\Prog^M$  then we find that $H \la B_1, \ldots, B_n$ is in $\Prog^M$. Since $\Set{B_1, \ldots, B_m} \subseteq M$ and $M$ is a model then $M(H) = \top$. Thus $H \in M$.\qed



\end{proof}

The Lemma \ref{l:1} only ensures that if the body of a rule is true under an answer set $M$ then the head is also in $M$. However, in general, if the head of a rule is in a answer set $M$ then there is no guarantee that the body  is always true under $M$. For example, suppose we have a program $\Set{p \la \top., p \la q.}$. In this example the only answer set is $M = \Set{p}$. We can see that $p$ is in $M$. However, $q$ is not in $M$, thus, $M(q)$ is false.

\begin{lemma}
\label{l:2}
Let $M$ be an answer set of program \Prog  and let $H$ be in  $M$. Then, there is a rule in \Prog where $H$ as the head. 
\end{lemma}

\begin{proof}
Suppose that $M$ is an answer set of program \Prog. Then we find that $M$ is a minimal model of $\Prog^M$. Suppose $H \in M$ and there is no rule in $\Prog^M$ such that $H$ as the head. Then, we find that $M' = M / \Set{H}$ and $M'$ is a model of $\Prog^M$. Since  $M$ is a minimal model of $\Prog^M$ but we have $M' \subset M$. Therefore we find a contradiction. Thus, there should be a rule in $\Prog^M$ such that $H$ as the head. Hence, there is a rule in $\Prog$ such that $H$ as the head. \qed
\end{proof}

\begin{lemma}
\label{l:3}
Let $M$ be an answer set of program \Prog and let $H$ be in $M$. Then, there exists a rule where $H$ as the head and the body is true under $M$. 
\end{lemma}

\begin{proof}
Suppose that $M$ is an answer set of program \Prog. Since $H$ is in $M$ thus, by Lemma \ref{l:2}, we find that there is a rule in $\Prog$ in a form $H \la \Body$. Suppose that $M(\Body) \neq \top$. Therefore, $H \la \Body$ is not in $\Prog^M$. Moreover, we  can find another interpretation $M'$ such that $M / \Set{H}$ and $M'$ is also a model of $\Prog^M$. However, we know that $M$ is a minimal model for $\Prog^M$ but we have $M' \subset M$. Thus, there is a contradiction. \qed
\end{proof}


\newcommand{\ppol}{{\ensuremath{\mc{P}}}}
\newcommand{\pmatch}{\ensuremath{\Prog^\mc{M}}}
\newcommand{\pallof}{\ensuremath{\Prog^\mc{A}}}
\newcommand{\panyof}{\ensuremath{\Prog^\mc{E}}}
\newcommand{\ptarget}{\ensuremath{\Prog^\mc{T}}}
\newcommand{\pcond}{\ensuremath{\Prog^\mc{C}}}
\newcommand{\prule}{\ensuremath{\Prog^\mc{R}}}
\newcommand{\ppolicy}{\ensuremath{\Prog^\mc{\ppol}}}
\newcommand{\pcomb}{\ensuremath{\Prog^\mc{\comb}}}
\newcommand{\pprog}{\ensuremath{\Prog}}

We define some notation:
\[\footnotesize
\begin{array}{|l|l|l|}
\hline
\textbf{XACML Components} & \textbf{XACML Symbols} & \textbf{LP Symbols} \\
\hline
\hbox{\Match} & \mcM & \pmatch = \Prog_{\mcM} \\
\hbox{\AllOf} & \mcA = \bigwedge \mcM_i & \pallof = \bigcup \Prog^{\mcM_i} \cup \Prog_\mcA \\
\hbox{\AnyOf} & \mcE = \bigvee \mcA_i & \panyof =  \bigcup \Prog^{\mcA_i} \cup \Prog_\mcE\\
\hbox{\Target} & \mcT = \bigwedge \mcE_i & \ptarget = \bigcup \Prog^{\mcE_i} \cup \Prog_{\mcT}\\
\hbox{\Condition} & \mcC & \pcond = \Prog_{\mcC}\\
\hbox{\Rule} & \mcR = [E, \mcT, \mcC] & \prule = \Prog^\mcT \cup \Prog^\mcC \cup \Prog_\mcR \\
\hbox{\Policy} & \mcP = [\mcT, \seq{\mcR_1, \dotsc, \mcR_n}, \comb] & \ppolicy = \bigcup \Prog^{\mc{R}_i} \cup \Prog^{\mc{T}}  \cup \Prog^{\comb} \cup \Prog_{\mc{\ppol}} \\
\hbox{\PS} & \mcPS = [\mcT, \seq{\mcP_1, \dotsc, \mcP_n}, \comb] & \ppolicy = \bigcup \Prog^{\mcP_i} \cup \Prog^{\mc{T}}  \cup \Prog^{\comb} \cup \Prog_{\mcPS} \\
\hbox{Combining Algorithm} & \textrm{\comb\ is either \po\ or \fa or \ooa} & \pcomb = \bigcup \Prog_{\mcR_i} \cup \Prog_{\mcP_j} \Prog_{\comb}\\
\hline  
\end{array}
\]

\vspace{5pt}\noindent\textbf{Match Evaluation.}

\begin{lemma} 
\label{p:match1}
Let $\Prog = \Prog_{\mc{Q}} \cup \pmatch$ be a program and $M$ be an answer set of \Prog. Then,  \niceiff{\mcM}{\m}{.}

\end{lemma}
\begin{proof}
($\Rightarrow$)
Suppose that $\semantics{\mc{M}} = \m$ holds. 
Then, as defined in \eqref{eq:match}, $\mcM \in \mc{\mc{Q}}$ and $\error(\mcM) \not \in \mc{\mc{Q}}$. 
Based on the transformation of \Request element, we find out that $\mcM \la \top $ is in $\Prog$ and there is no rule where $\error(\mcM)$ as the head in $\Prog$. 
Since $M$ is the minimal model of $\Prog$, we get that $\mcM \in M$ and $\error(\mcM) \not \in M$.
Thus we get that $M\left(\mcM \land \mynot \error(\mcM)\right)$ $= \top$.
Therefore, by Lemma, \ref{l:1} $\val(\mcM, \m) \in M$.
\\
($\Leftarrow$)
Suppose that $\val(\mc{M}, \m) \in M$.
By Lemma \ref{l:3} we get that there is a rule  where $\val(\mcM, \m)$ as the head and the body is true under $M$. 
Since there is only one rule  where $\val(\mcM, \m)$ as the head in \Prog, i.e., $\val(\mcM, \m) \la \mcM, \mynot \error(\mcM)$, then, we find that $M(\mcM \land \mynot \error(\mcM)) = \top$. 
Therefore, $\mcM \in M$ and $\error(\mcM) \not \in M$. 
Since the only possible to have $\mcM$ true in this case is only through the \Request transformation, we get that $\mcM \in \mc{Q}$ and $\error(\mcM) \not \in \mc{\mc{Q}}$.
Therefore, we obtain $\semantics{\mcM}(\mcQ) = \m$.  
\qed
\end{proof}

\begin{lemma} 
\label{p:match2}
Let $\Prog = \Prog_{\mc{Q}} \cup \pmatch$ be a program and $M$ be an answer set of \Prog. Then, \niceiff{\mcM}{\nm}{.}
\end{lemma}
\begin{proof}
($\Rightarrow$)
Suppose that $\semantics{\mc{M}} = \nm$. 
Then, as defined in \eqref{eq:match} we have that $\mcM \not \in \mc{\mc{Q}}$ and $\error(\mcM) \not \in \mc{\mc{Q}}$. 
Based on the transformation of \Request, we find out that there is no rule where $\mcM$ and $\error(\mcM)$ as the heads. 
Since $M$ is the minimal model of $\Prog$, we get that $\mcM$ and $\error(\mcM)$ are not in $M$.
Thus, we get that $M(\mynot \mcM \wedge \mynot \error(\mcM)) = \top$.
Therefore, by Lemma \ref{l:1}, $\val(\mcM, \nm) \in M$.
\\
($\Leftarrow$)
Suppose that $\val(\mc{M}, \nm) \in M$ where $\mc{M} = \mcM$.
Based on Lemma \ref{l:3} we get that there is a rule  where $\val(\mcM, \nm)$ as the head and the body is true under $M$. 
Since there is only one rule  where $\val(\mcM, \nm)$ as the head in \Prog, i.e., $\val(\mc{M}, \nm) \la \mynot \mcM, \mynot \error(\mcM)$, then, we find that $M(\mynot \mcM \wedge \mynot \error(\mcM)) = \top$. 
Therefore, $\mcM \not \in M$ and $\error(\mcM) \not \in M$. 
Since the only possible of declaring facts in this case is only through the \Request transformation, we get that $\mcM \not \in \mcQ$ and $\error(\mcM) \not \in \mcQ$.
Therefore, we obtain $\semantics{\mcM}(\mcQ) = \nm$.  
\qed
\end{proof}

\begin{lemma} 
\label{p:match3}
Let $\Prog = \Prog_{\mc{Q}} \cup \pmatch$ be a program and $M$ be an answer set of \Prog. Then,  \niceiff{\mcM}{\idt}{.}
\end{lemma}
\begin{proof}
($\Rightarrow$)
Suppose that $\semantics{\mc{M}}(\mc{Q}) = \idt$ holds where $\mc{M} = \mcM$.
Then, as defined in \eqref{eq:match}, we have that $\error(\mcM) \in \mc{Q}$.
Based on the transformation of \Request element, we find out that $\error(\mcM) \la \top $ is in $\Prog$.
Since $M$ is the minimal model of $\Prog$, then, we get that $\error(\mcM) \in M$.
Thus, we get that $M(\error(\mcM)) = \top$.
Therefore, $\val(\mcM, \idt) \in M$ since $M$ is the minimal model of $\Prog$.
\\
($\Leftarrow$)
Suppose that $\val(\mc{M}, \idt) \in M$.
Based on Lemma \ref{l:3} we get that there is a rule  where $\val(\mcM, \idt)$ as the head and the body is true under $M$. 
Since there is only one rule in \Prog with $\val(\mcM, \idt)$ in the head, i.e., $\val(\mc{M}, \idt) \la\error(\mcM)$, then we find that $M(\error(\mcM)) = \top$. 
Therefore, $\error(\mcM)  \in M$. 
Since the only possible to have $\error(\mcM)$ true in this case is only through the \Request transformation, we get that $\error(\mcM) \in \mc{Q}$.
Therefore, we  obtain $\semantics{\mcM}(\mcQ) = \idt$.  
\qed
\end{proof}

\begin{proposition}
\label{p:match}
Let $\Prog = \Prog_{\mc{Q}} \cup \pmatch$ be a program and $M$ be an answer set of \Prog. Then, \niceiff{\mcM}{V}{.}
\end{proposition}
\begin{proof}
It follows from Lemma \ref{p:match1}, Lemma \ref{p:match2} and Lemma \ref{p:match3} since the value of $V$ only has three possibilities, i.e., $\Set{\m, \nm, \idt}$.
\qed
\end{proof}

\vspace{5pt}\noindent\textbf{\AllOf Evaluation.}
\begin{lemma}
\label{p:allof1}
Let $\Prog = \Prog_{\mc{Q}} \cup \pallof$ be a program and $M$ be an answer set of \Prog. Then, \niceiff{\mcA}{\m}{.}
\end{lemma}
\begin{proof}
Let $\mcA = \bigwedge _{i = 1}^n \mcM_i$. \\
($\Rightarrow$)
Suppose that $\semantics{\mc{A}}(\mc{Q}) = \m$ holds. 
Then, as defined in \eqref{eq:allof}, $\forall i: \semantics{\mc{M}_i}(\mc{Q}) = \m, 1 \leq i \leq n$. 
Based on Prop. \ref{p:match}, $\forall_i: \val(\mc{M}_i, \m) \in M, 1 \leq i \leq n$.
Therefore, $M(\val(\mc{M}_1, \m) \wedge \ldots \wedge \val(\mc{M}_n, \m)) = \top$.
Hence, by Lemma \ref{l:1}, $\val(\mc{A}, \m) \in M$.
\\
($\Leftarrow$)
Suppose that $\val(\mc{A}, \m) \in M$.
Based on Lemma \ref{l:3}, there is a rule  where $\val(\mc{A}, \m)$ as the head and the body is true under $M$. 
Since there is only one rule  in \Prog with $\val(\mc{A}, \m)$ in the head, i.e., 
$\val(\mc{A}, \m) \la \val(\mc{M}_1, \m), \dotsc, \val(\mc{M}_n, \m), $
we find that $M( \val(\mc{M}_1, \m) \wedge \ldots \wedge \val(\mc{M}_n, \m)) = \top$. 
Therefore, $\val(\mc{M}_i, \m) \in M, 1 \leq i \leq n$.
Based on Prop. \ref{p:match}, $\semantics{\mc{M}_i}(\mc{Q}) = \m, 1 \leq i \leq n$.
Therefore, based on \eqref{eq:allof}, we obtain $\semantics{\mc{A}}(\mc{Q}) = \m$.  
\qed
\end{proof}

\begin{lemma} 
\label{p:allof2}
Let $\Prog = \Prog_{\mc{Q}} \cup \pallof$ be a program and $M$ be an answer set of \Prog. Then, \niceiff{\mcA}{\nm}{.}
\end{lemma}
\begin{proof}
Let $\mcA = \bigwedge _{i = 1}^n \mcM_i$. \\
($\Rightarrow$)
Suppose that $\semantics{\mc{A}}(\mc{Q}) = \nm$ holds.
Then, as defined in \eqref{eq:allof} we have that $\exists i: \semantics{\mc{M}_i}(\mc{Q}) = \nm$. 
Based on Prop. \ref{p:match} we get that $\exists i: \val(\mc{M}_i, \nm) \in M$.
Thus, we get that $\exists i: M(\val(\mc{M}_i), \nm) = \top$.
Therefore, by Lemma \ref{l:1}, $\val(\mc{M}, \nm) \in M$.
\\
($\Leftarrow$)
Suppose that $\val(\mc{A}, \nm) \in M$.
Based on Lemma \ref{l:3} we get that there is a rule  where $\val(\mc{A}, \nm)$ as the head and the body is true under $M$. 
Based on \AllOf transformation, $\exists i: M(\val(\mc{M}_i), \nm) = \top$. 
Therefore, $\exists i: \val(\mc{M}_i, \nm) \in M$. 
Based on Prop. \ref{p:match} we get that $\exists i: \semantics{\mc{M}_i}(\mc{Q}) = \nm$.
Therefore, based on \eqref{eq:allof}, we obtain $\semantics{\mc{A}}(\mc{Q}) = \nm$.  
\qed
\end{proof}

\begin{lemma} 
\label{p:allof3}
Let $\Prog = \Prog_{\mc{Q}} \cup \pallof$ be a program and $M$ be an answer set of \Prog. Then, \niceiff{\mcA}{\idt}{.}
\end{lemma}
\begin{proof}
($\Rightarrow$)
Suppose that $\semantics{\mc{A}}(\mc{Q}) = \idt$.
Then, as defined in \eqref{eq:allof}, $\semantics{\mc{A}}(\mc{Q}) \neq \m$ and $\semantics{\mc{A}}(\mc{Q}) \neq \nm$.
Thus, by Lemma \ref{p:allof1} and Lemma \ref{p:allof2}, $\val(\mc{A}, \m) \not \in M$ and $\val(\mc{A}, \nm) \not \in M$. 
Hence, $M(\mynot \val(\mc{A}, \m) \wedge \mynot \val(\mc{A}, \nm)) = \top.$
Therefore, by Lemma \ref{l:1}, $\val(\mc{A}, \idt) \in M$.
\\
($\Leftarrow$)
Suppose that $\val(\mc{A}, \idt) \in M$.
Based on Lemma \ref{l:3}, there is a rule  where $\val(\mc{A}, \idt)$ as the head and the body is true under $M$. 
There is only one rule  where $\val(\mc{A}, \idt)$ as the head in \Prog, i.e., 
$\val(\mc{A}, \idt) \la \mynot \val(\mc{A}, \m), \mynot \val(\mc{A}, \nm).$ 
Hence, $\val(\mc{A}, \m) \not \in M$ and $\val(\mc{A}, \nm) \not \in M$. 
Based on Lemma \ref{p:allof1} and Lemma \ref{p:allof2} we get that $\semantics{\mc{A}}(\mc{Q}) \neq \m$ and $\semantics{\mc{A}}(\mc{Q}) \neq \nm$.
Therefore, based on \eqref{eq:allof}, we obtain $\semantics{\mc{A}}(\mc{Q}) = \idt$.  
\qed
\end{proof}

\begin{proposition}
\label{p:allof}
Let $\Prog = \pprog_\mc{Q} \cup \pallof$ be a program obtained by merging \Request transformation program $\pprog_{\mc{Q}}$ and   \AllOf $\mc{A}$ transformations program with all of its components $\pallof$. Let $M$ be an answer set of \pprog. 
Then, \niceiff{\mcA}{V}{.}
\end{proposition}
\begin{proof}
It follows from Lemma \ref{p:allof1}, Lemma \ref{p:allof2} and Lemma \ref{p:allof3} since the value of $V$ only has three possibilities, i.e., $\Set{\m, \nm, \idt}$.
\qed
\end{proof}

\vspace{5pt}\noindent\textbf{\AnyOf Evaluation.}
\begin{lemma}
\label{p:anyof1}
Let $\Prog = \Prog_{\mc{Q}} \cup \panyof$ be a program and $M$ be an answer set of \Prog. Then, \niceiff{\mcE}{\m}{.}
\end{lemma}
\begin{proof}
Let $\mcE = \bigvee_{i = 1}^n \mcA_i $\\
($\Rightarrow$)
Suppose that $\semantics{\mc{E}}(\mc{Q}) = \m$ holds. 
Then, as defined in \eqref{eq:anyof}, $\exists i: \semantics{\mc{A}_i}(\mc{Q}) = \m, 1 \leq i \leq n$. 
Based on Prop. \ref{p:allof}, $\exists_i: \val(\mc{A}_i, \m) \in M, 1 \leq i \leq n$.
Thus, $\exists i: M(\val(\mc{A}_i, \m)) = \top$.
Therefore, by Lemma \ref{l:1}, $\val(\mc{E}, \m) \in M$.
\\
($\Leftarrow$)
Suppose that $\val(\mc{E}, \m) \in M$.
Based on Lemma \ref{l:3}, here is a rule  where $\val(\mc{E}, \m)$ as the head and the body is true under $M$. 
Based on \AnyOf transformation, $\exists i: M(\val(\mc{E}_i), \m) = \top$. 
Therefore, $\exists i: \val(\mc{E}_i, \m) \in M$. 
Based on Prop. \ref{p:allof}, $\exists i: \semantics{\mc{E}_i}(\mc{Q}) = \m$.
Therefore, based on \eqref{eq:anyof}, we obtain $\semantics{\mc{E}}(\mc{Q}) = \m$.  
\qed
\end{proof}

\begin{lemma}
\label{p:anyof2}
Let $\Prog = \Prog_{\mc{Q}} \cup \panyof$ be a program and $M$ be an answer set of \Prog. Then, \niceiff{\mcE}{\nm}{.}
\end{lemma}
\begin{proof}
Let $\mcE = \bigvee_{i = 1}^n \mcA_i $\\
($\Rightarrow$)
Suppose that $\semantics{\mc{E}}(\mc{Q}) = \nm$ holds.
Then, as defined in \eqref{eq:anyof}, $\forall i: \semantics{\mc{A}_i}(\mc{Q}) = \nm$. 
Based on Prop. \ref{p:allof}, $\forall i: \val(\mc{A}_i, \nm) \in M$.
Thus,  $M(\val(\mc{A}_1, \nm) \wedge \dotsb \wedge \val(\mc{A}_n, \nm)) = \top$.
Therefore, by Lemma \ref{l:1}, $\val(\mc{E}, \nm) \in M$.
\\
($\Leftarrow$)
Suppose that $\val(\mcE, \nm) \in M$.
By Lemma \ref{l:3}, there is a rule  where $\val(\mc{E}, \nm)$ as the head and the body is true under $M$. 
There is only one rule  in \Prog  with $\val(\mc{E}, \m)$ in the head in \Prog, i.e., 
$\val(\mc{E}, \nm) \la \val(\mc{A}_1, \nm), \ldots, \val(\mc{A}_n, \nm).$
Thus, $M( \val(\mc{A}_1, \nm) \wedge \ldots \wedge \val(\mc{A}_n, \nm)) = \top$. 
Therefore, $\val(\mc{A}_i, \nm) \in M, 1 \leq i \leq n$.
Based on Prop. \ref{p:allof}, $\semantics{\mc{A}_i}(\mc{Q}) = \nm, 1 \leq i \leq n$.
Therefore, based on \eqref{eq:anyof}, we obtain $\semantics{\mc{E}}(\mc{Q}) = \nm$.  
\qed
\end{proof}

\begin{lemma}
\label{p:anyof3}
Let $\Prog = \Prog_{\mc{Q}} \cup \panyof$ be a program and $M$ be an answer set of \Prog. Then, \niceiff{\mcE}{\idt}{.}
\end{lemma}
\begin{proof}
($\Rightarrow$)
Suppose that $\semantics{\mc{E}}(\mc{Q}) = \idt$.
Then, as defined in \eqref{eq:anyof}, $\semantics{\mc{E}}(\mc{Q}) \neq \m$ and $\semantics{\mc{E}}(\mc{Q}) \neq \nm$.
Thus, by Lemma \ref{p:anyof1} and Lemma \ref{p:anyof2},  $\val(\mc{E}, \m) \not \in M$ and $\val(\mc{E}, \nm) \not \in M$. 
Hence, $M(\mynot \val(\mc{E}, \m) \wedge \mynot \val(\mc{E}, \nm)) = \top.$
By Lemma \ref{l:1}, $\val(\mc{E}, \idt) \in M$.
\\
($\Leftarrow$)
Suppose that $\val(\mc{E}, \idt) \in M$.
Based on Lemma \ref{l:3}, there is a rule  where $\val(\mc{E}, \idt)$ as the head and the body is true under $M$. 
There is only one rule in \Prog with $\val(\mc{E}, \idt)$ in the head, i.e., 
$\val(\mc{E}, \idt) \la \mynot \val(\mc{E}, \m), \mynot \val(\mc{E}, \nm).$
Hence, $\val(\mc{E}, \m) \not \in M$ and $\val(\mc{E}, \nm) \not \in M$. 
Based on Lemma \ref{p:anyof1} and Lemma \ref{p:anyof2},  $\semantics{\mc{E}}(\mc{Q}) \neq \m$ and $\semantics{\mc{E}}(\mc{Q}) \neq \nm$. 
Therefore, based on \eqref{eq:anyof}, we obtain $\semantics{\mc{E}}(\mc{Q}) = \idt$.  
\qed
\end{proof}

\begin{proposition}
\label{p:anyof}
Let $\Prog = \pprog_\mc{Q} \cup \panyof$ be a program obtained by merging \Request transformation program $\pprog_\mc{Q}$ and \AnyOf   $\mc{E}$ transformations program with all of of its components \panyof. Let $M$ be an answer set of \pprog. 
Then,  \niceiff{\mcE}{V}{.}
\end{proposition}
\begin{proof}
It follows from Lemma \ref{p:anyof1}, Lemma \ref{p:anyof2} and Lemma \ref{p:anyof3} since the value of $V$ only has three possibilities, i.e., $\Set{\m, \nm, \idt}$.
\qed
\end{proof}

\vspace{5pt}\noindent\textbf{Target Evaluation.}
\begin{lemma}
\label{p:target1}
Let $\Prog = \Prog_{\mc{Q}} \cup \ptarget$ be a program and $M$ be an answer set of \Prog. Then, \niceiff{\mcT}{\m}{.}
\end{lemma}
\begin{proof}
Let $\mcT = \bigwedge_{i = 1}^n \mcE$.\\
($\Rightarrow$)
Suppose that $\semantics{\mc{T}}(\mc{Q}) = \m$ holds. 
Then, as defined in  \eqref{eq:target}, we have that 
\begin{enumerate}
\item 
$\forall i: \semantics{\mc{E}_i}(\mc{Q}) = \m, 1 \leq i \leq n$. 
Based on Prop. \ref{p:anyof}, $\forall_i: \val(\mc{E}_i, \m) \in M, 1 \leq i \leq n$.
Thus,  $M(\val(\mc{E}_1, \m) \wedge \ldots \wedge \val(\mc{E}_n, \m)) = \top$.
Therefore, by Lemma \ref{l:1}, $\val(\mc{T}, \m) \in M$.
\item $\mc{T} = \mathsf{null} $.
Based on \Target transformation we get that $\val(\mathsf{null}, \m) \la \top.$
Thus, $\val(\mc{T}, \m) \in M$ since $M$ is the minimal model of $\Prog$.
\end{enumerate}
($\Leftarrow$)
Suppose that $\val(\mc{T}, \m) \in M$.
Based on Lemma \ref{l:3}, there is a clause where $\val(\mc{T}, \m)$ as the head and the body is true under $M$. 
\begin{enumerate}
\item $\mc{T} \neq \texttt{null}$.
There is a rule where $\val(\mc{T}, \m)$ as the head, i.e., $\val(\mc{T}, \m) \la \val(\mc{E}_1, \m), \ldots, \val(\mc{E}_n, \m).$
Then, we find that $M( \val(\mc{E}_1, \m) \wedge \ldots \wedge \val(\mc{E}_n, \m)) = \top$. 
Therefore, $\val(\mc{E}_i, \m) \in M, 1 \leq i \leq n$.
Based on Prop. \ref{p:anyof},  $\semantics{\mc{E}_i}(\mc{Q}) = \m, 1 \leq i \leq n$.
Therefore, based on \eqref{eq:target}, we obtain $\semantics{\mc{T}}(\mc{Q} = \m$.  
\item $\mc{T} = \mathsf{null}$.
Then, there is a rule in \Prog where $\val(\mathsf{null}, \m)$ as the head, i.e., $\val(\mathsf{null}, \m) \la \top.$
Thus, based on the definition \eqref{eq:target}, we obtain $\semantics{\mc{T}}(\mc{Q}) = \m$. \qed  
\end{enumerate}
\end{proof}

\begin{lemma}
\label{p:target2}
Let $\Prog = \Prog_{\mc{Q}} \cup \ptarget$ be a program and $M$ be an answer set of \Prog. Then, \niceiff{\mcT}{\nm}{.}
\end{lemma}
\begin{proof}
($\Rightarrow$)
Suppose that $\semantics{\mc{T}}(\mc{Q}) = \nm$ holds.
Then, as defined in  \eqref{eq:target}, $\exists i: \semantics{\mc{E}_i}(\mc{Q}) = \nm$. 
Therefore, based on Prop. \ref{p:anyof}, $\exists i: \val(\mc{E}_i, \nm) \in M$.
Hence, $\exists i: M(\val(\mc{E}_i), \nm) = \top$.
Thus, by Lemma \ref{l:1}, $\val(\mc{E}, \nm) \in M$.
\\
($\Leftarrow$)
Suppose that $\val(\mc{T}, \nm) \in M$.
Based on Lemma \ref{l:3}, there is a clause where $\val(\mc{T}, \nm)$ as the head and the body is true under $M$. 
Based on \AllOf transformation,  $\exists i: M(\val(\mc{E}_i), \nm) = \top$. 
Therefore, $\exists i: \val(\mc{E}_i, \nm) \in M$. 
Based on Prop. \ref{p:anyof}, $\exists i: \semantics{\mc{E}_i}(\mc{Q}) = \nm$.
Therefore, based on  \eqref{eq:target}, we obtain $\semantics{\mc{T}}(\mc{Q}) = \nm$.  
\qed
\end{proof}

\begin{lemma}
\label{p:target3}
Let $\Prog = \Prog_{\mc{Q}} \cup \ptarget$ be a program and $M$ be an answer set of \Prog. 
Then, \niceiff{\mcT}{\idt}{.}
\end{lemma}
\begin{proof}
($\Rightarrow$)
Suppose that $\semantics{\mc{T}}(\mc{Q}) = \idt$.
Then, as defined in  \eqref{eq:target}, $\semantics{\mc{T}}(\mc{Q}) \neq \m$ and $\semantics{\mc{T}}(\mc{Q}) \neq \nm$.
Thus, by Lemma \ref{p:target1} and Lemma \ref{p:target2}, $\val(\mc{T}, \m) \not \in M$ and $\val(\mc{T}, \nm) \not \in M$.
Hence, $M(\mynot \val(\mc{T}, \m) \wedge \mynot \val(\mc{T}, \nm)) = \top.$
Therefore, by Lemma \ref{l:1}, $\val(\mc{T}, \idt) \in M$.
\\
($\Leftarrow$)
Suppose that $\val(\mc{T}, \idt) \in M$.
Based on Lemma \ref{l:3}, there is a clause where $\val(\mc{T}, \idt)$ as the head and the body is true under $M$. 
There is  only one rule in \Prog  with $\val(\mc{T}, \idt)$ in the head  , i.e., 
$\val(\mc{T}, \idt) \la \mynot \val(\mc{T}, \m), \mynot \val(\mc{T}, \nm). $
Thus, $\val(\mc{T}, \m) \not \in M$ and $\val(\mc{T}, \nm) \not \in M$. 
Based on Lemma \ref{p:target1} and Lemma \ref{p:target2}, $\semantics{\mc{T}}(\mc{Q}) \neq \m$ and $\semantics{\mc{T}}(\mc{Q}) \neq \nm$. 
Therefore, based on \eqref{eq:target} we obtain $\semantics{\mc{T}}(\mc{Q}) = \idt$.  
\qed
\end{proof}

\begin{proposition}
\label{p:target}
Let $\Prog = \pprog_\mc{Q} \cup \ptarget$ be a program obtained by merging \Request transformation program $\pprog_{\mc{Q}}$ and \Target   $\mc{T}$ transformations program with all of of its components \ptarget. Let $M$ be an answer set of \pprog. 
Then, \niceiff{\mcT}{V}{.}
\end{proposition}
\begin{proof}
It follows from Lemma \ref{p:target1}, Lemma \ref{p:target2} and Lemma \ref{p:target3} since the value of $V$ only has three possibilities, i.e., $\Set{\m, \nm, \idt}$.
\qed
\end{proof}

\vspace{5pt}\noindent\textbf{Condition Evaluation.}
\begin{proposition}
\label{p:condition}
Let $\Prog = \Prog_{\mc{Q}} \cup \Prog_{\mc{C}}$ be a program obtained from merging \Request transformation program $\Prog_{\mc{Q}}$ and \Condition transformation program $\Prog_{\mc{C}}$ and let $M$ be an answer set of $\Prog$. Then, \niceiff{\mcC}{V}{.}
\end{proposition}
\begin{proof}
It follows from the equation \eqref{eq:condition} that the \Condition evaluation based on the value of $\eval$ function, the same case in the \Condition program transformation. 
\qed
\end{proof}

\vspace{5pt}\noindent\textbf{Rule Evaluation.}
\begin{lemma}
\label{p:rule1}
Let $\Prog = \pprog_\mc{Q} \cup \prule$ be a program obtained by merging \Request transformation program $\pprog_{\mc{Q}}$ and \Rule   $\mc{R}$ transformations program with all of of its components \prule. Let $M$ be an answer set of \pprog. 
Then, \niceiff{\mcR}{E}{}
where $E$ is \Rule's effect, either $\p$ or $\d$.
\end{lemma}
\begin{proof}
$(\Rightarrow)$
Suppose that $\semantics{\mc{R}}(\mc{Q}) = E$ holds. 
Then, as defined in  \eqref{eq:target}, $\semantics{\mc{T}}(\mc{Q}) = \m$ and $\semantics{\mc{C}}(\mc{Q} = \t)$.
Based on Prop. \ref{p:rule} and Prop. \ref{p:condition}, $\val(\mc{T}, \m) \in M$ and $\val(\mc{C}, \t) \in M$.
Thus,  $M(\val(\mc{T}, \m) \wedge \val(\mc{C}, \t)) = \top$.
Therefore, by Lemma \ref{l:1}, $\val(\mc{R}, E) \in M$ .
\\
$(\Leftarrow)$
Suppose that $\val(\mc{R}, E) \in M$.
Based on Lemma \ref{l:3}, there is a clause where $\val(\mc{R}, E)$ as the head and the body is true under $M$. 
There is only one rule in $\Prog$ with $\val(\mc{R}, E)$ in the head, i.e., 
$\val(\mc{R}, \m) \la \val(\mc{T}, \m),\val(\mc{C}, \t). $
Then, we find that $M( \val(\mc{T}, \m) \wedge \val(\mc{C}, \t)) = \top$. 
Therefore, $\val(\mc{T}, \m)  \in M$ and $\val(\mc{C}, \t)  \in M $.
Based on Prop. \ref{p:target} and Prop. \ref{p:condition}, $\semantics{\mc{T}}(\mc{Q}) = \m$ and $\semantics{\mc{C}}(\mc{Q}) = \t$.
Therefore, based on \eqref{eq:rule} we obtain $\semantics{\mc{R}}(\mc{Q}) = E$.  
\qed
\end{proof}

\begin{lemma}
\label{p:rule2}
Let $\Prog = \pprog_\mc{Q} \cup \prule$ be a program obtained by merging \Request transformation program $\pprog_{\mc{Q}}$ and \Rule   $\mc{R}$ transformations program with all of of its components \prule. Let $M$ be an answer set of \pprog. 
Then, \niceiff{\mcR}{\na}{.}
\end{lemma}
\begin{proof}
$(\Rightarrow)$
Suppose that $\semantics{\mc{R}}(\mc{Q}) = \na$ holds. 
Then, as defined in  \eqref{eq:rule}, we have that 
\begin{enumerate}
\item 
$\semantics{\mc{T}}(\mc{Q}) = \m$ and  $\semantics{\mc{C}}(\mc{Q}) = \f$.
Based on Prop. \ref{p:target} and Prop. \ref{p:condition},  $\val(\mc{T}, \m) \in M$ and $\val(\mc{C}, \f) \in M$.
Thus, $M(\val(\mc{T}, \m) \wedge \val(\mc{C}, \f)) = \top$.
Therefore, by Lemma \ref{l:1}, $\val(\mc{R}, \na) \in M$.
\item 
$\semantics{\mc{T}}(\mc{Q}) = \nm$. 
Based on Prop. \ref{p:target}, $\val(\mc{T}, \nm) \in M$.
Thus, $M(\val(\mc{T}, \nm)) = \top$.
Therefore, by Lemma \ref{l:1}, $\val(\mc{R}, \na) \in M$.
\end{enumerate}
$(\Leftarrow)$
Suppose that $\val(\mc{R}, \na) \in M$.
Based on Lemma \ref{l:3}, there is a clause in \Prog where $\val(\mc{R}, \na)$ as the head and the body is true under $M$. 
There are rules in $\Prog$ where $\val(\mc{R}, \na)$ as the head, i.e., 
\begin{enumerate}
\item $\val(\mc{R}, \na) \la \val(\mc{T}, \m),\val(\mc{C}, \f).$ \\
Then, we find that $M( \val(\mc{T}, \m) \wedge \val(\mc{C}, \f)) = \top$. 
Therefore, $\val(\mc{T}, \m)  \in M$ and $\val(\mc{C}, \f) \in M$ .
Based on Prop. \ref{p:target} and Prop. \ref{p:condition}, $\semantics{\mc{T}}(\mc{Q}) = \m$ and $\semantics{\mc{C}}(\mc{Q}) = \f$.
Therefore, based on \eqref{eq:rule}, we obtain $\semantics{\mc{R}}(\mc{Q} = \na$.
\item $\val(\mc{R}, \na) \la \val(\mc{T}, \nm).$\\
Then, we find that $M( \val(\mc{T}, \nm)) = \top$. 
Therefore, $\val(\mc{T}, \nm)  \in M$.
Based on Prop. \ref{p:target}, $\semantics{\mc{T}}(\mc{Q}) = \nm$.
Therefore, based on \eqref{eq:rule}, we obtain $\semantics{\mc{R}}(\mc{Q} = \na$.
\qed
\end{enumerate}
\end{proof}

\begin{lemma}
\label{p:rule3}
Let $\Prog = \Prog_{\mc{Q}} \cup \prule$ be a program and $M$ be an answer set of \Prog. 
Then, \niceiff{\mcR}{\texttt{i}_E}{} 
where $E$ is \Rule's effect, either $\p$ or $\d$.
\end{lemma}
\begin{proof}
($\Rightarrow$)
Suppose that $\semantics{\mc{R}}(\mc{Q}) = \texttt{i}_E$.
Then, as defined in  \eqref{eq:rule}, $\semantics{\mc{R}}(\mc{Q}) \neq E$ and $\semantics{\mc{R}}(\mc{Q}) \neq \na$.
By Lemma \ref{p:rule1} and Lemma \ref{p:rule2}, $\val(\mc{R}, E) \not \in M$ and $\val(\mc{R}, \na) \not \in M$. 
Hence,  $M(\mynot \val(\mc{R}, E) \wedge \mynot \val(\mc{R}, \na)) = \top.$
Thus, by Lemma \ref{l:1}, $\val(\mc{R}, \texttt{i}_E) \in M$.
\\
($\Leftarrow$)
Suppose that $\val(\mc{R}, \idt) \in M$.
Based on Lemma \ref{l:3}, there is a clause where $\val(\mc{R}, \texttt{i}_E)$ as the head and the body is true under $M$. 
There is only one rule in \Prog with $\val(\mc{R}, \texttt{i}_E)$ in the head in, i.e.,
$\val(\mc{R}, \texttt{i}_E) \la \mynot \val(\mc{R}, E), \mynot \val(\mc{R}, \na).$ 
Therefore, $M(\mynot \val(\mc{R}, E) \wedge \mynot \val(\mc{R}, \na)) = \top$. 
Thus,$\val(\mc{R}, E) \not \in M$ and $\val(\mc{R}, \na) \not \in M$. 
Based on Lemma \ref{p:rule1} and Lemma \ref{p:rule2}, $\semantics{\mc{R}}(\mc{Q}) \neq E$ and $\semantics{\mc{R}}(\mc{Q}) \neq \na$.  
Hence, based on  \eqref{eq:rule} we obtain, $\semantics{\mc{R}}(\mc{Q}) = \texttt{i}_E$.  
\qed
\end{proof}

\begin{proposition}
\label{p:rule}
Let $\Prog = \pprog_\mc{Q} \cup \prule$ be a program obtained by merging \Request transformation program $\pprog_{\mc{Q}}$ and \Rule   $\mc{R}$ transformations program with all of of its components \prule. Let $M$ be an answer set of \pprog. 
Then, \niceiff{\mcR}{V}{.}
\end{proposition}
 
\begin{proof}
It follows from Lemma \ref{p:rule1}, Lemma \ref{p:rule2} and Lemma \ref{p:rule3} since the value of $V$ only has five possibilities, i.e., $\Set{\p, \d, \id, \ip, \na}$.
\qed
\end{proof}

\vspace{5pt}\noindent\textbf{Combining Algorithm: Permit-Overrides.}

\begin{lemma}
\label{p:po1}
Let $\pprog = \pprog_{\mc{Q}} \cup \pprog_{\po} \cup \ppolicy$  be a program obtained by merging \Request transformation program $\pprog_{\mc{Q}}$, permit-overrides combining algorithm transformation program $\pprog_{\po}$ and \Policy\ \ppol\ transformation program with its components \ppolicy. Let $M$ be an answer set of \pprog.
Then, \niceiffcomb{\po}{\p}
where $\mb{R} = \seq{\semantics{\mc{R}_1}(\mc{Q}), \ldots, \semantics{\mc{R}_n}(\mc{Q})}$ be a sequence of policy value where each $\mc{R}_i$ is a \Rule in the sequence inside  \Policy\ \ppol.
\end{lemma}
\begin{proof}
$(\Rightarrow)$
Suppose that $\bigoplus_{\po}(\mb{R}) = \p$ holds. 
Then, as defined in  \eqref{eq:po}, $\exists i: \semantics{\mc{R}_i}(\mc{Q}) = \p$ where $\mc{R}_i$ is a \Rule in the sequence inside \Policy\ \ppol.
Based on Prop. \ref{p:rule}, $\val(\mc{R}_i, \p) \in M$.
Based on the \Policy transformation, there is a rule in \Prog $\dec(\ppol, \mc{R}_i, \p) \la \val(\mc{R}_i, \p).$
Therefore, by Lemma \ref{l:1}, $\dec(\ppol, \mc{R}_i, \p) \in M$.
Thus, by Lemma \ref{l:1}, $\algo(\po, \ppol, \p) \in M$.
\\
$(\Leftarrow)$
Suppose that $\algo(\po, \ppol, \p) \in M$.
Based on Lemma \ref{l:3}, there is a rule where $\algo(\po, \ppol, \p)$ as the head and the body is true under $M$. 
There is only one rule in $\Prog$ with $\algo(\po, \ppol, \p)$ as the head, i.e., 
$\algo(\po, \ppol, \p) \la \dec(\ppol, \mc{R}, \p).$
Then, $M(\dec(\ppol, \mc{R}, \p)) = \top$. 
Therefore, $\dec(\ppol, \mc{R}, \p) \in M$.
Based on Lemma \ref{l:3}, there is a rule where $\dec(\ppol, \mc{R}, \p)$ as the head and the body is true under $M$. 
There is only one rule in $\Prog$, i.e., 
$\dec(\ppol, \mc{R}, \p) \la \val(\mc{R}, \p).$
Then, $M(\val(\mc{R}, \p)) = \top$. 
Therefore, $\val(\mc{R}, \p) \in M$.
Based on Prop. \ref{p:rule}, $\semantics{\mc{R}}(\mc{Q}) = \p$ and $\mc{R}$ belongs to the sequence inside \Policy\ \ppol. 
Therefore, based on \eqref{eq:po}, we obtain $\bigoplus_{\po}(\mb{R}) = \p$   
\qed
\end{proof}

\begin{lemma}
\label{p:po2}
Let $\pprog = \pprog_{\mc{Q}} \cup \pprog_{\po} \cup \ppolicy$  be a program obtained by merging \Request transformation program $\pprog_{\mc{Q}}$, permit-overrides combining algorithm transformation program $\pprog_{\po}$ and \Policy\ \ppol\ transformation program with its components \ppolicy. Let $M$ be an answer set of \pprog.
Then, \niceiffcomb{\po}{\idp}
where $\mb{R} = \seq{\semantics{\mc{R}_1}(\mc{Q}), \ldots, \semantics{\mc{R}_n}(\mc{Q})}$ be a sequence of policy value where each $\mc{R}_i$ is a \Rule in the sequence inside  \Policy\ \ppol.
\end{lemma}
\begin{proof}
$(\Rightarrow)$
Suppose that $\bigoplus_{\po}(\mb{R}) = \idp$ holds. 
Then, as defined in  \eqref{eq:po} we have that
\begin{enumerate}
\item
$\forall i: \semantics{\mc{R}_i}(\mc{Q}) \neq \p$ and
$\exists j: \semantics{\mc{R}_j}(\mc{Q}) = \idp$ where $\mc{R}_i$ and $\mc{R}_j$ are \Rule in the sequence inside \Policy\ \ppol.
Based on Prop. \ref{p:rule}, $\forall i: \val(\mc{R}_i, \p) \not \in M$ and $\exists j: \val(\mc{R}_j, \idp) \in M$.
Based on Lemma \ref{p:po1}, $\algo(\po, \ppol, \p) \not \in M$ since if it is in $M$ ,  there exists a \Rule $\mc{R}$ in the \Policy\ \ppol sequence such that $\semantics{\mc{R}}(\mc{Q}) = \p$.
Based on the \Policy transformation, there is a rule $\dec(\ppol, \mc{R}_j, \idp) \la \val(\mc{R}_j, \idp).$
By Lemma \ref{l:1}, $\dec(\ppol, \mc{R}_j, \idp) \in M$.
Thus, by Lemma \ref{l:1}, $\algo(\po, \ppol, \idp) \in M$.

\item 
$\forall i: \semantics{\mc{R}_i}(\mc{Q}) \neq \p$ and
$\exists j: \semantics{\mc{R}_j}(\mc{Q}) = \ip$ and $\exists j': \semantics{\mc{R_{j'}}
}(\mc{Q}) = \d$ where $\mc{R}_i$, $\mc{R}_j$ and $\mc{R}_{j'}$ are Rules in the sequence inside \Policy\ \ppol.
Based on Prop. \ref{p:rule}, $\forall i: \val(\mc{R}_i, \p) \not \in M$ and $\exists j: \val(\mc{R}_j, \ip) \in M$ and $\exists j: \val(\mc{R}_{j'}, \d) \in M$.
Based on Lemma \ref{p:po1}, $\algo(\po, \ppol, \p) \not \in M$ since if it is in $M$ ,  there exists a \Rule $\mc{R}$ in the \Policy\ \ppol sequence such that $\semantics{\mc{R}}(\mc{Q}) = \p$.
Based on the \Policy transformation, there are rules in \Prog in the form $\dec(\ppol, \mc{R}_j, \idp) \la \val(\mc{R}_j, \ip)$ and $\dec(\ppol, \mc{R}_j, \ip) \la \val(\mc{R}_{j'}, \d).$
Thus, by Lemma \ref{l:1}, $\dec(\ppol, \mc{R}_j, \ip) \in M$ and $\dec(\ppol, \mc{R}_{j'}, \d) \in M$.
Hence, by Lemma \ref{l:1}, $\algo(\po, \ppol, \idp) \in M$.

 \item 
$\forall i: \semantics{\mc{R}_i}(\mc{Q}) \neq \p$ and
$\exists j: \semantics{\mc{R}_j}(\mc{Q}) = \ip$ and $\exists j': \semantics{\mc{R_{j'}}}(\mc{Q}) = \id$ where $\mc{R}_i$, $\mc{R}_j$ and $\mc{R}_{j'}$ are \Rule in the sequence inside \Policy\ \ppol.
Based on Prop. \ref{p:rule},  $\forall i: \val(\mc{R}_i, \p) \not \in M$  and $\exists j: \val(\mc{R}_j, \ip) \in M$ and $\exists j: \val(\mc{R}_{j'}, \id) \in M$.
Based on Lemma \ref{p:po1},  $\algo(\po, \ppol, \p) \not \in M$ since if it is in $M$ ,  there exists a \Rule $\mc{R}$ in the \Policy\ \ppol sequence such that $\semantics{\mc{R}}(\mc{Q}) = \p$.
Based on the \Policy transformation there are rules in \Prog in the form $\dec(\ppol, \mc{R}_j, \idp) \la \val(\mc{R}_j, \ip$ and $\dec(\ppol, \mc{R}_j, \ip) \la \val(\mc{R}_{j'}, \id).$
Thus, by Lemma \ref{l:1}, $\dec(\ppol, \mc{R}_j, \ip) \in M$  and $\dec(\ppol, \mc{R}_{j'}, \id) \in M$.
Hence, by Lemma \ref{l:1}, $\algo(\po, \ppol, \idp) \in M$.
\end{enumerate}
$(\Leftarrow)$
Suppose that $\algo(\po, \ppol, \idp) \in M$.
Based on Lemma \ref{l:3} ,  there is a rule where $\algo(\po, \ppol, \idp)$ as the head and the body is true under $M$. 
There are rules in $\Prog$ where $\algo(\po, \ppol, \idp)$ as the head, i.e., 
\begin{enumerate}
\item
$\algo(\po, \ppol, \idp) \la \mynot \algo(\po, \ppol, \p), \dec(\ppol, \mc{R}, \idp).$ \\
Then, $M(\mynot \algo(\po, \ppol, \p) \wedge \dec(\ppol, \mc{R}, \idp)) = \top$. 
Thus, $\algo(\po, \ppol, \p) \not \in M$ and $\dec(\ppol, \mc{R}, \idp) \in M$.
Based on Lemma \ref{p:po1}, $\bigoplus_{\po}(\mb{R}) \neq \p$. 
Based on Lemma \ref{l:3}, there is a rule where $\dec(\ppol, \mc{R}, \idp)$ as the head and the body is true under $M$. 
There is only one rule in $\Prog$, i.e., $\dec(\ppol, \mc{R}, \idp) \la \val(\mc{R}, \idp).$
Then, $M(\val(\mc{R}, \idp)) = \top$. 
Therefore, $\val(\mc{R}, \ip) \in M$.
As defined in \eqref{eq:po}, $\forall i: \semantics{\mc{R}_i}(\mc{Q}) \neq \p$ since $\bigoplus_{\po}(\mb{R}) \neq \p$ .
Based on Prop. \ref{p:rule},  $\semantics{\mc{R}}(\mc{Q}) = \idp$ and $\mc{R}$ belongs to the sequence inside \Policy\ \ppol. 
Hence, based on \eqref{eq:po}, we obtain $\bigoplus_{\po}(\mb{R}) = \idp$

 \item   
$\algo(\po, \ppol, \idp) \la \mynot \algo(\po, \ppol, \p), \dec(\ppol, \mc{R}, \ip), \dec(\ppol, \mc{R}', \d).$\\
Then, $M(\mynot \algo(\po, \ppol, \p) \wedge \dec(\ppol, \mc{R}, \ip) \wedge \dec(\mc{R}', \d)) = \top$. 
Thus, $\algo(\po, \ppol, \p) \not \in M$, $\dec(\ppol, \mc{R}, \ip) \in M$ and  $\dec(\ppol, \mc{R}, \d) \in M$.
Based on Lemma \ref{p:po1}, $\bigoplus_{\po}(\mb{R}) \neq \p$.
Based on Lemma \ref{l:3}, there is a rule where $\dec(\ppol, \mc{R}, \ip)$ as the head and the body is true under $M$. 
There is only one rule in $\Prog$, i.e., $\dec(\ppol, \mc{R}, \ip) \la \val(\mc{R}, \ip).$
Then, $M(\val(\mc{R}, \ip)) = \top$. 
Thus, $\val(\mc{R}, \ip) \in M$.
Based on Lemma \ref{l:3}, there is a rule where $\dec(\ppol, \mc{R}, \d)$ as the head and the body is true under $M$. 
There is only one rule in $\Prog$, i.e., $\dec(\ppol, \mc{R}, \ip) \la \val(\mc{R}', \d).$
Then, $M(\val(\mc{R}', \d)) = \top$. 
Thus, $\val(\mc{R}', \d) \in M$.
Based on \eqref{eq:po}, $\forall i: \semantics{\mc{R}_i}(\mc{Q}) \neq \p$ since $\bigoplus_{\po}(\mb{R}) \neq \p$ .
Based on Prop. \ref{p:rule}, $\semantics{\mc{R}}(\mc{Q}) = \ip$  and $\semantics{\mc{R}'}(\mc{Q}) = \d$ and $\mc{R}, \mc{R}'$ belongs to the sequence inside \Policy\ \ppol. 
Therefore, based on \eqref{eq:po} we obtain $\bigoplus_{\po}(\mb{R}) = \idp$

 \item   
$\algo(\po, \ppol, \idp) \la \mynot \algo(\po, \ppol, \p), \dec(\ppol, \mc{R}, \ip), \dec(\ppol, \mc{R}', \id).$ \\
Then, $M(\mynot \algo(\po, \ppol, \p) \wedge \dec(\ppol, \mc{R}, \ip) \wedge \dec(\mc{R}', \id)) = \top$. 
Thus, $\algo(\po, \ppol, \p) \not \in M$, $\dec(\ppol, \mc{R}, \ip) \in M$ and  $\dec(\ppol, \mc{R}, \id) \in M$.
Based on Lemma \ref{p:po1}, $\bigoplus_{\po}(\mb{R}) \neq \p$ since if $\bigoplus_{\po}(\mb{R}) = \p$. 
Based on Lemma \ref{l:3}, there is a rule where $\dec(\ppol, \mc{R}, \ip)$ as the head and the body is true under $M$. 
There is only one rule in $\Prog$, i.e.,  $\dec(\ppol, \mc{R}, \ip) \la \val(\mc{R}, \ip).$
Then,  $M(\val(\mc{R}, \ip)) = \top$. 
Therefore, $\val(\mc{R}, \ip) \in M$.
Based on Lemma \ref{l:3}, there is a rule where $\dec(\ppol, \mc{R}, \id)$ as the head and the body is true under $M$. 
There is only one rule in $\Prog$, i.e., $\dec(\ppol, \mc{R}, \ip) \la \val(\mc{R}', \id).$
Then,  $M(\val(\mc{R}', \d)) = \top$. 
Therefore, $\val(\mc{R}', \id) \in M$.
Based on \eqref{eq:po}, $\forall i: \semantics{\mc{R}_i}(\mc{Q}) \neq \p$ since $\bigoplus_{\po}(\mb{R}) \neq \p$ .
Based on Prop. \ref{p:rule}, $\semantics{\mc{R}}(\mc{Q}) = \ip$  and $\semantics{\mc{R}'}(\mc{Q}) = \id$ and $\mc{R}, \mc{R}'$ belongs to the sequence inside \Policy\ \ppol. 
Therefore, based on \eqref{eq:po}, we obtain $\bigoplus_{\po}(\mb{R}) = \idp$
\qed
\end{enumerate}
\end{proof}

\begin{lemma}
\label{p:po3}
Let $\pprog = \pprog_{\mc{Q}} \cup \pprog_{\po} \cup \ppolicy$  be a program obtained by merging \Request transformation program $\pprog_{\mc{Q}}$, permit-overrides combining algorithm transformation program $\pprog_{\po}$ and \Policy\ \ppol\ transformation program with its components \ppolicy. Let $M$ be an answer set of \pprog.
Then, \niceiffcomb{\po}{\ip}
where $\mb{R} = \seq{\semantics{\mc{R}_1}(\mc{Q}), \ldots, \semantics{\mc{R}_n}(\mc{Q})}$ be a sequence of policy value where each $\mc{R}_i$ is a \Rule in the sequence inside  \Policy\ \ppol.
\end{lemma}
\begin{proof}
$(\Rightarrow)$
Suppose that $\bigoplus_{\po}(\mb{R}) = \ip$ holds. 
Then, as defined in  \eqref{eq:po},  $\exists i: \semantics{\mc{R}_i}(\mc{Q}) = \ip$ and 
$\forall j: \semantics{\mc{R}_j}(\mc{Q}) \neq \ip \Rightarrow \semantics{\mc{R}_j}(\mc{Q}) = \na$   where $\mc{R}_i$ and $\mc{R}_j$ are \Rule in the sequence inside \Policy\ \ppol.
Based on Prop. \ref{p:rule}, $\exists i: \val(\mc{R}_j, \ip) \in M$.
Based on Lemma \ref{p:po1}, $\algo(\po, \ppol, \p) \not \in M$ since if it is in $M$ ,  there exists a \Rule $\mc{R}$ in the \Policy\ \ppol sequence such that $\semantics{\mc{R}}(\mc{Q}) = \p$.
Based on Lemma \ref{p:po2}, $\algo(\po, \ppol, \idp) \not \in M$ since if it is in $M$ ,  there exists a \Rule $\mc{R}$ in the \Policy\ \ppol sequence such that $\semantics{\mc{R}}(\mc{Q}) = \idp$, and  $\semantics{\mc{R}}(\mc{Q}) = \d \textrm{ or } \id$. 
Based on the \Policy transformation, there is a rule $\dec(\ppol, \mc{R}_i, \ip) \la \val(\mc{R}_i, \ip).$
Therefore, by Lemma \ref{l:1}, $\dec(\ppol, \mc{R}_i, \ip) \in M$.
Thus, by Lemma \ref{l:1}, $\algo(\po, \ppol, \ip) \in M$.
\\
$(\Leftarrow)$
Suppose that $\algo(\po, \ppol, \ip) \in M$.
Based on Lemma \ref{l:3}, there is a rule where $\algo(\po, \ppol, \ip)$ as the head and the body is true under $M$. 
There is only a rule in $\Prog$, i.e., 
$\algo(\po, \ppol, \idp) \la \mynot \algo(\po, \ppol, \p),$ $\mynot \algo(\po, \ppol, \idp), \dec(\ppol, \mc{R}, \idp).$
Then, $M(\mynot \algo(\po, \ppol, \p) \wedge \mynot \algo(\po, \ppol, \idp) \wedge \dec(\ppol, \mc{R}, \idp)) = \top$. 
Therefore, $\algo(\po, \ppol, \p) \not \in M$, $\algo(\po, \ppol, \idp) \not \in M$  and $\dec(\ppol, \mc{R}, \idp) \in M$.
Based on Lemma \ref{p:po1} and Lemma \ref{p:po2}, $\bigoplus_{\po}(\mb{R}) \neq \p$ and $\bigoplus_{\po}(\mb{R}) \neq \idp$.
Based on Lemma \ref{l:3}, ,  there is a rule where $\dec(\ppol, \mc{R}, \ip)$ as the head and the body is true under $M$. 
There is only one rule in $\Prog$, i.e., $\dec(\ppol, \mc{R}, \ip) \la \val(\mc{R}, \ip).$
Then, $M(\val(\mc{R}, \ip)) = \top$. 
Therefore, $\val(\mc{R}, \ip) \in M$.
Based on \eqref{eq:po}, $\forall i: \semantics{\mc{R}_i}(\mc{Q}) \neq \p$ since $\bigoplus_{\po}(\mb{R}) \neq \p$ and $\forall i: \semantics{\mc{R}_i}(\mc{Q}) \neq (\idp \textrm{ or } \d \textrm{ or } \id)$.  
Thus, the only possibilities of the value of $\semantics{\mc{R}_i}$ is either $\ip$ or $\na$. 
Based on Prop. \ref{p:rule}, $\semantics{\mc{R}}(\mc{Q}) = \ip$ and $\mc{R}$ belongs to the sequence inside \Policy\ \ppol. 
Therefore, based on \eqref{eq:po} we obtain $\bigoplus_{\po}(\mb{R}) = \ip$
\qed
\end{proof}

\begin{lemma}
\label{p:po4}
Let $\pprog = \pprog_{\mc{Q}} \cup \pprog_{\po} \cup \ppolicy$  be a program obtained by merging \Request transformation program $\pprog_{\mc{Q}}$, permit-overrides combining algorithm transformation program $\pprog_{\po}$ and \Policy\ \ppol\ transformation program with its components \ppolicy. Let $M$ be an answer set of \pprog.
Then, \niceiffcomb{\po}{\d}
where $\mb{R} = \seq{\semantics{\mc{R}_1}(\mc{Q}), \ldots, \semantics{\mc{R}_n}(\mc{Q})}$ be a sequence of policy value where each $\mc{R}_i$ is a \Rule in the sequence inside  \Policy\ \ppol.
\end{lemma}
\begin{proof}
$(\Rightarrow)$
Suppose that $\bigoplus_{\po}(\mb{R}) = \d$ holds. 
Then, as defined in  \eqref{eq:po}, $\exists i: \semantics{\mc{R}_i}(\mc{Q}) = \d$ and 
$\forall j: \semantics{\mc{R}_j}(\mc{Q}) \neq \d \Rightarrow \semantics{\mc{R}_j}(\mc{Q}) = (\id \textrm{ or } \na)$   where $\mc{R}_i$ and $\mc{R}_j$ are \Rule in the sequence inside \Policy\ \ppol.
Based on Prop. \ref{p:rule}, $\exists i: \val(\mc{R}_j, \d) \in M$.
Based on Lemma \ref{p:po1},  $\algo(\po, \ppol, \p) \not \in M$ since if it is in $M$ ,  there exists a \Rule $\mc{R}$ in the \Policy\ \ppol sequence such that $\semantics{\mc{R}}(\mc{Q}) = \p$.
Based on Lemma \ref{p:po2},  $\algo(\po, \ppol, \idp) \not \in M$ since if it is in $M$ ,  there exists a \Rule $\mc{R}$ in the \Policy\ \ppol sequence such that $\semantics{\mc{R}}(\mc{Q}) = \idp$.
Based on Lemma \ref{p:po3},  $\algo(\po, \ppol, \ip) \not \in M$ since if it is in $M$ ,  there exists a \Rule $\mc{R}$ in the \Policy\ \ppol sequence such that $\semantics{\mc{R}}(\mc{Q}) = \ip$.
Based on the \Policy transformation, there is a rule $\dec(\ppol, \mc{R}_i, \d) \la \val(\mc{R}_i, \d).$ 
Therefore, by Lemma \ref{l:1}, $\dec(\ppol, \mc{R}_i, \d) \in M$.
Thus,  by Lemma \ref{l:1}, $\algo(\po, \ppol, \id) \in M$.
\\
$(\Leftarrow)$
Suppose that $\algo(\po, \ppol, \d) \in M$.
Based on Lemma \ref{l:3},  there is a rule where $\algo(\po, \ppol, \d)$ as the head and the body is true under $M$. 
There is only a rule in $\Prog$, i.e., 
$\algo(\po, \ppol, \idp) \la \mynot \algo(\po, \ppol, \p),$ $\mynot \algo(\po, \ppol, \idp),  \mynot \algo(\po, \ppol, \ip),$ $\dec(\ppol, \mc{R}, \d).$
Hence, we obtain $M(\mynot \algo(\po, \ppol, \p) \wedge \mynot \algo(\po, \ppol, \idp) \wedge$ $\mynot \algo(\po, \ppol, \ip) \wedge \dec(\ppol, \mc{R}, \d)) = \top$. 
Therefore, $\algo(\po, \ppol, \p) \not \in M$, $\algo(\po, \ppol, \idp) \not \in M$, $\algo(\po, \ppol, \ip) \not \in M$   and $\dec(\ppol, \mc{R}, \d) \in M$.
Based on Lemma \ref{p:po1},  $\bigoplus_{\po}(\mb{R}) \neq \p$ since if $\bigoplus_{\po}(\mb{R}) = \p$   it will lead a contradiction.
Based on Lemma \ref{p:po2},  $\bigoplus_{\po}(\mb{R}) \neq \idp$ since if $\bigoplus_{\po}(\mb{R}) = \idp$   it will lead a contradiction.
Based on Lemma \ref{p:po3},  $\bigoplus_{\po}(\mb{R}) \neq \ip$ since if $\bigoplus_{\po}(\mb{R}) = \ip$   it will lead a contradiction.
Based on Lemma \ref{l:3},  there is a rule where $\dec(\ppol, \mc{R}, \id)$ as the head and the body is true under $M$. 
There is only one rule in $\Prog$, i.e., 
$\dec(\ppol, \mc{R}, \d) \la \val(\mc{R}, \d).$
Then,  $M(\val(\mc{R}, \d)) = \top$. 
Therefore, $\val(\mc{R}, \d) \in M$.
Based on \eqref{eq:po},  $\forall i: \semantics{\mc{R}_i}(\mc{Q}) \neq \p$ since $\bigoplus_{\po}(\mb{R}) \neq \p$ .
Based on \eqref{eq:po},  $\forall i: \semantics{\mc{R}_i}(\mc{Q}) \neq \idp$.  
Based on  \eqref{eq:po},  $\forall i: \semantics{\mc{R}_i}(\mc{Q}) \neq \ip$.
Thus, the only possibilities of the value of $\semantics{\mc{R}_i}$ is either $\d$, $\id$ or $\na$. 
Based on Prop. \ref{p:rule} ,  $\semantics{\mc{R}}(\mc{Q}) = \d$ and $\mc{R}$ belongs to the sequence inside \Policy\ \ppol. 
Therefore, based on \eqref{eq:po},  $\bigoplus_{\po}(\mb{R}) = \d$
\qed
\end{proof}

\begin{lemma}
\label{p:po5}
Let $\pprog = \pprog_{\mc{Q}} \cup \pprog_{\po} \cup \ppolicy$  be a program obtained by merging \Request transformation program $\pprog_{\mc{Q}}$, permit-overrides combining algorithm transformation program $\pprog_{\po}$ and \Policy\ \ppol\ transformation program with its components \ppolicy. Let $M$ be an answer set of \pprog.
Then, \niceiffcomb{\po}{\id}
where $\mb{R} = \seq{\semantics{\mc{R}_1}(\mc{Q}), \ldots, \semantics{\mc{R}_n}(\mc{Q})}$ be a sequence of policy value where each $\mc{R}_i$ is a \Rule in the sequence inside  \Policy\ \ppol.
\end{lemma}
\begin{proof}
$(\Rightarrow)$
Suppose that $\bigoplus_{\po}(\mb{R}) = \id$ holds. 
Then, as defined in  \eqref{eq:po}  $\exists i: \semantics{\mc{R}_i}(\mc{Q}) = \id$ and 
$\forall j: \semantics{\mc{R}_j}(\mc{Q}) \neq \d \Rightarrow \semantics{\mc{R}_j}(\mc{Q}) = \na$   where $\mc{R}_i$ and $\mc{R}_j$ are \Rule in the sequence inside \Policy\ \ppol.
Based on Prop. \ref{p:rule} ,  $\exists i: \val(\mc{R}_j, \d) \in M$.
Based on Lemma \ref{p:po1},  $\algo(\po, \ppol, \p) \not \in M$ since if it is in $M$ ,  there exists a \Rule $\mc{R}$ in the \Policy\ \ppol sequence such that $\semantics{\mc{R}}(\mc{Q}) = \p$.
Based on Lemma \ref{p:po2},  $\algo(\po, \ppol, \idp) \not \in M$ since if it is in $M$ ,  there exists a \Rule $\mc{R}$ in the \Policy\ \ppol sequence such that $\semantics{\mc{R}}(\mc{Q}) = \idp$.
Based on Lemma \ref{p:po3},  $\algo(\po, \ppol, \ip) \not \in M$ since if it is in $M$ ,  there exists a \Rule $\mc{R}$ in the \Policy\ \ppol sequence such that $\semantics{\mc{R}}(\mc{Q}) = \ip$.
Based on Lemma \ref{p:po4},  $\algo(\po, \ppol, \d) \not \in M$ since if it is in $M$ ,  there exists a \Rule $\mc{R}$ in the \Policy\ \ppol sequence such that $\semantics{\mc{R}}(\mc{Q}) = \d$.
Based on the \Policy transformation, there is a rule $\dec(\ppol, \mc{R}_i, \id) \la \val(\mc{R}_i, \id).$ 
Hence, by Lemma \ref{l:1}, $\dec(\ppol, \mc{R}_i, \id) \in M$.
Thus, by Lemma \ref{l:1}, $\algo(\po, \ppol, \id) \in M$. 
\\
$(\Leftarrow)$
Suppose that $\algo(\po, \ppol, \d) \in M$.
Based on Lemma \ref{l:3},  there is a rule where $\algo(\po, \ppol, \d)$ as the head and the body is true under $M$. 
There is only a rule in $\Prog$, i.e., 
$\algo(\po, \ppol, \d) \la \mynot \algo(\po, \ppol, \p),$ $\mynot \algo(\po, \ppol, \idp),  \mynot \algo(\po, \ppol, \ip),$ $\mynot \algo(\po, \ppol, \d), \dec(\ppol, \mc{R}, \id).$
Hence, we find that $M(\mynot \algo(\po, \ppol, \p) \wedge \mynot \algo(\po, \ppol, \idp) \wedge  \mynot \algo(\po, \ppol, \ip)  \wedge  \mynot \algo(\po, \ppol, \d) \wedge \dec(\ppol, \mc{R}, \id)) = \top$. 
Thus, $\algo(\po, \ppol, \p) \not \in M$, $\algo(\po, \ppol, \idp) \not \in M$, $\algo(\po, \ppol, \ip) \not \in M$, $\algo(\po, \ppol, \d)$ $\not \in M$   and $\dec(\ppol, \mc{R}, \id) \in M$.
Based on Lemma \ref{p:po1},  $\bigoplus_{\po}(\mb{R}) \neq \p$ since if $\bigoplus_{\po}(\mb{R}) = \p$   it will lead a contradiction.
Based on Lemma \ref{p:po2},  $\bigoplus_{\po}(\mb{R}) \neq \idp$ since if $\bigoplus_{\po}(\mb{R}) = \idp$   it will lead a contradiction.
Based on Lemma \ref{p:po3},  $\bigoplus_{\po}(\mb{R}) \neq \ip$ since if $\bigoplus_{\po}(\mb{R}) = \ip$   it will lead a contradiction.
Based on Lemma \ref{p:po4},  $\bigoplus_{\po}(\mb{R}) \neq \ip$ since if $\bigoplus_{\po}(\mb{R}) = \d$   it will lead a contradiction.
Based on Lemma \ref{l:3},  there is a rule where $\dec(\ppol, \mc{R}, \d)$ as the head and the body is true under $M$. 
There is only one rule in $\Prog$, i.e., 
$\dec(\ppol, \mc{R}, \id) \la \val(\mc{R}, \id).$
Then, we find that $M(\val(\mc{R}, \id)) = \top$. 
Therefore, $\val(\mc{R}, \id) \in M$.
Based on \eqref{eq:po},  $\forall i: \semantics{\mc{R}_i}(\mc{Q}) \neq \p$ since $\bigoplus_{\po}(\mb{R}) \neq \p$ .
Based on eqref{eq:po},  $\forall i: \semantics{\mc{R}_i}(\mc{Q}) \neq \idp$.  
Based on \eqref{eq:po} ,  $\forall i: \semantics{\mc{R}_i}(\mc{Q}) \neq \ip$ and $\forall i: \semantics{\mc{R}_i}(\mc{Q}) \neq \d$.
Thus, the only possibilities of the value of $\semantics{\mc{R}_i}$ is either $\id$ or $\na$. 
Based on Prop. \ref{p:rule},  $\semantics{\mc{R}}(\mc{Q}) = \id$ and $\mc{R}$ belongs to the sequence inside \Policy\ \ppol. 
Therefore, based on \eqref{eq:po},  $\bigoplus_{\po}(\mb{R}) = \id$
\qed
\end{proof}

\begin{lemma}
\label{p:po6}
Let $\pprog = \pprog_{\mc{Q}} \cup \pprog_{\po} \cup \ppolicy$  be a program obtained by merging \Request transformation program $\pprog_{\mc{Q}}$, permit-overrides combining algorithm transformation program $\pprog_{\po}$ and \Policy\ \ppol\ transformation program with its components \ppolicy. Let $M$ be an answer set of \pprog.
Then, \niceiffcomb{\po}{\na}
where $\mb{R} = \seq{\semantics{\mc{R}_1}(\mc{Q}), \ldots, \semantics{\mc{R}_n}(\mc{Q})}$ be a sequence of policy value where each $\mc{R}_i$ is a \Rule in the sequence inside  \Policy\ \ppol.
\end{lemma}

\begin{proof}
$(\Rightarrow)$
Suppose that $\bigoplus_{\po}(\mb{R}) = \na$ holds. 
Then, as defined in  \eqref{eq:po} we have that
$\bigoplus_{\po}(\mb{R}) \neq \p$, 
$\bigoplus_{\po}(\mb{R}) \neq \idp$,
$\bigoplus_{\po}(\mb{R}) \neq \ip$,
$\bigoplus_{\po}(\mb{R}) \neq \d$, and
$\bigoplus_{\po}(\mb{R}) \neq \id$.
Based on Lemma \ref{p:po1},  $\algo(\po, \ppol, \p) \not \in M$.
Based on Lemma \ref{p:po2},  $\algo(\po, \ppol, \idp) \not \in M$.
Based on Lemma \ref{p:po3} ,  $\algo(\po, \ppol, \ip) \not \in M$.
Based on Lemma \ref{p:po4} ,  $\algo(\po, \ppol, \d) \not \in M$.
Based on Lemma \ref{p:po5} ,  $\algo(\po, \ppol, \id) \not \in M$.
Thus,  $M(\mynot \algo(\po, \ppol, \p) \wedge \mynot \algo(\po, \ppol, \idp)$ $\wedge \mynot \algo(\po, \ppol, \ip) \wedge \mynot \algo(\po, \ppol, \d) \wedge \mynot \algo(\po, \ppol, \id)) = \top$
Therefore, by Lemma \ref{l:1}, $\algo(\po, \ppol, \na) \in M$.
\\
$(\Leftarrow)$
Suppose that $\algo(\po, \ppol, \na) \in M$.
Based on Lemma \ref{l:3} ,  there is a rule where $\algo(\po, \ppol, \na)$ as the head and the body is true under $M$. 
There is only a rule in $\Prog$ where $\algo(\po, \ppol, \na)$ as the head, i.e., 
$\algo(\po, \ppol, \na) \la \mynot \algo(\po, \ppol, \p),$ $\mynot \algo(\po, \ppol, \idp),  \mynot \algo(\po, \ppol, \ip), \mynot \algo(\po, \ppol, \d), \mynot \algo(\po, \ppol, \id).$
Then, $M(\mynot \algo(\po, \ppol, \p) \wedge \mynot \algo(\po, \ppol, \idp) \wedge  \mynot \algo(\po, \ppol, \ip)  \wedge  \mynot \algo(\po, \ppol, \d) \wedge \mynot \algo(\po, \ppol, \id)) = \top$. 
Therefore, $\algo(\po, \ppol, \p) \not \in M$, $\algo(\po, \ppol, \idp) \not \in M$, $\algo(\po, \ppol, \ip) \not \in M$, $\algo(\po, \ppol, \d) \not \in M$   and $\algo(\po, \ppol, \id) \not \in M$.
Based on Lemma \ref{p:po1},  $\bigoplus_{\po}(\mb{R}) \neq \p$.
Based on Lemma \ref{p:po2},  $\bigoplus_{\po}(\mb{R}) \neq \idp$. 
Based on Lemma \ref{p:po3},  $\bigoplus_{\po}(\mb{R}) \neq \ip$. 
Based on Lemma \ref{p:po4},  $\bigoplus_{\po}(\mb{R}) \neq \ip$. 
Based on Lemma \ref{p:po5},  $\bigoplus_{\po}(\mb{R}) \neq \ip$.
Therefore, based on \eqref{eq:po},  $\bigoplus_{\po}(\mb{R}) = \na$.
\qed
\end{proof}

\begin{proposition}
\label{p:po}
Let $\pprog = \pprog_{\mc{Q}} \cup \pprog_{\po} \cup \ppolicy$  be a program obtained by merging \Request transformation program $\pprog_{\mc{Q}}$, permit-overrides combining algorithm transformation program $\pprog_{\po}$ and \Policy\ \ppol\ transformation program with its components \ppolicy. Let $M$ be an answer set of \pprog.
Then, \niceiffcomb{\po}{V}
where $\mb{R} = \seq{\mc{R}_1(\mc{Q}), \ldots, \mc{R}_n(\mc{Q})}$ be a sequence of policy value where each $\mc{R}_i$ is a \Rule in the sequence inside  \Policy\ \ppol.
\end{proposition}

\begin{proof}
It follows from Lemma \ref{p:po1}, Lemma \ref{p:po2}, Lemma \ref{p:po3}, Lemma \ref{p:po4}, Lemma \ref{p:po5} and Lemma \ref{p:po6} since the value of $V$ only has six possibilities, i.e., $\Set{\p, \d, \ip, \id, \idp, \na}$.
\qed
\end{proof}


\vspace{5pt}\noindent\textbf{Combining Algorithm: First Applicable.}
\begin{proposition}
\label{p:fa}
Let $\pprog = \pprog_{\mc{Q}} \cup \pprog_{\fa} \cup \ppolicy$  be a program obtained by merging \Request transformation program $\pprog_{\mc{Q}}$, first-applicable combining algorithm transformation program $\pprog_{\fa}$ and \Policy\ \ppol\ transformation program with its components \ppolicy. Let $M$ be an answer set of \pprog.
Then, \niceiffcomb{\fa}{V}
where $\mb{R} = \seq{\semantics{\mc{R}_1}(\mc{Q}), \ldots, \semantics{\mc{R}_n}(\mc{Q})}$ be a sequence of policy value where each $\mc{R}_i$ is a \Rule in the sequence inside  \Policy\ \ppol.
\end{proposition}
\begin{proof}
$(\Rightarrow)$
Suppose that $\bigoplus_{\fa}(\mb{R}) = V$ holds. 
Then, as defined in  \eqref{eq:fa}, $\exists i: \semantics{\mc{R}_i}(\mc{Q}) = V$ and $ V\neq \na$ and $\forall j: j < i \Rightarrow \semantics{\mc{R}_j}(\mc{Q}) = \na$. 
Based on Prop. \ref{p:rule}, $\exists i: \val(\mc{R}_i, V) \in M$  where $V \neq \na$ and $\forall j: j < i \Rightarrow \val(\mc{R}_j, \na) \in M$. 
Based on the \Policy transformation there is a rule $\dec(\ppol, \mc{R}_i, V) \la \val(\mc{R}_i, V)$ in \Prog and $\forall j: j < i$ we get that there are rules in the form $\dec(\ppol, \mc{R}_j, \na) \la \val(\mc{R}_j, \na)$ in \Prog.
Therefore, we have $\dec(\ppol, \mc{R}_i, V) \in M$ and $\forall j: j < i$ we also have $\dec(\ppol, \mc{R}_j, \na) \in M$ since  $M$ is a minimal model for \Prog.
Thus, by Lemma \ref{l:1}, $\algo(\fa, \ppol, V) \in M$.
\\
$(\Leftarrow)$
Suppose that $\algo(\fa, \ppol, V) \in M$.
Based on Lemma \ref{l:3}, there is a clause in $\mc{P}$ where $\algo(\fa, \ppol, V)$ as the head and the body is true under $M$.
There are several rules in $\mc{P}$ where $\algo(\fa, \ppol, V)$ as the head. 
We can see that in each rule the body contains $\exists i: \dec(\ppol, \mc{R}_i, V), V \neq \na$ and $\forall j: j < i$ the body also contains $\dec(\ppol, \mc{R}_j, \na)$.
Therefore $\exists i: \dec(\ppol, \mc{R}_i, V) \in M$ and $\forall j: j < i$, $\dec(\ppol, \mc{R}_j, \na) \in M$. 
Based on Lemma \ref{l:3}, there is a clause where $\dec(\ppol, \mc{R}_i, V)$ as the head and the body is true under $M$ and  $\forall j: j < i$, there is  $\dec(\ppol, \mc{R}_j, \na)$ as the head and the body is true under $M$. 
There is only one rule in $\mc{P}$ where  $\dec(\ppol, \mc{R}_i, V)$ as the head, i.e., 
$\dec(\ppol, \mc{R}_i, V) \la \val(\mc{R}_i, V).$ The same case for $\dec(\ppol, \mc{R}_j, \na)$.
Then. $\exists i: M(\val(\mc{R}_i, V)) = \top$ and $\forall j < i:  M(\val(\mc{R}_j, \na)) = \top$. 
Therefore, $ \exists i: \val(\mc{R}_i, V) \in M$ and $\forall j: j< i \Rightarrow \val(\mc{R}_j, \na) \in M$.
Based on Prop. \ref{p:rule}, $\exists i: \semantics{\mc{R}_i}(\mc{Q}) = V$  and $\forall j: j< i \Rightarrow \semantics{\mc{R}_j}(\mc{Q}) = \na$ and $\mc{R}_i$ and $\mc{R}_j$ belong to the sequence inside \Policy\ \ppol. 
Therefore, based on \eqref{eq:po} we obtain $\bigoplus_{\fa}(\mb{R}) = V$.   
\qed
\end{proof}

\vspace{5pt}\noindent\textbf{Combining Algorithm: Only-One-Applicable.}
\begin{lemma}
\label{p:ooa1}
Let $\pprog = \pprog_{\mc{Q}} \cup \pprog_{\ooa} \cup \ppolicy$  be a program obtained by merging \Request transformation program $\pprog_{\mc{Q}}$, only-one-applicable combining algorithm transformation program $\pprog_{\ooa}$ and \Policy\ \ppol\ transformation program with its components \ppolicy. Let $M$ be an answer set of \pprog.
Then, \niceiffcomb{\ooa}{\idp}
where $\mb{R} = \seq{\mc{R}_1(\mc{Q}), \ldots, \mc{R}_n(\mc{Q})}$ be a sequence of policy value where each $\mc{R}_i$ is a \Rule in the sequence inside  \Policy\ \ppol.
\end{lemma}
\begin{proof}
$(\Rightarrow)$
Suppose that $\bigoplus_{\ooa}(\mb{R}) = \idp$ holds.
Then, as defined in  \eqref{eq:ooa}, we have that 
\begin{enumerate}
\item $\exists i : \semantics{\mc{R}_i}(\mc{Q}) = \idp$.\\
Based on Prop. \ref{p:rule}, $\val(\mc{R}_i, \idp) \in M$.
Based on the \Policy transformation, there is a rule $\dec(\ppol, \mc{R}_i, \idp) \la \val(\mc{R}_i, \idp).$
Therefore, by Lemma \ref{l:1},  we find that $\dec(\ppol, \mc{R}_i, \idp)$ $\in M$.
Thus, by Lemma \ref{l:1}, $\algo(\ooa, \ppol, \idp) \in M$.

\item $\exists i, j: \semantics{\mc{R}_i}(\mc{Q}) = \d$ and $\semantics{\mc{R}_j}(\mc{Q}) = \p$.\\
Based on Prop. \ref{p:rule}, $\val(\mc{R}_i, \d) \in M$ and $\val(\mc{R}_j, \p) \in M$.
Based on the \Policy transformation, there are rules in \Prog with the form $\dec(\ppol, \mc{R}_i, \d) \la \val(\mc{R}_i, \d)$ and $\dec(\ppol, \mc{R}_j, \p)$ $\la \val(\mc{R}_i, \p)$.
Therefore, by Lemma \ref{l:1}, we find that $\dec(\ppol, \mc{R}_i, \d) \in M$ and $\dec(\ppol, \mc{R}_i, \p) \in M$.
Thus, by Lemma \ref{l:1}, $\algo(\ooa, \ppol, \idp) \in M$.

\item $\exists i, j: \semantics{\mc{R}_i}(\mc{Q}) = \id$ and $\semantics{\mc{R}_j}(\mc{Q}) = \p$.\\
Based on Prop. \ref{p:rule}, $\val(\mc{R}_i, \id) \in M$ and $\val(\mc{R}_j, \p) \in M$.
Based on the \Policy transformation there are rules in \Prog in the form $\dec(\ppol, \mc{R}_i, \d) \la \val(\mc{R}_i, \id).$ and $\dec(\ppol, \mc{R}_j, \p) \la \val(\mc{R}_i, \p).$
Then, by Lemma \ref{l:1}, $\dec(\ppol, \mc{R}_i, \id)$ $\in M$ and $\dec(\ppol, \mc{R}_i, \p) \in M$.
Thus, by Lemma \ref{l:1},  $\algo(\ooa, \ppol, \idp) \in M$.

\item $\exists i, j: \semantics{\mc{R}_i}(\mc{Q}) = \d$ and $\semantics{\mc{R}_j}(\mc{Q}) = \ip$.\\
Based on Prop. \ref{p:rule}, $\val(\mc{R}_i, \d) \in M$ and $\val(\mc{R}_j, \ip) \in M$.
Based on the \Policy transformation, there are rules in \Prog in the form $\dec(\ppol, \mc{R}_i, \d) \la \val(\mc{R}_i, \d).$ and $\dec(\ppol, \mc{R}_j, \ip) \la \val(\mc{R}_i, \ip).$
Then, by Lemma \ref{l:1},  $\dec(\ppol, \mc{R}_i, \d) \in M$ and $\dec(\ppol, \mc{R}_i, \ip) \in M$.
Thus, by Lemma \ref{l:1},  $\algo(\ooa, \ppol, \idp) \in M$.

\item $\exists i, j: \semantics{\mc{R}_i}(\mc{Q}) = \id$ and $\semantics{\mc{R}_j}(\mc{Q}) = \ip$.\\
Based on Prop. \ref{p:rule}, $\val(\mc{R}_i, \id) \in M$ and $\val(\mc{R}_j, \ip) \in M$.
Based on the \Policy transformation there are rules in \Prog in the form $\dec(\ppol, \mc{R}_i, \id) \la \val(\mc{R}_i, \id).$ and $\dec(\ppol, \mc{R}_j, \ip) \la \val(\mc{R}_i, \ip).$
Then, by Lemma \ref{l:1},  $\dec(\ppol, \mc{R}_i, \id) \in M$ and $\dec(\ppol, \mc{R}_i, \ip) \in M$. 
Thus by Lemma \ref{l:1},  $\algo(\ooa, \ppol, \idp) \in M$.
\end{enumerate}
$(\Leftarrow)$
Suppose that $\algo(\ooa, \ppol, \idp) \in M$. 
Based on Lemma \ref{l:3}, there is a rule where $\algo(\ooa, \ppol, \idp)$ as the head and the body is true under $M$. 
There are five rules in $\Prog$, i.e., 
\begin{enumerate}
\item $\algo(\ooa, \ppol, \idp) \la \dec(\ppol, \mc{R}, \idp).$ \\
Then, $M(\dec(\ppol, \mc{R}, \idp)) = \top$.
Therefore, $\dec(\ppol, \mc{R}, \idp) \in M$.
Based on Lemma \ref{l:3}, there is a rule where $\dec(\ppol, \mc{R}, \idp) \in M$ as the head and the body is true under $M$. 
There is only one rule in $\Prog$, i.e., 
$\dec(\ppol, \mc{R}, \idp)$ $\la \val(\mc{R}, \idp).$
Then,  $M(\val(\mc{R}, \idp)) = \top$. 
Therefore, $\val(\mc{R}, \idp) \in M$.
Based on Prop. \ref{p:rule}, $\semantics{\mc{R}}(\mc{Q}) = \idp$ and $\mc{R}$ belongs to the sequence inside \Policy\ \ppol. 
Therefore, based on  \eqref{eq:ooa}, $\bigoplus_{\ooa}(\mb{R}) = \idp$   
\item $\algo(\ooa, \ppol, \idp) \la \dec(\ppol, \mc{R}1, \id), \dec(\ppol, \mc{R}2, \ip).$ \\
Then, $M(\dec(\ppol, \mc{R}1, \id) \wedge \dec(\ppol, \mc{R}2), \ip) = \top$.
Hence, we find that $\dec(\ppol, \mc{R}1, \id) \in M$ and $\dec(\ppol, \mc{R}2, \ip) \in M$.
Based on Lemma \ref{l:3}, there is a rule where $\dec(\ppol, \mc{R}1, \id) \in M$ as the head and the body is true under $M$, i.e., $\dec(\ppol, \mc{R}1, \id) \la \val(\mc{R}1, \id)$,  and there is a rule where $\dec(\ppol, \mc{R}2, \ip) \in M$ as the head and the body is true under $M$, i.e., $\dec(\ppol, \mc{R}2, \ip) \la \val(\mc{R}2, \ip).$
Therefore, $M(\val(\mc{R}1, \id)) = \top$ and $M(\val(\mc{R}2, \ip)) = \top$ . 
Then, $\val(\mc{R}1, \id) \in M$ and $\val(\mc{R}2, \ip) \in M$.
Based on Prop. \ref{p:rule}, $\semantics{\mc{R}1}(\mc{Q}) = \id$ and $\semantics{\mc{R}2}(\mc{Q}) = \ip$  and $\mc{R}1$ and $\mc{R}2$ belong to the sequence inside \Policy\ \ppol. 
Therefore, based on  \eqref{eq:ooa}, $\bigoplus_{\ooa}(\mb{R}) = \idp$.
\item $\algo(\ooa, \ppol, \idp) \la \dec(\ppol, \mc{R}1, \id), \dec(\ppol, \mc{R}2, \p).$\\
Then, we find that $M(\dec(\ppol, \mc{R}1, \id) \wedge \dec(\ppol, \mc{R}2), \p) = \top$.
Therefore, $\dec(\ppol, \mc{R}1, \id) \in M$ and $\dec(\ppol, \mc{R}2, \p) \in M$.
Based on Lemma \ref{l:3}, there is a rule where $\dec(\ppol, \mc{R}1, \id) \in M$ as the head and the body is true under $M$, i.e., $\dec(\ppol, \mc{R}1, \id) \la \val(\mc{R}1, \id)$ and there is a rule where $\dec(\ppol, \mc{R}2, \p) \in M$ as the head and the body is true under $M$, i.e., $\dec(\ppol, \mc{R}2, \p) \la \val(\mc{R}2, \p).$
Then, we find that $M(\val(\mc{R}1, \id)) = \top$ and $M(\val(\mc{R}2, \p)) = \top$ . 
Therefore, $\val(\mc{R}1, \id) \in M$ and $\val(\mc{R}2, \p) \in M$.
Based on Prop. \ref{p:rule}, $\semantics{\mc{R}1}(\mc{Q}) = \id$ and $\semantics{\mc{R}2}(\mc{Q}) = \p$  and $\mc{R}1$ and $\mc{R}2$ belong to the sequence inside \Policy\ \ppol. 
Therefore, based on  \eqref{eq:ooa}, $\bigoplus_{\ooa}(\mb{R}) = \idp$ 
\item $\algo(\ooa, \ppol, \idp) \la \dec(\ppol, \mc{R}1, \d), \dec(\ppol, \mc{R}2, \ip).$\\
Then we find that $M(\dec(\ppol, \mc{R}1, \d) \wedge \dec(\ppol, \mc{R}2), \ip) = \top$.
Therefore, $\dec(\ppol, \mc{R}1, \d) \in M$ and $\dec(\ppol, \mc{R}2, \ip) \in M$.
Based on Lemma \ref{l:3}, there is a rule where $\dec(\ppol, \mc{R}1, \d) \in M$ as the head and the body is true under $M$, i.e.,  $\dec(\ppol, \mc{R}1, \d) \la \val(\mc{R}1, \d)$ and there is a rule where $\dec(\ppol, \mc{R}2, \ip) \in M$ as the head and the body is true under $M$, i.e., $\dec(\ppol, \mc{R}2, \ip)$ $\la \val(\mc{R}2, \ip)$.
Then, we find that $M(\val(\mc{R}1, \d)) = \top$ and $M(\val(\mc{R}2, \ip)) = \top$. 
Therefore, $\val(\mc{R}1, \d) \in M$ and $\val(\mc{R}2, \ip) \in M$.
Based on Prop. \ref{p:rule}, $\semantics{\mc{R}1}(\mc{Q}) = \d$ and $\semantics{\mc{R}2}(\mc{Q}) = \ip$  and $\mc{R}1$ and $\mc{R}2$ belong to the sequence inside \Policy\ \ppol. 
Therefore, based on  \eqref{eq:ooa}, $\bigoplus_{\ooa}(\mb{R}) = \idp$ 
\item $\algo(\ooa, \ppol, \idp) \la \dec(\ppol, \mc{R}1, \d), \dec(\ppol, \mc{R}2, \p).$ \\
Then we find that $M(\dec(\ppol, \mc{R}1, \d) \wedge \dec(\ppol, \mc{R}2), \p) = \top$.
Therefore, $\dec(\ppol, \mc{R}1, \d) \in M$ and $\dec(\ppol, \mc{R}2, \p) \in M$.
Based on Lemma \ref{l:3}, there is a rule where $\dec(\ppol, \mc{R}1, \d) \in M$ as the head and the body is true under $M$, i.e., 
$\dec(\ppol, \mc{R}1, \d) \la \val(\mc{R}1, \d)$  and there is a rule where $\dec(\ppol, \mc{R}2, \p) \in M$ as the head and the body is true under $M$, i.e., 
$\dec(\ppol, \mc{R}2, \p)$ $\la \val(\mc{R}2, \p).$
Then, we find that $M(\val(\mc{R}1, \d)) = \top$ and $M(\val(\mc{R}2, \p)) = \top$. 
Therefore, $\val(\mc{R}1, \d) \in M$ and $\val(\mc{R}2, \p) \in M$.
Based on Prop. \ref{p:rule}, $\semantics{\mc{R}1}(\mc{Q}) = \d$ and $\semantics{\mc{R}2}(\mc{Q}) = \p$  and $\mc{R}1$ and $\mc{R}2$ belong to the sequence inside \Policy\ \ppol. 
Therefore, based on  \eqref{eq:ooa}, $\bigoplus_{\ooa}(\mb{R}) = \idp$ 
\qed
\end{enumerate}
\end{proof}

\begin{lemma}
\label{p:ooa2}
Let $\pprog = \pprog_{\mc{Q}} \cup \pprog_{\ooa} \cup \ppolicy$  be a program obtained by merging \Request transformation program $\pprog_{\mc{Q}}$, only-one-applicable combining algorithm transformation program $\pprog_{\ooa}$ and \Policy\ \ppol\ transformation program with its components \ppolicy. Let $M$ be an answer set of \pprog.
Then, \niceiffcomb{\ooa}{\id}
where $\mb{R} = \seq{\mc{R}_1(\mc{Q}), \ldots, \mc{R}_n(\mc{Q})}$ be a sequence of policy value where each $\mc{R}_i$ is a \Rule in the sequence inside  \Policy\ \ppol.
\end{lemma}
\begin{proof}
Suppose that $\bigoplus_{\ooa}(\mb{R}) = \id$ holds.
Then, as defined in  \eqref{eq:ooa}, we have that 
\begin{enumerate}
\item $\forall i: \semantics{\mc{R}_i}(\mc{Q}) \neq (\p \textrm{ or } \ip \textrm{ or } \idp)$ and $\exists j: s_j = \id$ \\
Based on Prop. \ref{p:rule}, $\forall i: \val(\mc{R}_i, \p) \not \in M$, $ \val(\mc{R}_i, \ip) \not \in M$ and $ \val(\mc{R}_i, \idp) \not \in M$.
Based on Prop. \ref{p:rule}, $\exists j: \val(\mc{R}_j, \id) \in M$. 
Based on Lemma \ref{p:ooa1}, $\algo(\ooa, \ppol, \idp) \not \in M$ since if it is in $M$, there is a \Rule $\mc{R}$ such that $\semantics{\mc{R}}(\mc{Q}) =  (\p \textrm{ or } \ip \textrm{ or } \idp)$
Based on the \Policy transformation there is a rule $\dec(\ppol, \mc{R}_j, \id) \la \val(\mc{R}_j, \id).$
Then, by Lemma \ref{l:1}, $\dec(\ppol, \mc{R}_j, \id) \in M$.
Then, by Lemma \ref{l:1}, we obtain $\algo(\ooa, \ppol, \id) \in M$. 

\item $\forall i: \semantics{\mc{R}_i}(\mc{Q}) \neq (\p \textrm{ or } \ip \textrm{ or } \idp)$ and $\exists j, k: j \neq k \textrm{ and } s_j = s_k = \d$ 
Based on Prop. \ref{p:rule}, $\forall i: \val(\mc{R}_i, \p) \not \in M$, $ \val(\mc{R}_i, \ip) \not \in M$ and $ \val(\mc{R}_i, \idp) \not \in M$ since if they are in $M$ it will lead a contradiction. 
Based on Prop. \ref{p:rule}, $\exists j, k: j \neq k$ and $\val(\mc{R}_j, \d) \in M$ and $\val(\mc{R}_jk \d) \in M$. 
Based on Lemma \ref{p:ooa1}, $\algo(\ooa, \ppol, \idp) \not \in M$ since if it is in $M$, there is a \Rule $\mc{R}$ such that $\semantics{\mc{R}}(\mc{Q}) =  (\p \textrm{ or } \ip \textrm{ or } \idp)$
Based on the \Policy transformation there is a rule $\dec(\ppol, \mc{R}_j, \d) \la \val(\mc{R}_j, \d).$ and there is a rule $\dec(\ppol, \mc{R}_k,  \d) \la \val(\mc{R}_k, \d).$
Therefore $\dec(\ppol, \mc{R}_j, \d) \in M$ and $\dec(\ppol, \mc{R}_k, \d) \in M$ since $M$ is the minimal model of $\Prog$.
Thus, $\algo(\ooa, \ppol, \id) \in M$ since $M$ is the minimal model of $\Prog$. 
\end{enumerate}
$(\Leftarrow)$
Suppose that $\algo(\ooa, \ppol, \id) \in M$. 
Based on Lemma \ref{l:3}, there is a rule where $\algo(\ooa, \ppol, \id)$ as the head and the body is true under $M$. 
There are rules in $\Prog$ where $\algo(\ooa, \ppol, \id)$ as the head, i.e.,
\begin{enumerate}
\item 
$\algo(\ooa, \ppol, \id) \la \mynot \algo(\ooa, \ppol, \idp), \dec(\ppol, \mc{R}, \id).$\\
Then we find that $M(\mynot \algo(\ooa, \ppol, \idp) \wedge \dec(\ppol, \mc{R}, \id)) = \top$.
Therefore, $\algo(\ooa, \ppol, \idp) \not \in M$ and $\dec(\ppol, \mc{R}, \id) \in M$.
Based on Lemma \ref{p:ooa1}, $\bigoplus_{\ooa}(\mb{R}) \neq \idp$. 
Based on Lemma \ref{l:3}, there is a rule where $\dec(\ppol, \mc{R}, \id)$ as the head and the body is true under $M$, i.e., $\dec(\ppol, \mc{R}, \id) \la \val(\mc{R}, \id).$
Then we find that $M(\val(\mc{R}, \id)) = \top$.
Therefore, $\val(\mc{R}, \id) \in M$.
Based on Prop. \ref{p:rule}, $\semantics{\mc{R}}(\mc{Q}) = \id$.
Based on  \eqref{eq:ooa}, $\forall i: \semantics{\mc{R}_i}(\mc{Q}) \neq (\p \textrm{ or } \ip \textrm{ or } \idp)$ since $\bigoplus_{\ooa}(\mb{R}) \neq \idp$.
Therefore, based on  \eqref{eq:ooa}, $\bigoplus_{\ooa}(\mb{R}) = \id$.
\item
$\algo(\ooa, \ppol, \id) \la \mynot \algo(\ooa, \ppol, \idp), \dec(\ppol, \mc{R}1, \d),\\ \dec(\ppol, \mc{R}2, \d),\mc{R}1 \neq \mc{R}2.$\\
Then, $M(\mynot \algo(\ooa, \ppol, \idp) \wedge \dec(\ppol, \mc{R}1, \d) \wedge \dec(\ppol, \mc{R}2, \d)) = \top$, $\mc{R}1 \neq \mc{R}2$.
Therefore, $\algo(\ooa, \ppol, \idp) \not \in M$ and $\dec(\ppol, \mc{R}1, \d) \in M$ and $\dec(\ppol, \mc{R}2, \d) \in M$, $\mc{R}1 \neq \mc{R}2$.
Based on Lemma \ref{p:ooa1}, $\bigoplus_{\ooa}(\mb{R}) \neq \idp$. 
Based on Lemma \ref{l:3}, there is a rule where $\dec(\ppol, \mc{R}1, \d)$ as the head and the body is true under $M$, i.e., $\dec(\ppol, \mc{R}1, \d) \la \val(\mc{R}1, \d).$
Then we find that $M(\val(\mc{R}1, \d)) = \top$.
Therefore, $\val(\mc{R}1, \d) \in M$.
Based on Prop. \ref{p:rule}, $\semantics{\mc{R}1}(\mc{Q}) = \d$.
Based on Lemma \ref{l:3}, there is a rule where $\dec(\ppol, \mc{R}2, \d)$ as the head and the body is true under $M$, $\mc{R}1 \neq \mc{R}2$.
There is only one rule in $\Prog$ where $\dec(\ppol, \mc{R}2, \d)$ as the head, i.e., $\dec(\ppol, \mc{R}2, \d) \la \val(\mc{R}2, \d).$, $\mc{R}1 \neq \mc{R}2$.
Then we find that $M(\val(\mc{R}2, \d)) = \top$, $\mc{R}1 \neq \mc{R}2$.
Therefore, $\val(\mc{R}2, \d) \in M$, $\mc{R}1 \neq \mc{R}2$.
Based on Prop. \ref{p:rule}, $\semantics{\mc{R}2}(\mc{Q}) = \d$.
Based on  \eqref{eq:ooa}, $\forall i: \semantics{\mc{R}_i}(\mc{Q}) \neq (\p \textrm{ or } \ip \textrm{ or } \idp)$ since $\bigoplus_{\ooa}(\mb{R}) \neq \idp$.
Therefore, based on  \eqref{eq:ooa}, $\bigoplus_{\ooa}(\mb{R}) = \id$.
\qed
\end{enumerate}
\end{proof}

\begin{lemma}
\label{p:ooa3}
Let $\pprog = \pprog_{\mc{Q}} \cup \pprog_{\ooa} \cup \ppolicy$  be a program obtained by merging \Request transformation program $\pprog_{\mc{Q}}$, only-one-applicable combining algorithm transformation program $\pprog_{\ooa}$ and \Policy\ \ppol\ transformation program with its components \ppolicy. Let $M$ be an answer set of \pprog.
Then, \niceiffcomb{\ooa}{\ip}
where $\mb{R} = \seq{\mc{R}_1(\mc{Q}), \ldots, \mc{R}_n(\mc{Q})}$ be a sequence of policy value where each $\mc{R}_i$ is a \Rule in the sequence inside  \Policy\ \ppol.
\end{lemma}
\begin{proof}
\textit{Note:} The proof is similar to the proof of Lemma \ref{p:ooa2}.
\end{proof}

\begin{lemma}
\label{p:ooa4}
Let $\pprog = \pprog_{\mc{Q}} \cup \pprog_{\ooa} \cup \ppolicy$  be a program obtained by merging \Request transformation program $\pprog_{\mc{Q}}$, only-one-applicable combining algorithm transformation program $\pprog_{\ooa}$ and \Policy\ \ppol\ transformation program with its components \ppolicy. Let $M$ be an answer set of \pprog.
Then, \niceiffcomb{\ooa}{\p}
where $\mb{R} = \seq{\mc{R}_1(\mc{Q}), \ldots, \mc{R}_n(\mc{Q})}$ be a sequence of policy value where each $\mc{R}_i$ is a \Rule in the sequence inside  \Policy\ \ppol.
\end{lemma}
\begin{proof}
$(\Rightarrow)$
Suppose that $\bigoplus_{\ooa}(\mb{R}) = \p$ holds.
Then, as defined in  \eqref{eq:ooa} we have that $\exists  i: \semantics{\mc{R}_i}(\mc{Q}) = \d$ and $\forall j: j \neq i$, $\semantics{\mc{R}_j}(\mc{Q}) = \na$.
Based on Prop. \ref{p:rule}, $\exists i: \val(\mc{R}_i, \p) \in M$.
Based on Lemma \ref{p:ooa1}, $\algo(\ooa, \ooa, \idp) \not \in M$ since if it is in $M$, there exists a \Rule $\mc{R}$ in \Policy\ \ppol sequence such that $\semantics{\mc{R}}(\mc{Q}) = \idp.$
Based on Lemma \ref{p:ooa2}, $\algo(\ooa, \ooa, \id) \not \in M$ since if it is in $M$, there exists a \Rule $\mc{R}$ in \Policy\ \ppol sequence such that $\semantics{\mc{R}}(\mc{Q}) = \id$ or there are at least two \Rule elements $\mc{R}1$ and $\mc{R}2$, $\mc{R}1 \neq \mc{R}2$, such that $\semantics{\mc{R}1}(\mc{Q}) = \semantics{\mc{R}2}(\mc{Q}) = \d$.
Based on Lemma \ref{p:ooa3}, $\algo(\ooa, \ooa, \ip) \not \in M$ since if it is in $M$, there exists a \Rule $\mc{R}$ in \Policy\ \ppol sequence such that $\semantics{\mc{R}}(\mc{Q}) = \ip$ or there are at least two \Rule elements $\mc{R}1$ and $\mc{R}2$, $\mc{R}1 \neq \mc{R}2$, such that $\semantics{\mc{R}1}(\mc{Q}) = \semantics{\mc{R}2}(\mc{Q}) = \p$.
Based on the \Policy transformation there is a rule $\dec(\ppol, \mc{R}_i, \p) \la \val(\mc{R}_i, \p).$ 
Therefore $\dec(\ppol, \mc{R}_i, \p) \in M$  and $\dec(\ppol, \mc{R}_j, \na) \in M$ since $M$ is the minimal model of $\Prog$.
Thus, $\algo(\ooa, \ppol, \p) \in M$ since $M$ is the minimal model of $\Prog$. 
\\
$(\Leftarrow)$
Suppose that $\algo(\ooa, \ppol, \p) \in M$. 
Based on Lemma \ref{l:3}, there is a rule where $\algo(\ooa, \ppol, \d)$ as the head and the body is true under $M$. 
There are rules in $\Prog$ where $\algo(\ooa, \ppol, \d)$ as the head, i.e.,
$\algo(\ooa, \ppol, \id) \la$ $\mynot \algo(\ooa, \ppol, \idp),$\\ $\mynot \algo(\ooa, \ppol, \id), \mynot \algo(\ooa, \ppol, \ip),  \dec(\ppol, \mc{R}, \p).$ \\
Then, we find that $M(\mynot \algo(\ooa, \ppol, \idp) \wedge  \mynot \algo(\ooa, \ppol, \id) \wedge \mynot \algo(\ooa, \ppol, \ip)$ $\wedge \dec(\ppol, \mc{R}, \id)) = \top$.
Therefore, $\algo(\ooa, \ppol, \idp) \not \in M$, $\algo(\ooa, \ppol, \id) \not \in M$, $\algo(\ooa, \ppol, \ip) \not \in M$  and $\dec(\ppol, \mc{R}, \id) \in M$.
By Lemma \ref{p:ooa1}, $\bigoplus_{\ooa}(\mb{R}) \neq \idp$. 
By Lemma \ref{p:ooa2}, $\bigoplus_{\ooa}(\mb{R}) \neq \id$. 
By Lemma \ref{p:ooa3}, $\bigoplus_{\ooa}(\mb{R}) \neq \ip$. 
By Lemma \ref{l:3}, there is a rule where $\dec(\ppol, \mc{R}, \p)$ as the head and the body is true under $M$, i.e., $\dec(\ppol, \mc{R}, \p) \la \val(\mc{R}, \p).$
Then we find that $M(\val(\mc{R}, \id)) = \top$.
Therefore, $\val(\mc{R}, \p) \in M$.
Based on Prop. \ref{p:rule}, $\semantics{\mc{R}}(\mc{Q}) = \p$.
Based on  \eqref{eq:ooa}, $\forall i: \semantics{\mc{R}_i}(\mc{Q}) \neq (\p \textrm{ or } \ip \textrm{ or } \idp)$ since $\bigoplus_{\ooa}(\mb{R}) \neq \idp$.
Based on  \eqref{eq:ooa}, $\forall i: \semantics{\mc{R}_i}(\mc{Q}) \neq \id$  or there are no two rules which $\semantics{\mc{R}1}(\mc{Q}) = \semantics{\mc{R}2}(\mc{Q}) = \d$ since $\bigoplus_{\ooa}(\mb{R}) \neq \id$.
Based on  \eqref{eq:ooa}, $\forall i: \semantics{\mc{R}_i}(\mc{Q}) \neq \ip \textrm{ or } \p$  since $\bigoplus_{\ooa}(\mb{R}) \neq \ip$.
Therefore, based on  \eqref{eq:ooa}, $\bigoplus_{\ooa}(\mb{R}) = \p$.
\qed
\end{proof}

\begin{lemma}
\label{p:ooa5}
Let $\pprog = \pprog_{\mc{Q}} \cup \pprog_{\ooa} \cup \ppolicy$  be a program obtained by merging \Request transformation program $\pprog_{\mc{Q}}$, only-one-applicable combining algorithm transformation program $\pprog_{\ooa}$ and \Policy\ \ppol\ transformation program with its components \ppolicy. Let $M$ be an answer set of \pprog.
Then, \niceiffcomb{\ooa}{\d}
where $\mb{R} = \seq{\mc{R}_1(\mc{Q}), \ldots, \mc{R}_n(\mc{Q})}$ be a sequence of policy value where each $\mc{R}_i$ is a \Rule in the sequence inside  \Policy\ \ppol.
\end{lemma}
\begin{proof}
\textit{Note:} The proof is similar to the proof of Lemma \ref{p:ooa4}.
\end{proof}

\begin{lemma}
\label{p:ooa6}
Let $\pprog = \pprog_{\mc{Q}} \cup \pprog_{\ooa} \cup \ppolicy$  be a program obtained by merging \Request transformation program $\pprog_{\mc{Q}}$, only-one-applicable combining algorithm transformation program $\pprog_{\ooa}$ and \Policy\ \ppol\ transformation program with its components \ppolicy. Let $M$ be an answer set of \pprog.
Then, \niceiffcomb{\ooa}{\na}
where $\mb{R} = \seq{\mc{R}_1(\mc{Q}), \ldots, \mc{R}_n(\mc{Q})}$ be a sequence of policy value where each $\mc{R}_i$ is a \Rule in the sequence inside  \Policy\ \ppol.
\end{lemma}
\begin{proof}
$(\Rightarrow)$
Suppose that $\bigoplus_{\ooa}(\mb{R}) = \na$ holds. 
Then, as defined in  \eqref{eq:po} we have that
$\bigoplus_{\ooa}(\mb{R}) \neq \idp$, 
$\bigoplus_{\ooa}(\mb{R}) \neq \id$,
$\bigoplus_{\ooa}(\mb{R}) \neq \ip$,
$\bigoplus_{\ooa}(\mb{R}) \neq \d$, and
$\bigoplus_{\ooa}(\mb{R}) \neq \p$.
By Lemma \ref{p:ooa1}, $\algo(\ooa, \ppol, \idp) \not \in M$.
By Lemma \ref{p:ooa2}, $\algo(\ooa, \ppol, \id) \not \in M$.
By Lemma \ref{p:ooa3}, $\algo(\ooa, \ppol, \ip) \not \in M$.
By Lemma \ref{p:ooa4}, $\algo(\ooa, \ppol, \p) \not \in M$.
By Lemma \ref{p:ooa5}, $\algo(\ooa, \ppol, \d) \not \in M$.
Thus, $M(\mynot \algo(\ooa, \ppol, \idp) \wedge \mynot \algo(\ooa, \ppol, \id) \wedge \mynot \algo(\ooa, \ppol, \ip) \wedge \mynot \algo(\ooa, \ppol, \d) \wedge \mynot \algo(\ooa, \ppol, \p)) = \top$
Therefore, $\algo(\ooa, \ppol, \na) \in M$ since $M$ is the minimal model of $\Prog$. 
\\
$(\Leftarrow)$
Suppose that $\algo(\ooa, \ppol, \na) \in M$.
Based on Lemma \ref{l:3}, there is a rule where $\algo(\ooa, \ppol, \na)$ as the head and the body is true under $M$. 
There is only a rule in $\Prog$ where $\algo(\ooa, \ppol, \na)$ as the head, i.e., 
$\algo(\ooa, \ppol, \na) \la \mynot \algo(\ooa, \ppol, \idp),$ $\mynot \algo(\ooa, \ppol, \id),  \mynot \algo(\ooa, \ppol, \ip), \mynot \algo(\ooa, \ppol, \d),$ $\mynot \algo(\ooa, \ppol, \p).$
Then we find that $M(\mynot \algo(\ooa, \ppol, \idp) \wedge \mynot \algo(\ooa, \ppol, \id) \wedge  \mynot \algo(\ooa, \ppol, \ip)$ $ \wedge  \mynot \algo(\ooa, \ppol, \d) \wedge \mynot \algo(\ooa, \ppol, \p)) = \top$. 
Therefore, $\algo(\ooa, \ppol, \idp) \not \in M$, $\algo(\ooa, \ppol, \id) \not \in M$, $\algo(\ooa, \ppol, \ip) \not \in M$, $\algo(\ooa, \ppol, \d) \not \in M$   and $\algo(\ooa, \ppol, \p)$ $\not \in M$.
By Lemma \ref{p:po1}, $\bigoplus_{\ooa}(\mb{R}) \neq \idp$. 
By Lemma \ref{p:po2}, $\bigoplus_{\ooa}(\mb{R}) \neq \idp$.
By Lemma \ref{p:po3}, $\bigoplus_{\ooa}(\mb{R}) \neq \ip$.
By Lemma \ref{p:po4}, $\bigoplus_{\ooa}(\mb{R}) \neq \ip$.
By Lemma \ref{p:po5}, $\bigoplus_{\ooa}(\mb{R}) \neq \ip$.
Therefore, based on  \eqref{eq:po}, $\bigoplus_{\ooa}(\mb{R}) = \na$
\qed
\end{proof}

\begin{proposition}
\label{p:ooa}
Let $\pprog = \pprog_{\mc{Q}} \cup \pprog_{\ooa} \cup \ppolicy$  be a program obtained by merging \Request transformation program $\pprog_{\mc{Q}}$, only-one-applicable combining algorithm transformation program $\pprog_{\ooa}$ and \Policy\ \ppol\ transformation program with its components \ppolicy. Let $M$ be an answer set of \pprog.
Then, \niceiffcomb{\ooa}{V}
where $\mb{R} = \seq{\mc{R}_1(\mc{Q}), \ldots, \mc{R}_n(\mc{Q})}$ be a sequence of policy value where each $\mc{R}_i$ is a \Rule in the sequence inside  \Policy\ \ppol.
\end{proposition}
\begin{proof}
It follows from Lemma \ref{p:ooa1}, Lemma \ref{p:ooa2}, Lemma \ref{p:ooa3}, Lemma \ref{p:ooa4}, Lemma \ref{p:ooa5} and Lemma \ref{p:ooa6} since the value of $V$ only has six possibilities, i.e., $\Set{\p, \d, \ip, \id, \idp, \na}$.
\qed
\end{proof}

\vspace{5pt}\noindent\textbf{Evaluation to Combining Algorithms.}
\begin{proposition}
\label{p:comb}
Let $\pprog = \pprog_{\mc{Q}} \cup \pprog_{\comb} \cup \ppolicy$  be a program obtained by merging \Request transformation program $\pprog_{\mc{Q}}$,  combining algorithm transformation program $\pprog_{\comb}$ and \Policy\ \ppol\ transformation program with its components \ppolicy. Let $M$ be an answer set of \pprog.
Then,  \niceiffcomb{\comb}{V}
where $\mb{R} = \seq{\mc{R}_1(\mc{Q}), \ldots, \mc{R}_n(\mc{Q})}$ be a sequence of policy value where each $\mc{R}_i$ is a \Rule in the sequence inside  \Policy\ \ppol.
\end{proposition}

\begin{proof}
It follows from Prop. \ref{p:po}, Prop. \ref{p:fa} and Prop. \ref{p:ooa}.
\qed
\end{proof}

\vspace{5pt}\noindent\textbf{\Policy Evaluation.}
\begin{lemma}
\label{p:policy1}
Let $\Prog = \Prog_{\mc{Q}} \cup \ppolicy$ be a program obtained  by merging \Request transformation program $\Prog_{\mc{Q}}$ \Policy\ \ppol transformation program and its components \ppolicy. Let $M$ be an answer set of $\Prog$. 
Then, \[ \semantics{\ppol}(\mc{Q}) = \id \textrm{ iff } \val(\ppol, \id) \in M. \]
\end{lemma}
\begin{proof}
$(\Rightarrow)$
Suppose that $\semantics{\ppol}(\mc{Q}) = \id$ holds. 
Then, as defined in  \eqref{eq:policy} we have that 
\begin{enumerate}
\item
$\semantics{\mc{T}}(\mc{Q}) = \idt$ and $\bigoplus_{\comb}(\mb{R}) = \d$.
Based on Prop. \ref{p:target} and Prop. \ref{p:comb}, $\val(\mc{T}, \idt) \in M$ and $\algo(\comb, \ppol, \d) \in M$.
Thus, $M(\val(\mc{T}, \idt) \wedge \algo(\comb, \ppol, \d)) = \top$.
Hence, by Lemma \ref{l:1}, $\val(\ppol, \id) \in M$.

\item $\semantics{\mc{T}}(\mc{Q}) = \idt$ and $\bigoplus_{\comb}(\mb{R}) = \id$ and $\forall i: \semantics{\mc{R}_i}(\mc{Q}) \neq \na$.
Based on Prop. \ref{p:target} and Prop. \ref{p:comb}, $\val(\mc{T}, \idt) \in M$ and $\algo(\comb, \ppol, \id) \in M$.
Based on Prop. \ref{p:rule}, $\exists i: \val(\mc{R}_i, V) \in M,  V \neq \na$. 
Therefore, we have $\dec(\ppol, \mc{R}_i, V) \in M$ since $M$ is the minimal model of $\Prog$
Thus, $M(\val(\mc{T}, \idt) \wedge \algo(\comb, \ppol, \id) \wedge \dec(\ppol, \mc{R}_i, V) \wedge (V \neq \na)) = \top$.
Hence, by Lemma \ref{l:1}, $\val(\ppol, \id) \in M$.

\item $\semantics{\mc{T}}(\mc{Q}) = \m$ and $\bigoplus_{\comb}(\mb{R}) = \id$ and $\forall i: \semantics{\mc{R}_i}(\mc{Q}) \neq \na$.
Based on Prop. \ref{p:target} and Prop. \ref{p:comb}, $\val(\mc{T}, \m) \in M$ and $\algo(\comb, \ppol, \id) \in M$.
Based on Prop. \ref{p:rule}, $\exists i: \val(\mc{R}_i, V) \in M,  V \neq \na$. 
Therefore, we have $\dec(\ppol, \mc{R}_i, V) \in M$ since $M$ is the minimal model of $\Prog$
Thus, $M(\val(\mc{T}, \m) \wedge \algo(\comb, \ppol, \id) \wedge \dec(\ppol, \mc{R}_i, V) \wedge (V \neq \na)) = \top$.
Therefore, $\val(\ppol, \id) \in M$ since $M$ is the minimal model of $\Prog$.
\end{enumerate}
$(\Leftarrow)$
Suppose that $\val(\ppol, \id) \in M$.
Based on Lemma \ref{l:3}, there is a clause where $\val(\ppol, \id)$ as the head and the body is true under $M$. 
There are rules in $\Prog$ where $\val(\ppol, \id)$ as the head, i.e., 
\begin{enumerate}
\item
$\val(\ppol, \id) \la \val(\mc{T}, \idt), \algo(\comb, \ppol, \d).$ \\
Then,  $M( \val(\mc{T}, \idt) \wedge \algo(\comb, \ppol, \d)) = \top$. 
Therefore, $\val(\mc{T}, \idt)  \in M$ and $\algo(\comb, \ppol, \d)  \in M $.
Based on Prop. \ref{p:target} and Prop. \ref{p:comb}, $\semantics{\mc{T}}(\mc{Q}) = \idt$ and $\bigoplus_{\comb}(\mb{R}) = \d$.
Therefore, based on   \eqref{eq:policy}, $\semantics{\ppol}(\mc{Q}) = \id$.  
\item 
$\val(\ppol, \id) \la \val(\mc{T}, \idt), \dec(\ppol, \mc{R}, V), V \neq \na, \algo(\comb, \ppol, \id), \id \neq \d.$ \\
Then, $M( \val(\mc{T}, \idt) \wedge \dec(\ppol, \mc{R}, V) \wedge  (V \neq \na) \wedge \algo(\comb, \ppol, \id) \wedge (\id \neq \d)) = \top$. 
Therefore, $\val(\mc{T}, \idt)  \in M$, $\dec(\ppol, \mc{R}, V) \in M$, $V \neq \na$ and $\algo(\comb, \ppol, \id)  \in M $.
Based on Prop. \ref{p:target} and Prop. \ref{p:comb}, $\semantics{\mc{T}}(\mc{Q}) = \idt$ and $\bigoplus_{\comb}(\mb{R}) = \id$.
Based on Lemma \ref{l:3}, there is a clause where $\dec(\ppol, \mc{R}, V)$ as the head and the body is true under $M$.
There is a rule $\Prog$ where $\dec(\ppol, \mc{R}, V)$ as the head, i.e., $\dec(\ppol, \mc{R}, V) \la \val(\mc{R}, V).$
Then we find that $M(\val(\mc{R}, V)) = \top$. 
Thus, $\val(\mc{R}, V) \in M$.
Based on Prop. \ref{p:rule}, $\semantics{\mc{R}}(\mc{Q}) = V$, $V \neq \na$.
Therefore, based on   \eqref{eq:policy}, $\semantics{\ppol}(\mc{Q}) = \id$. 
\item
$\val(\ppol, \id) \la \val(\mc{T}, \m), \dec(\ppol, \mc{R}, V), V \neq \na, \algo(\comb, \ppol, \id).$ \\
Then we find that $M( \val(\mc{T}, \m) \wedge \dec(\ppol, \mc{R}, V) \wedge  (V \neq \na) \wedge \algo(\comb, \ppol, \id)) = \top$. 
Therefore, $\val(\mc{T}, \m)  \in M$, $\dec(\ppol, \mc{R}, V) \in M$, $V \neq \na$ and $\algo(\comb, \ppol, \id)  \in M $.
Based on Prop. \ref{p:target} and Prop. \ref{p:comb}, $\semantics{\mc{T}}(\mc{Q}) = \m$ and $\bigoplus_{\comb}(\mb{R}) = \id$.
Based on Lemma \ref{l:3}, there is a clause where $\dec(\ppol, \mc{R}, V)$ as the head and the body is true under $M$.
There is a rule $\Prog$ where $\dec(\ppol, \mc{R}, V)$ as the head, i.e., $\dec(\ppol, \mc{R}, V) \la \val(\mc{R}, V).$
Then we find that $M(\val(\mc{R}, V)) = \top$. 
Thus, $\val(\mc{R}, V) \in M$.
Based on Prop. \ref{p:rule}, $\semantics{\mc{R}}(\mc{Q}) = V$, $V \neq \na$.
Therefore, based on   \eqref{eq:policy}, $\semantics{\ppol}(\mc{Q}) = \id$.   
\qed
\end{enumerate}
\end{proof}

\begin{lemma}
\label{p:policy2}
Let $\Prog = \Prog_{\mc{Q}} \cup \ppolicy$ be a program obtained  by merging \Request transformation program $\Prog_{\mc{Q}}$ \Policy\ \ppol transformation program and its components \ppolicy. Let $M$ be an answer set of $\Prog$. Then, \niceiff{\mcP}{\ip}{.}
\end{lemma}
\begin{proof}
\textit{Note:} The proof is similar with the proof in Lemma \ref{p:policy1}.
\end{proof}

\begin{lemma}
\label{p:policy3}
Let $\Prog = \Prog_{\mc{Q}} \cup \ppolicy$ be a program obtained  by merging \Request transformation program $\Prog_{\mc{Q}}$ \Policy\ \ppol transformation program and its components \ppolicy. Let $M$ be an answer set of $\Prog$. Then, \niceiff{\mcP}{\na}{.}
\end{lemma}
\begin{proof}
$(\Rightarrow)$
Suppose that $\semantics{\ppol}(\mc{Q}) = \na$ holds. 
Then, as defined in  \eqref{eq:policy} we have that
\begin{enumerate}
\item 
$\semantics{\mc{T}}(\mc{Q}) = \nm$.
Based on Prop. \ref{p:target}, $\val(\mc{T}, \nm) \in M$. 
Thus, $M(\val(\mc{T}, \nm)) = \top$.
Therefore, $\val(\ppol, \na) \in M$ since $M$ is the minimal model of $\Prog$.
\item 
$\forall i: \semantics{\mc{R}_i}(\mc{Q}) = \na$.
Based on Prop. \ref{p:rule}, $\forall i: \val(\mc{R}_i, \na) \in M$. 
Thus, $M(\val(\mc{R}_1, \na) \wedge \ldots \wedge \val(\mc{R}_n, \na)) = \top$.
Therefore, $\val(\ppol, \na) \in M$ since $M$ is the minimal model of $\Prog$.
\end{enumerate}
$(\Leftarrow)$
Suppose that $\val(\ppol, \id) \in M$.
Based on Lemma \ref{l:3}, there is a clause where $\val(\ppol, \ip)$ as the head and the body is true under $M$. 
There are rules in $\Prog$ where $\val(\ppol, \ip)$ as the head, i.e., 
\begin{enumerate}
\item
$\val(\ppol, \na) \la \val(\mc{T}, \nm).$
Then we find that $M( \val(\mc{T}, \nm)) = \top$. 
Therefore, $\val(\mc{T}, \nm)  \in M$.
Based on Prop. \ref{p:target}, $\semantics{\mc{T}}(\mc{Q}) = \na$.
Therefore, based on   \eqref{eq:policy}, $\semantics{\ppol}(\mc{Q}) = \na$.  
\item
$\val(\ppol, \na) \la \val(\mc{R}_1, \na), \ldots, \val(\mc{R}_n, \na).$
Then we find that $M( \val(\mc{R}_1, \na) \wedge \ldots \wedge \val(\mc{R}_n, \na)) = \top$. 
Therefore, $\forall i: \val(\mc{R}_i, \na)  \in M$.
Based on Prop. \ref{p:rule}, $\forall i: \semantics{\mc{R}_i}(\mc{Q}) = \na$.
Therefore, based on   \eqref{eq:policy}, $\semantics{\ppol}(\mc{Q}) = \na$.  
\qed
\end{enumerate}
\end{proof}

\begin{lemma}
\label{p:policy4}
Let $\Prog = \Prog_{\mc{Q}} \cup \ppolicy$ be a program obtained  by merging \Request transformation program $\Prog_{\mc{Q}}$ \Policy\ \ppol transformation program and its components \ppolicy. Let $M$ be an answer set of $\Prog$. Then, \niceiff{\mcP}{\idp}{.} 
\end{lemma}
\begin{proof}
$(\Rightarrow)$
Suppose that $\semantics{\ppol}(\mc{Q}) = \idp$ holds. 
Then, as defined in  \eqref{eq:policy} we have that 
\begin{enumerate}
\item $\semantics{\mc{T}}(\mc{Q}) = \idt$ and $\bigoplus_{\comb}(\mb{R}) = \idp$ and $\forall i: \semantics{\mc{R}_i}(\mc{Q}) \neq \na$.
Based on Prop. \ref{p:target} and Prop. \ref{p:comb}, $\val(\mc{T}, \idt) \in M$ and $\algo(\comb, \ppol, \idp) \in M$.
Based on Prop. \ref{p:rule}, $\exists i: \val(\mc{R}_i, V) \in M,  V \neq \na$. 
Therefore, we have $\dec(\ppol, \mc{R}_i, V) \in M$ since $M$ is the minimal model of $\Prog$
Thus, $M(\val(\mc{T}, \idt) \wedge \algo(\comb, \ppol, \idp) \wedge \dec(\ppol, \mc{R}_i, V) \wedge (V \neq \na)) = \top$.
Therefore, $\val(\ppol, \idp) \in M$ since $M$ is the minimal model of $\Prog$.

\item $\semantics{\mc{T}}(\mc{Q}) = \m$ and $\bigoplus_{\comb}(\mb{R}) = \idp$ and $\forall i: \semantics{\mc{R}_i}(\mc{Q}) \neq \na$.
Based on Prop. \ref{p:target} and Prop. \ref{p:comb}, $\val(\mc{T}, \m) \in M$ and $\algo(\comb, \ppol, \idp) \in M$.
Based on Prop. \ref{p:rule}, $\exists i: \val(\mc{R}_i, V) \in M,  V \neq \na$. 
Therefore, we have $\dec(\ppol, \mc{R}_i, V) \in M$ since $M$ is the minimal model of $\Prog$
Thus, $M(\val(\mc{T}, \m) \wedge \algo(\comb, \ppol, \idp) \wedge \dec(\ppol, \mc{R}_i, V) \wedge (V \neq \na)) = \top$.
Therefore, $\val(\ppol, \idp) \in M$ since $M$ is the minimal model of $\Prog$.
\end{enumerate}
$(\Leftarrow)$
Suppose that $\val(\ppol, \idp) \in M$.
Based on Lemma \ref{l:3}, there is a clause where $\val(\ppol, \idp)$ as the head and the body is true under $M$. 
There are rules in $\Prog$ where $\val(\ppol, \idp)$ as the head, i.e., 
\begin{enumerate}
\item
$\val(\ppol, \idp) \la \val(\mc{T}, \idt), \dec(\ppol, \mc{R}, V), V \neq \na, \algo(\comb, \ppol, \idp), \idp \neq \d.$
Then we find that $M( \val(\mc{T}, \idt) \wedge \dec(\ppol, \mc{R}, V) \wedge  (V \neq \na) \wedge \algo(\comb, \ppol, \idp) \wedge (\idp \neq \d)) = \top$. 
Therefore, $\val(\mc{T}, \idt)  \in M$, $\dec(\ppol, \mc{R}, V) \in M$, $V \neq \na$ and $\algo(\comb, \ppol, \idp)  \in M $.
Based on Prop. \ref{p:target} and Prop. \ref{p:comb}, $\semantics{\mc{T}}(\mc{Q}) = \idt$ and $\bigoplus_{\comb}(\mb{R}) = \idp$.
Based on Lemma \ref{l:3}, there is a clause where $\dec(\ppol, \mc{R}, V)$ as the head and the body is true under $M$.
There is a rule $\Prog$ where $\dec(\ppol, \mc{R}, V)$ as the head, i.e., $\dec(\ppol, \mc{R}, V) \la \val(\mc{R}, V).$
Then we find that $M(\val(\mc{R}, V)) = \top$. 
Thus, $\val(\mc{R}, V) \in M$.
Based on Prop. \ref{p:rule}, $\semantics{\mc{R}}(\mc{Q}) = V$, $V \neq \na$.
Therefore, based on   \eqref{eq:policy}, $\semantics{\ppol}(\mc{Q}) = \idp$. 
\item
$\val(\ppol, \idp) \la \val(\mc{T}, \m), \dec(\ppol, \mc{R}, V), V \neq \na, \algo(\comb, \ppol, \idp).$
Then we find that $M( \val(\mc{T}, \m) \wedge \dec(\ppol, \mc{R}, V) \wedge  (V \neq \na) \wedge \algo(\comb, \ppol, \idp)) = \top$. 
Therefore, $\val(\mc{T}, \m)  \in M$, $\dec(\ppol, \mc{R}, V) \in M$, $V \neq \na$ and $\algo(\comb, \ppol, \idp)  \in M $.
Based on Prop. \ref{p:target} and Prop. \ref{p:comb}, $\semantics{\mc{T}}(\mc{Q}) = \m$ and $\bigoplus_{\comb}(\mb{R}) = \idp$.
Based on Lemma \ref{l:3}, there is a clause where $\dec(\ppol, \mc{R}, V)$ as the head and the body is true under $M$.
There is a rule $\Prog$ where $\dec(\ppol, \mc{R}, V)$ as the head, i.e., $\dec(\ppol, \mc{R}, V) \la \val(\mc{R}, V).$
Then we find that $M(\val(\mc{R}, V)) = \top$. 
Thus, $\val(\mc{R}, V) \in M$.
Based on Prop. \ref{p:rule}, $\semantics{\mc{R}}(\mc{Q}) = V$, $V \neq \na$.
Therefore, based on   \eqref{eq:policy}, $\semantics{\ppol}(\mc{Q}) = \idp$.   
\qed
\end{enumerate}
\end{proof}

\begin{lemma}
\label{p:policy5}
Let $\Prog = \Prog_{\mc{Q}} \cup \ppolicy$ be a program obtained  by merging \Request transformation program $\Prog_{\mc{Q}}$ \Policy\ \ppol transformation program and its components \ppolicy. Let $M$ be an answer set of $\Prog$. Then, \niceiff{\mcP}{\p}{.} 
\end{lemma}
\begin{proof}
$(\Rightarrow)$
Suppose that $\semantics{\ppol}(\mc{Q}) = \p$ holds. 
Then, as defined in  \eqref{eq:policy} we have that 
$\semantics{\mc{T}}(\mc{Q}) = \m$ and $\bigoplus_{\comb}(\mb{R}) = \p$ and $\forall i: \semantics{\mc{R}_i}(\mc{Q}) \neq \na$.
Based on Prop. \ref{p:target} and Prop. \ref{p:comb}, $\val(\mc{T}, \m) \in M$ and $\algo(\comb, \ppol, \p) \in M$.
Based on Prop. \ref{p:rule}, $\exists i: \val(\mc{R}_i, V) \in M,  V \neq \na$. 
Therefore, we have $\dec(\ppol, \mc{R}_i, V) \in M$ since $M$ is the minimal model of $\Prog$
Thus, $M(\val(\mc{T}, \m) \wedge \algo(\comb, \ppol, \p) \wedge \dec(\ppol, \mc{R}_i, V) \wedge (V \neq \na)) = \top$.
Therefore, $\val(\ppol, \p) \in M$ since $M$ is the minimal model of $\Prog$.
\\
$(\Leftarrow)$
Suppose that $\val(\ppol, \id) \in M$.
Based on Lemma \ref{l:3}, there is a clause where $\val(\ppol, \id)$ as the head and the body is true under $M$. 
There are rules in $\Prog$ where $\val(\ppol, \id)$ as the head, i.e., 
$\val(\ppol, \p) \la \val(\mc{T}, \m), \dec(\ppol, \mc{R}, V), V \neq \na, \algo(\comb, \ppol, \p).$
Then we find that $M( \val(\mc{T}, \m) \wedge \dec(\ppol, \mc{R}, V) \wedge  (V \neq \na) \wedge \algo(\comb, \ppol, \p)) = \top$. 
Therefore, $\val(\mc{T}, \m)  \in M$, $\dec(\ppol, \mc{R}, V) \in M$, $V \neq \na$ and $\algo(\comb, \ppol, \p)  \in M $.
Based on Prop. \ref{p:target} and Prop. \ref{p:comb}, $\semantics{\mc{T}}(\mc{Q}) = \m$ and $\bigoplus_{\comb}(\mb{R}) = \p$.
Based on Lemma \ref{l:3}, there is a clause where $\dec(\ppol, \mc{R}, V)$ as the head and the body is true under $M$.
There is a rule $\Prog$ where $\dec(\ppol, \mc{R}, V)$ as the head, i.e., $\dec(\ppol, \mc{R}, V) \la \val(\mc{R}, V).$
Then we find that $M(\val(\mc{R}, V)) = \top$. 
Thus, $\val(\mc{R}, V) \in M$.
Based on Prop. \ref{p:rule}, $\semantics{\mc{R}}(\mc{Q}) = V$, $V \neq \na$.
Therefore, based on   \eqref{eq:policy}, $\semantics{\ppol}(\mc{Q}) = \p$.
\qed
\end{proof}

\begin{lemma}
\label{p:policy6}
Let $\Prog = \Prog_{\mc{Q}} \cup \ppolicy$ be a program obtained  by merging \Request transformation program $\Prog_{\mc{Q}}$ \Policy\ \ppol transformation program and its components \ppolicy. Let $M$ be an answer set of $\Prog$. Then, \niceiff{\mcP}{\d}{.} 
\end{lemma}
\begin{proof}
$(\Rightarrow)$
Suppose that $\semantics{\ppol}(\mc{Q}) = \d$ holds. 
Then, as defined in  \eqref{eq:policy} we have that 
$\semantics{\mc{T}}(\mc{Q}) = \m$ and $\bigoplus_{\comb}(\mb{R}) = \d$ and $\forall i: \semantics{\mc{R}_i}(\mc{Q}) \neq \na$.
Based on Prop. \ref{p:target} and Prop. \ref{p:comb}, $\val(\mc{T}, \m) \in M$ and $\algo(\comb, \ppol, \d) \in M$.
Based on Prop. \ref{p:rule}, $\exists i: \val(\mc{R}_i, V) \in M,  V \neq \na$. 
Therefore, we have $\dec(\ppol, \mc{R}_i, V) \in M$ since $M$ is the minimal model of $\Prog$
Thus, $M(\val(\mc{T}, \m) \wedge \algo(\comb, \ppol, \d) \wedge \dec(\ppol, \mc{R}_i, V) \wedge (V \neq \na)) = \top$.
Therefore, $\val(\ppol, \d) \in M$ since $M$ is the minimal model of $\Prog$.
\\
$(\Leftarrow)$
Suppose that $\val(\ppol, \id) \in M$.
Based on Lemma \ref{l:3}, there is a clause where $\val(\ppol, \id)$ as the head and the body is true under $M$. 
There are rules in $\Prog$ where $\val(\ppol, \id)$ as the head, i.e., 
$\val(\ppol, \d) \la \val(\mc{T}, \m), \dec(\ppol, \mc{R}, V), V \neq \na, \algo(\comb, \ppol, \d).$
Then we find that $M( \val(\mc{T}, \m) \wedge \dec(\ppol, \mc{R}, V) \wedge  (V \neq \na) \wedge \algo(\comb, \ppol, \d)) = \top$. 
Therefore, $\val(\mc{T}, \m)  \in M$, $\dec(\ppol, \mc{R}, V) \in M$, $V \neq \na$ and $\algo(\comb, \ppol, \d)  \in M $.
Based on Prop. \ref{p:target} and Prop. \ref{p:comb}, $\semantics{\mc{T}}(\mc{Q}) = \m$ and $\bigoplus_{\comb}(\mb{R}) = \d$.
Based on Lemma \ref{l:3}, there is a clause where $\dec(\ppol, \mc{R}, V)$ as the head and the body is true under $M$.
There is a rule $\Prog$ where $\dec(\ppol, \mc{R}, V)$ as the head, i.e., $\dec(\ppol, \mc{R}, V) \la \val(\mc{R}, V).$
Then we find that $M(\val(\mc{R}, V)) = \top$. 
Thus, $\val(\mc{R}, V) \in M$.
Based on Prop. \ref{p:rule}, $\semantics{\mc{R}}(\mc{Q}) = V$, $V \neq \na$.
Therefore, based on   \eqref{eq:policy}, $\semantics{\ppol}(\mc{Q}) = \d$.
\qed
\end{proof}

\begin{proposition}
\label{p:policy}
Let $\Prog = \Prog_{\mc{Q}} \cup \ppolicy$ be a program obtained  by merging \Request transformation program $\Prog_{\mc{Q}}$ \Policy\ \ppol transformation program and its components \ppolicy. Let $M$ be an answer set of $\Prog$. Then, \niceiff{\mcP}{V}{.} 
\end{proposition}

\begin{proof}
It follows from Lemma \ref{p:policy1}, Lemma \ref{p:policy2}, Lemma \ref{p:policy3}, Lemma \ref{p:policy4}, Lemma \ref{p:policy5} and Lemma \ref{p:policy6} since the value of $V$ only has six possibilities, i.e., $\Set{\p, \d, \ip, \id, \idp, \na}$.
\qed
\end{proof}

\vspace{5pt}\noindent\textbf{Evaluation to XACML Component.}
\begin{corollary}
\label{p:xacml}
Let $\pprog = \pprog_{\mc{Q}} \cup \PXACML$  be a program obtained by merging \Request transformation program $\pprog_{\mc{Q}}$ and all XACML components transformation programs \PXACML. Let $M$ be an answer set of \pprog.
Then, \niceiff{X}{V}{}
where $X$ is an XACML component. 
\end{corollary}

\begin{proof}
It follows from Prop. \ref{p:match}, Prop. \ref{p:allof}, Prop. \ref{p:anyof}, Prop. \ref{p:target}, Prop. \ref{p:condition}, Prop. \ref{p:rule} and Prop. \ref{p:policy}.
\qed
\end{proof}

\end{document}